\documentclass[a4paper,12pt]{article}
\usepackage{amsmath,amsthm,amsfonts,amssymb,natbib,graphicx,booktabs,float,caption,subcaption,multirow,setspace,afterpage,lscape,xcolor,setspace,etoolbox}
\usepackage[colorlinks=true, linkcolor=blue, citecolor=blue, urlcolor=blue]{hyperref}

\newtheorem{theorem}{Theorem}

\newtheorem{assumption}{Assumption}
\newtheorem*{assumption*}{Assumption}
\newtheorem{proposition}[theorem]{Proposition}

\newtheorem*{start*}{Starting Point}
\newtheorem*{proof*}{Proof}
\newtheorem*{note*}{Side Remark}

\AfterEndEnvironment{proposition}{\noindent\ignorespaces}
\AfterEndEnvironment{assumption}{\noindent\ignorespaces}
\AfterEndEnvironment{proof}{\noindent\ignorespaces}
\AfterEndEnvironment{corollary}{\noindent\ignorespaces}
\AfterEndEnvironment{remark}{\noindent\ignorespaces}
\DeclareMathOperator*{\Prb}{\text{Pr}^{b}}

\begin{document}

\title{Staggered Adoption DiD Designs with Misclassification and Anticipation\thanks{We thank seminar participants at Ca' Foscari Venice, KU Leuven, the Frankfurt Econometrics Workshop, the HKMetrics workshop in Mannheim, the 2025 IAAE meeting in Turin, the 2025 IPDC meeting in Montpellier, the 2025 VfS meeting in Cologne for helpful comments and discussions. We are also grateful to Emilie Berkhout and her co-authors  for kindly sharing data and code from their paper \textit{Using technology to prevent fraud in high stakes national school
examinations: Evidence from Indonesia} (Journal of Development Economics, 170(2024), 103307) with us. Financial support from the DFG (DFG Grant No. 463916029) is gratefully acknowledged.}}
\author{Clara Augustin\thanks{Faculty of Economics and Business, Goethe University Frankfurt, 60629
Frankfurt am Main, Germany. E-mail: \texttt{Augustin@econ.uni-frankfurt.de.}%
} \and Daniel Gutknecht\thanks{%
Faculty of Economics and Business, Goethe University Frankfurt, 60629
Frankfurt am Main, Germany. E-mail: \texttt{Gutknecht@wiwi.uni-frankfurt.de.}%
} \and Cenchen Liu\thanks{%
Faculty of Economics and Business, Goethe University Frankfurt, 60629
Frankfurt am Main, Germany. E-mail: \texttt{Liu@wiwi.uni-frankfurt.de.}%
}}
\date{\today }
\maketitle

\begin{abstract}
\noindent This paper examines the identification and estimation of treatment effects in staggered adoption designs -- a common extension of the canonical Difference-in-Differences (DiD) model to multiple groups and time-periods -- in the presence of (time varying)
misclassification of the treatment status as well as of anticipation. We demonstrate that standard estimators are biased
with respect to commonly used causal parameters of interest under such forms of misspecification. To address this issue, we provide modified estimators that recover the Average Treatment Effect of observed and true switching units, respectively. Additionally,  we suggest two moment based specification tests aimed at detecting Parallel Trends violations in pre-treatment periods as well as the timing and extent of misclassification and anticipation effects. We illustrate the proposed methods with an application to the effects of an anti-cheating policy on school mean test scores in high stakes national exams in Indonesia.



\end{abstract}

\onehalfspacing

\pagebreak

\section{Introduction}\label{sec:intro}

The Difference-in-Differences (DiD) methodology is arguably one of the most popular research designs in applied economics, offering a simple and intuitive framework to estimate causal effects in observational settings. It serves as a preferred tool for policy evaluations, labor market studies, and other areas where randomized controlled trials are impractical. Staggered adoption designs naturally extend the canonical two-by-two DiD framework to settings where multiple groups adopt treatments at different times. This structure enables researchers to exploit variation in treatment timing and group-level adoption, aggregate treatment effects across groups, and accommodate heterogeneity in treatment effects.\footnote{To illustrate the prevalence of staggered adoption designs, \citet{BLW2022} report that more than 54\% of all DiD studies published in the leading Finance and Accounting journals between 2000 and 2019 (744 studies in total) were based on staggered adoption designs.}

Despite the widespread use of staggered adoption designs in practice, common forms of misspecification, such as (time-varying) treatment status misclassification or anticipatory effects, have received relatively little attention, even though they can undermine the validity of estimation and inference procedures in this setting. In this paper, we address these challenges and examine the implications of time-varying misclassification and anticipation for the identification and estimation of commonly used estimands in staggered adoption designs. In a second step, we propose remedies by introducing modified estimators that can recover well-defined treatment parameters, such as the Average Treatment Effect (ATE) of units in the \textit{observed} switching-to-treatment cells, or of units in the \textit{actual} or \textit{true} switching-to-treatment cells. The former parameter measures the ATE for all treated units (switcher-to-treatment and already treated units) at the time when treatment is observed in the data, while the latter is the ATE of all switching-to-treatment units at the time when they \textit{actually} switch status, which may be unobserved in the data. Additionally, we propose a testing procedure aimed at detecting the timing and extent of misclassification or anticipation effects.

Misclassification of the treatment status arises when the recorded treatment status does not coincide with the actual treatment status. This can occur due to data limitations, such as when the exact timing of treatment is not observed \citep[e.g.,][]{Bindler2018}, or due to reporting and classification errors when treatment status can only be inferred indirectly \citep[see][]{GS2017}. Anticipation effects, on the other hand, occur when individuals or units are either exposed to treatment prior to the recorded treatment date or adjust their behavior in anticipation of future treatment, thereby contaminating pre-treatment periods \citep[see, e.g.,][]{B2022,BBS2022,BPRSS2024}. Both issues are particularly problematic in staggered adoption designs, where treatment assignment evolves over time and across groups, amplifying the risk of ``forbidden comparisons'' and leading to invalid inference, even when estimands appropriate for staggered adoption settings are targeted.

This paper makes several contributions to the literature. First, it characterizes the bias introduced by misclassification and anticipation effects in commonly used estimators across different settings. Specifically, under the restrictive assumption of homogeneous treatment effects, we show that the Two-Way Fixed Effects (TWFE) estimator, the most prominent estimator in standard two-by-two DiD settings, is biased. The bias arises because misspecification distorts the observed timing and incidence of treatment, leading to ``forbidden comparisons'' where treated units may be incorrectly classified or compared to units that have already been partially treated. We then extend our analysis to alternative DiD estimands designed to accommodate heterogeneous treatment effects in staggered adoption designs. In particular, we consider the ATE of observed switching-to-treatment cells as well as the ATE of actual switching-to-treatment cells.   For both causal parameters, we derive the structure and nature of the bias under staggered adoption, highlighting the limitations of existing methods in the presence of such forms of misspecification.

As a second contribution, we define alternative estimators that allow to recover the aforementioned causal parameters in staggered adoption designs. In particular, we introduce a bias-corrected version of the staggered adoption estimator proposed by \citet{CH2020}, which recovers the ATE of observed switching-to-treatment units. Based on knowledge of the timing and extent of misspecification, the estimator adjusts the comparisons conducted by the original estimator of \citet{CH2020} by using groups unaffected by misspecification. In contrast, for the ATE of actual switching-to-treatment cells, additional steps are required. Since the actual timing of treatment adoption is unobserved at the unit level, we exploit a homogeneity assumption on the misclassification (anticipation) probability, which restricts heterogeneity across units but not over time. This allows us to consistently estimate the fraction of actual switching-to-treatment units by extending the Parallel Trends (PT) assumption to the (potential) outcomes under treatment. In a subsequent step, we construct an estimator for the ATE of actual switching-to-treatment cells and derive its asymptotic properties.

As a third contribution, we propose two flexible specification tests based on moment equalities that allow to check for violations of pre-treatment PT and for the presence, extent, and timing of misclassification. In fact, testing for both forms of misspecification is required as our proposed corrections  rely on PT assumptions, and distinguishing between bias arising from misclassification and bias due to PT violations in pre-treatment periods is inherently challenging. To address this, we firstly propose a moment equality based test for pre-trends. This test compares switching groups with not-yet-switched comparison groups in periods known not to be affected by misclassification or anticipation. This test may serve as a more powerful alternative to  event-study type pre-trend tests as it uses all possible combinations of comparisons across the time periods considered and not just with respect to a specific reference period. The second test on the other hand is a  test for misclassification or anticipation comparing the outcomes of observed switcher groups to not-yet-switched groups using periods potentially affected by the former. We implement these tests as moment equality tests and derive their asymptotic properties. Finally, we show how both tests may be combined to construct a decision rule for practicioners.

To illustrate the finite sample performance of both the estimators and the tests, we carry out a small-scale Monte Carlo simulation. The findings suggest a good performance even in cases of misspecification because of misclassification or anticipation. We then apply these methods to assess the effects and the presence of anticipation in a study by \citet{BPRSS2024} on the effects of a new computer based testing system on school mean test scores in Indonesia between 2015 to 2019. The testing system, which was voluntary and intended to curb cheating behavior, was reported to have led schools to anticipate measures that limit cheating. Our estimates for the ATE of observed switching-to-treatment units confirm this and align with \citet{BPRSS2024}, while the estimates for the ATE of actual switching-to-treatment units are slightly larger in magnitude. Moreover, we estimate the fraction of anticipating schools to be  around 15\% in 2018, and up to 30\% in 2019.  Finally, the misspecification tests provide evidence for anticipation effects, though they also  suggest some violation of parallel pre-trends in periods just prior to that anticipation for at least one cohort.

By providing a framework for inference in staggered adoption designs under misspecification, this paper contributes to the growing literature on the methodological underpinnings of DiD. In particular, the paper builds on three strands of the econometric literature. First, it contributes to the growing body of research on DiD methods tailored to staggered adoption designs, including recent advancements that address the challenges of heterogeneous treatment effects and dynamic treatment timing \citep[e.g.,][]{CH2020,CSA2021,BJS2024}. Second, the paper intersects with a broader literature on the misclassification of treatment status, which is well known to compromise the validity of inference procedures \citep[e.g.,][]{L2007,U2018,DK2022}. Unlike other approaches in the DiD literature that rely on proxy variables \citep[e.g.,][]{BG2018} or exclusion restrictions \citep[][]{NN2025}, our framework does not require such external information. Third, the paper relates to recent contributions on pre-testing of Parallel Trends and its implications for causal inference depending on the test outcome \citep[e.g.,][]{Rambachan2023,CH2024}.

The paper is structured as follows. Section \ref{sec:setup} introduces the outcome framework together with the misclassification design and outlines the key assumptions. Section \ref{sec:BIASANALYSIS} analyzes the bias of commonly defined estimands under homogeneous and heterogeneous treatment effects. Section \ref{sec:IdentTrueSwitch} discusses modified estimators that allow to recover of the ATE of observed and true switching units. Section \ref{sec:testing} proposes a testing procedure to detect misspecification in the form of misclassification or anticipation as well as pre-trend violations. Section \ref{sec:MC} examines the finite-sample behavior of the tests and the proposed estimators. Finally, Section \ref{sec:empirical} presents the empirical application, and Section \ref{sec:conclusion} concludes. All proofs are contained in the Appendix.

\section{Setup}\label{sec:setup}

We build on the setup of \citet{CH2020}, adapting it to account for treatment status misclassification. Specifically, as in \citet{CH2020}, we assume that the observations can be divided into \(G\) groups observed over \(T\) periods, with each group-period combination indexed by \((g, t) \in \{1, \ldots, G\} \times \{1, \ldots, T\}\). For each \((g, t)\) pair, we observe \(N_{g,t}\) units indexed by \(i \in \{1, \ldots, N_{g,t}\}\). Accordingly, for each \(i\) within a given \((g,t)\) combination, we define the potential outcomes \((Y_{i,g,t}(0), Y_{i,g,t}(1))\) alongside the latent, potentially misclassified, binary treatment status \(D_{i,g,t}^{*}\), where \(D_{i,g,t}^{*} = 1\) if unit \(i\) from group \(g\) receives treatment in period \(t\), and \(D_{i,g,t}^{*} = 0\) otherwise. The observed (or recorded) treatment status is given by a separate binary indicator \(D_{i,g,t}\), which may differ from \(D_{i,g,t}^{*}\) with positive probability, as formalized in Assumptions \ref{A6} and \ref{A7} below.  Additionally, to explicitly address potential misclassification of treatment timing under the sharp design (see Assumption \ref{A2} below), we introduce a latent indicator \(S^{*}_{i,g,t}\), equal to one if unit \(i\) from group \(g\) actually switches to treatment at the beginning of period \(t\), and zero otherwise. Unlike the observed treatment indicator \(D_{i,g,t}\), which might incorrectly classify some units as switchers, \(S^{*}_{i,g,t}\) precisely identifies true switching events. This distinction is especially useful in cross-sectional settings, where each unit is observed only once, allowing us to define causal parameters based explicitly on actual rather than observed switchers (see SubSection \ref{ssec:trueswitch}).

The observation rule in our setting is given by:
\[
Y_{i,g, t}=Y_{i,g, t}(D_{i,g,t}^* )=D_{i,g,t}^* Y_{i,g, t}(1)+\left(1-D_{i,g,t}^*\right) Y_{i,g ,t}(0)
\]
We further define averages of observed treatment indicators at the group-period level:
\begin{align*}
D_{g, t}&:=\frac{1}{N_{g,t}} \sum_{i=1}^{N_{g,t}} D_{i, g, t},  \ 
\mathbf{{D}}_g:=\left({D}_{g, t}\right)_{t\in \{1, \ldots, T\}}, \
\mathbf{{D}}:=\left({D}_{g, t}\right)_{(g, t) \in\{1, \ldots, G\} \times\{1, \ldots, T\}},
\end{align*}
with analogous definitions other quantities such as $D_{g, t}^*$,  $Y_{g, t}$ and $Y_{g, t}(0)$. Since  \(D_{i,g,t}\neq D_{i,g,t}^{*}\) may occur with positive probability, the above  setup explicitly accommodates scenarios in which individuals experience the treated potential outcome \(Y_{i,g,t}(1)\) despite having an observed treatment status \(D_{i,g,t}=0\), or vice versa. We refer to these discrepancies as instances of misclassification.\footnote{As noted earlier, misclassification can arise when exact treatment timing is unobserved or when treatment status is indirectly inferred.} Furthermore, in Section \ref{sec:BIASANALYSIS}, we demonstrate that this notation readily captures anticipatory effects, where units experience \(Y_{i,g,t}(1)\) prior to the observed treatment.

We make the following assumptions:

\begin{assumption}[Balanced Panel of Groups] \label{A1}
For each $(g, t) \in\{1, \ldots, G\}
\times\{1, \ldots, T\}:$ $$N_{g,t}>0.$$\end{assumption}

\begin{assumption}[Sharp Design] \label{A2}
For each $(g, t) \in\{1, \ldots, G\} \times\{1, \ldots, T\}$ and $i \in\left\{1, \ldots, N_{g,t}\right\}:$ $$ D_{i, g, t}={D}_{g, t}.$$\end{assumption}

\begin{assumption}[Mean Independence] \label{A3}
For each $(g, t) \in\{1, \ldots, G\} \times\{2, \ldots, T\}$:
$$\mathrm{E}\left[Y_{g, t}(0) \mid \mathbf{D}\right]=\mathrm{E}\left[Y_{g, t}(0) \mid \mathbf{D}_g\right].$$\end{assumption}

\begin{assumption}[Strong Exogeneity] \label{A4}
For each $(g, t) \in\{1, \ldots, G\} \times\{2, \ldots, T\}$: $$\mathrm{E}\left[{Y}_{g, t}(0)-{Y}_{g, t-1}(0) \mid \mathbf{{D}}_g\right]=\mathrm{E}\left[{Y}_{g, t}(0)-{Y}_{g, t-1}(0)\right].$$\end{assumption}

\begin{assumption}[Parallel Trends] \label{A5}
For each  $(g,k, t) \in\{1, \ldots, G\} \times\{1, \ldots, G\} \times\{2, \ldots, T\}$: 
$$\mathrm{E}\left[{Y}_{g, t}(0)-{Y}_{g, t-1}(0)\right]=\mathrm{E}\left[{Y}_{k, t}(0)-{Y}_{k, t-1}(0)\right].$$
\end{assumption}

Similar to \citet{CH2020}, we replace the independent groups assumption with mean independence, since the presence of stable groups (as defined below) inherently violates independence across groups (specifically, observed treatment status cannot be independent across groups in this case). Moreover, our sharp design assumption (Assumption \ref{A2}) is stated with respect to the observed treatment status $D_{i,g,t}$, thus explicitly accommodating scenarios in which researchers mistakenly assume a sharp design based solely on observed data. Importantly, we allow for violations of the sharp design at the level of the latent treatment status $D_{i,g,t}^{*}$, recognizing that misclassification may occur within specific groups and periods. Furthermore, mean independence and strong exogeneity are defined with respect to the observed treatment status $\mathbf{D}_g$. Finally, similar to \citet{CH2020}, we maintain a common trends assumption as our core identification assumption throughout the paper.\footnote{See \citet{CH2020} for a  discussion of the plausibility of the common trends assumption in settings with heterogeneous treatment effects.}

 We provide sufficient conditions for our assumptions to hold:
 
\begin{proposition}[Characterizing Mean Independence and Strong Exogeneity with Observed Treatment Status]\label{PROPCHARACTERIZATION}
Suppose that mean independence and strong exogeneity hold with respect to the actual treatment status $\mathbf{D}^*$:
\begin{align}\label{EQ:PROP1a}
\mathrm{E}\left[Y_{g, t}(0) \mid \mathbf{D}^*\right]=\mathrm{E}\left[Y_{g, t}(0) \mid \mathbf{D}_g^*\right],	
\end{align}
for all $(g, t) \in\{1, \ldots, G\} \times\{1, \ldots, T\}$, and
\begin{align}\label{EQ:PROP1b}
\mathrm{E}\left[Y_{g, t}(0)-Y_{g, t-1}(0) \mid \mathbf{D}_g^*\right]=\mathrm{E}\left[Y_{g, t}(0)-Y_{g, t-1}(0)\right],	
\end{align}
for all $(g, t) \in\{1, \ldots, G\} \times\{2, \ldots, T\}$.  Then Assumptions \ref{A3} and \ref{A4}, stated with respect to the observed treatment status $\mathbf{D}$, hold if the following are satisfied for all $(g, t) \in\{1, \ldots, G\} \times\{2, \ldots, T\}$:
\begin{enumerate}
    \item Trend-nondifferential misclassification:
\begin{align}\label{EQ:PROP1c}
	\mathrm{E}\left[{Y}_{g, t}(0)-{Y}_{g, t-1}(0) \mid \mathbf{D}_g^*,\mathbf{D}_g\right]&=\mathrm{E}\left[{Y}_{g, t}(0)-{Y}_{g, t-1}(0)\mid \mathbf{D}_g^*\right], \\
\label{EQ:PROP1d}
	\mathrm{E}\left[{Y}_{g, t}(0)-{Y}_{g, t-1}(0) \mid \mathbf{D}^*,\mathbf{D}\right]&=\mathrm{E}\left[{Y}_{g, t}(0)-{Y}_{g, t-1}(0)\mid \mathbf{D}^*\right].
\end{align}
	\item  Conditionally independence:
\begin{align}\label{EQ:PROP1e}
	\mathbf{D}_g^* \perp \mathbf{D}_{-g}\mid \mathbf{D}_g,
\end{align}
where $\mathbf{D}_{-g}:=\left(D_{k, t}\right)_{(k, t) \in\{1,\ldots,g-1,g+1,\ldots, G\} \times\{1, \ldots, T\}}$.
\end{enumerate}  
\end{proposition} 

The proposition bears some resemblance to Lemma 1 of \citet{DK2022}. It shows that assumptions stated with respect to the true treatment status $\mathbf{D}^*$ can be translated to the observed treatment status $\mathbf{D}$ under two additional conditions. 
First, misclassification is \emph{trend-nondifferential in mean}, in the sense that conditional on the true treatment status, the observed treatment indicator is mean-independent of the untreated outcome trend. This condition is the natural trend analogue of mean nondifferential misclassification in levels \citep[see, e.g.,][]{L2007,DK2022}. 
Second, misclassification is \emph{conditionally independent across groups}, meaning that within a given group, the true treatment status is independent of the observed treatment paths of other groups once the group's own observed treatment status is conditioned on.

To further clarify the mean independence assumption, note that \eqref{EQ:PROP1c} is equivalent to
\[
\mathrm{E}\!\left[Y_{g,t}(0)-Y_{g,t-1}(0)\mid \mathbf{D}_g^*=d,\mathbf{D}_g=d'\right]
=
\mathrm{E}\!\left[Y_{g,t}(0)-Y_{g,t-1}(0)\mid \mathbf{D}_g^*=d,\mathbf{D}_g=d''\right],
\]
for all $d\in\{0,1\}$ and all $d',d''\in\{0,1\}^T$, provided the expectations are well defined.  
This condition states that, conditional on the true treatment status, untreated potential outcomes follow the same trend regardless of whether the observed treatment status is correctly classified or misclassified.  
Such \emph{mean trend-nondifferential} misclassification is analogous to the standard mean nondifferential misclassification assumption in the measurement error literature, and is typically plausible when misclassification arises at the data collection or recording stage. By contrast, it may fail in settings where treatment status is self-reported.  Finally, we emphasize that the conditions in Proposition \ref{PROPCHARACTERIZATION} are sufficient but not necessary for Assumptions \ref{A3} and \ref{A4} to hold. In what follows, we therefore work directly with Assumptions \ref{A3} and \ref{A4} to derive transparent bias expressions and corresponding estimators.

We also assume staggered adoption designs, where observed and true treatment represents ``absorbing states":
\begin{assumption}
	[Staggered Adoption Designs]\label{A6} For each  $g\in \{1,\cdots, G\}$ and $t\in \{2, \ldots, T\}$: 
$$ D_{g, t} \geq D_{g, t-1}.$$
\end{assumption}

Assumption \ref{A6} is imposed solely for simplifying the notation. For designs involving more general switching patterns, similar assumptions about misclassification can be made for different types of switchers  (``reverse switchers'') at the cost of more complex notation. 

\section{Bias Analysis}\label{sec:BIASANALYSIS}

\subsection{Homogenous Treatment Effects}

In the absence of misclassification, it is well known that the standard TWFE estimator is generally biased in staggered adoption designs when treatment effects are heterogeneous. A notable exception is the case of homogeneous treatment effects, where
\begin{equation}\label{HOMOGENEOUSTE}
Y_{i,g,t}(1)-Y_{i,g,t}(0)=:\Delta
\end{equation}
holds for all $i$, $g$, and $t$. Specifically, let $\widehat{\beta}_{fe}$ denote the TWFE estimator obtained from a regression of $Y_{i,g,t}$ on group and time fixed effects as well as the observed treatment status $D_{i,g,t}$.\footnote{Under the sharp design (Assumption \ref{A2}), we have $D_{i,g,t}=D_{g,t}$ for all $i\in\{1,\dots,N_{g,t}\}$ and $(g,t)$.} In the absence of misclassification and under Assumptions \ref{A1}--\ref{A5}, it follows from Theorem 1 of \citet{CH2020} that:
\[
\mathrm{E}\left[\widehat{\beta}_{fe}\right]=\mathrm{E}\left[\frac{\sum_{(g,t)}\frac{N_{g,t}}{N_{1}}\varepsilon_{g,t}D_{g,t}\mathrm{E}\left[\Delta\left\vert \mathbf{D}\right.\right]}{\sum_{(g,t):D_{g,t}=1}\frac{N_{g,t}}{N_{1}}D_{g,t}\varepsilon_{g,t}}\right]=\Delta,
\]
where $\varepsilon_{g,t}$ is the estimated residual from the unit-level regression of $D_{g,t}$ on group and period specific fixed effects. Unfortunately, this is no longer true in the presence of misclassification as the following result demonstrates:

\begin{proposition}[Two-Way Fixed Effects]\label{PROPTWFE}
Suppose Assumptions \ref{A1}-\ref{A5} and Equation (\ref{HOMOGENEOUSTE}) hold, as well as $\mathrm{E}\left[\vert\widehat{\beta}_{fe}\vert\right]<\infty$. Then, the TWFE estimand $\mathrm{E}\left[\widehat{\beta}_{fe}\right]$ is given by:
\[
\mathrm{E}\left[\widehat{\beta}_{fe}\right]=\Delta \mathrm{E}\left[\frac{\sum_{(g,t)}\frac{N_{g,t}}{N_{1}}\varepsilon_{g,t}\left(\frac{1}{N_{g,t}}\sum_{i=1}^{N_{g,t}}\mathrm{Pr}\left(D_{i,g,t}^{\ast}=1\mid \mathbf{D}\right)\right)}{\sum_{(g,t)}\frac{N_{g,t}}{N_{1}}\varepsilon_{g,t}D_{g,t}}\right]
\]
\end{proposition} 

The above result illustrates that even under the restrictive assumption of homogeneous treatment effects, misclassification generally introduces bias into the TWFE estimator. Whether this bias results in attenuation or amplification depends on the behavior of the misclassification probability $\mathrm{Pr}(D_{i,g,t}^{*}=1\mid \mathbf{D})$.

In the subsequent subsections, we extend this analysis to show that misclassification induces bias not only in the TWFE estimand under homogeneous effects, but also in more general settings with heterogeneous treatment effects, even when considering estimands specifically designed for the latter scenario.

\subsection{ATE of Observed Switchers}

To formalize the source of bias in what follows, we focus on a specific pattern of misclassification that arises in the observed untreated period(s) prior to treatment adoption. This form of measurement error is common in applied work when the exact timing of treatment is not directly observed due to data limitations. In such cases, treatment is often coded using proxy rules that may lag the true policy change.

For example, in a study of the causal effects of capital punishment reform on jury behavior, \citet{Bindler2018} define the “treatment year” as the first year with zero recorded death sentences, while acknowledging that the underlying legal change may have occurred earlier. In a staggered adoption setting, this type of coding can generate systematic misclassification of pre-treatment periods for treated groups, even when treatment status is correctly measured once adoption is observed. The following assumption formalizes this structure of misclassification around the time of observed treatment adoption.

\begin{assumption}
	[Patterns of Misclassification]\label{A7} For each $(i,g, t) \in\{1,\cdots,N_{g,t}\}\times \{1, \ldots, G\} \times\{1, \ldots, T\}$ and $(j,g', t') \in\{1,\cdots,N_{g',t'-1}\}\times \{1, \ldots, G\} \times\{2, \ldots, T\}$:
	\begin{enumerate}
	\item	\(\mathrm{Pr}\left(D_{i,g',t}^* = 1 \mid D_{g',t} = 1\right) = 1\).
	\item \(\mathrm{Pr}\left(D_{j,g',t'-1}^* = 0 \mid D_{g',t'} = 0,D_{g',t'-1} = 0\right) = 1\).
	\end{enumerate}
\end{assumption}

Assumption \ref{A7} requires that the unit-level observed treatment status is correctly classified in two situations:
 (i) periods in which a group is observed to receive treatment, and (ii) observed untreated periods that are followed by another observed untreated period. This assumption is a simplification made for tractability, and restricts misclassification to periods immediately preceding observed treatment. It can be relaxed in a straightforward manner to more general settings.\footnote{Note however that this formulation allows for misclassification to occur in the comparison group during the ``treatment year", which is particularly relevant when using ``not-yet-treated" groups as controls. For instance, suppose a group is observed to be untreated in periods $t-1$ and $t$ (${D}_{g, t-1}=0$, ${D}_{g, t}=0$), but receives treatment in period $t+1$ (${D}_{g, t+1}=1$). Then, misclassification may arise in period $t$.}  In the context of \citet{Bindler2018}, this assumption implies that units from the year immediately preceding the ``treatment year,'' defined as the first year with zero death sentences, may be misclassified due to uncertainty about the timing of the legal reform, whereas units from all other years are correctly classified.

Let 
\(
N_S := \sum_{(g,t):\, t\geq 2,\; D_{g,t} > D_{g,t-1}} N_{g,t}
\)
denote the total number of observations in observed switching-to-treatment cells.  
In staggered adoption designs, this is the only type of observed switching. To keep the analysis tractable, we restrict attention to this staggered setting and do not consider more general switching patterns, although such extensions are conceptually straightforward at the cost of more involved assumptions and arguments.

We begin by defining a first causal parameter of interest:
\begin{align*}
	\delta^{S}
	&:=\mathrm{E}\!\left[\frac{1}{N_{S}}\sum_{(i,g,t):\, t\geq 2,\; D_{g,t}\neq D_{g,t-1}}
	\left(Y_{i,g,t}(1)-Y_{i,g,t}(0)\right)\right] \\
	&=\mathrm{E}\!\left[\frac{1}{N_{S}}\sum_{(i,g,t):\, t\geq 2,\; D_{g,t}>D_{g,t-1}}
	\left(Y_{i,g,t}(1)-Y_{i,g,t}(0)\right)\right],
\end{align*}
where the second equality follows from Assumption \ref{A6}.  The parameter $\delta^{S}$ represents the average treatment effect for units in the \emph{observed} switching-to-treatment cells, that is, periods in which a group is observed to initiate treatment, aggregated across all switching periods among eventually treated groups. However, under the misclassification patterns described in Assumption \ref{A7}, some observed switching cells coincide with the true treatment adoption period, whereas others correspond to periods in which treatment had already been in place but is first recorded as starting in the data.
 Consequently, $\delta^{S}$ averages treatment effects from a mixture of true switchers and units that actually switched to treatment earlier than observed. While we show that the estimator proposed by \citet{CH2020} is generally biased in this setting, we also demonstrate how it can be adjusted to account for misclassification. Finally, in Subsection \ref{ssec:trueswitch}, we introduce the ATE of \emph{true} switchers, an arguably more meaningful estimand, which we aim to identify and estimate in Section \ref{sec:IdentTrueSwitch}.

Before introducing the estimator for \(\delta^{S}\), we impose an additional assumption on the existence of comparison groups, following \citet{CH2020}:

\begin{assumption}
	[Existence of Comparison Groups] \label{A8}
For each  $t \in\{2, \ldots, T\}$ and $g \in\{1, \ldots, G\}$ such that $$D_{g, t-1}=0, \ D_{g, t}=1,$$
 there exists $g^{\prime} \in\{1, \ldots, G\}$ with $g^{\prime} \neq g$ such that $$D_{g^{\prime}, t-1}=D_{g^{\prime}, t}=0.$$\end{assumption}

Let \( N_{1,0,t} := \sum_{g : D_{g, t} = 1, D_{g, t-1} = 0} N_{g,t} \) and \( N_{0,0,t} := \sum_{g : D_{g, t} = 0, D_{g, t-1} = 0} N_{g,t} \) denote the number of observations in, respectively, the observed switching-to-treatment cells and the comparison cells at time \( t \). Define the estimator
\begin{align}
\widehat{\operatorname{DID}} := \frac{1}{N_S} \sum_{t=2}^T N_{1,0,t}  \widehat{\operatorname{DID}}_{t}, \label{DIDESTIMATOR}
\end{align}
where
\begin{align*}
\widehat{\operatorname{DID}}_{t}
:= \frac{1}{N_{1,0,t}} \sum_{g : D_{g,t} = 1, D_{g,t-1} = 0} N_{g,t} \left(Y_{g,t} - Y_{g,t-1}\right)
- \frac{1}{N_{0,0,t}} \sum_{g : D_{g,t} = 0, D_{g,t-1} = 0} N_{g,t} \left(Y_{g,t} - Y_{g,t-1}\right),
\end{align*}
 is set to zero when there is no switching-to-treatment group (\(D_{g, t}=1\), \(D_{g, t-1}=0\)), or no comparison group  (\(D_{g, t}=0, D_{g, t-1}=0\)). Intuitively, $\widehat{\mathrm{DID}}$ first computes a time-specific difference-in-differences estimator $\widehat{\mathrm{DID}}_{t}$ by comparing outcome changes between observed switching and comparison groups, and then aggregates these two-by-two comparisons across all periods. In staggered adoption designs, this estimator coincides with the estimator $\widehat{\mathrm{DID}}_{M}$ proposed by \citet{CH2020}.

Heuristically, the DID estimator based on the \emph{observed} treatment status may fail to identify the ATE of units in the observed switching-to-treatment cells. This is because, in staggered adoption designs, the identifying parallel trends assumption concerns only the evolution of the potential nontreatment outcome \( Y_{i, g, t}(0) \). However, under misclassification, the observed just-before-switching cells are contaminated by a mix of potential nontreatment and treatment outcomes. A similar issue arises in the comparison cells.

The following proposition formalizes this argument by characterizing the DID estimand
$\mathrm{E}\!\left[\widehat{\mathrm{DID}}\right]$ and decomposing it into three components:
(i) the ATE for units in the observed switching-to-treatment cells,
(ii) a bias term arising from misclassified units that are already treated in the observed just-before-switching cells, and
(iii) an additional bias term due to misclassified, actually treated units in the comparison cells. 

\begin{proposition}[Characterizing DID estimand]\label{PROPESTIMAND}
Suppose Assumptions \ref{A1}-\ref{A8} hold. Then, the DID estimand $\mathrm{E}\left[\widehat{\mathrm{DID}}\right]$ can be decomposed as:
\small\begin{align}
\mathrm{E}\left[\widehat{\mathrm{DID}}\right]   &= \delta^S- \mathrm{E}\left[\frac{1}{N_{S}} \sum_{(g,t):t\geq 2,  D_{g, t}=1, D_{g, t-1}=0}\sum_{i=1}^{N_{g,t}} \frac{N_{g,t}}{N_{g,t-1}}D_{i,g, t-1}^*\left(Y_{i, g, t-1}(1)-Y_{i, g, t-1}(0)\right)\right]\notag\\
 &\qquad\ - \mathrm{E}\left[\frac{1}{N_{S}}\sum_{(g,t):t\geq 2, D_{g, t}=0, D_{g, t-1}=0}\sum_{i=1}^{N_{g,t}}  \frac{N_{1,0,t}}{N_{0,0,t}} D_{i,g,t}^{\ast} \left(Y_{i, g, t}(1)-Y_{i, g, t}(0)\right)\right]\label{EQ:PROP2b}.
\end{align}\normalsize\end{proposition} 
The first term in \eqref{EQ:PROP2b} corresponds to $\delta^S$, the ATE for units in the observed switching-to-treatment cells. 
 The second term captures bias arising from units that are already treated in period $t-1$ but are misclassified as untreated in the observed just-before-switching cells.  
It disappears if such misclassification is absent, i.e.,  ${D}_{i,g,t-1}^*=0$ for all $(i,g, t) \in\{1, \ldots, N_{g,t}\} \times\{1, \ldots, G\} \times\{2, \ldots, T\}$  such that $D_{i,g, t-1}=0$ and $ D_{i,g, t}=1$. To facilitate interpretation, assume $\frac{N_{g,t}}{N_{g,t-1}}=1$ for all $(g, t) \in\{1, \ldots, G\} \times\{2, \ldots, T\}$, a requirement that is satisfied when the data are a balanced panel.  This bias can be interpreted as the \textit{reweighted} ATE of the misclassified, actually treated units in the observed  just-before-switching cells, 
$$\mathrm{E}\left[w\left(\mathbf{D}\right)\frac{1}{N_{1,1;1,0}} \sum_{(i,g,t):t\geq 2, D_{i,g, t}^*=1, D_{i,g, t-1}^*=1, D_{i,g, t}=1, D_{i,g, t-1}=0} Y_{i, g, t-1}(1)-Y_{i, g, t-1}(0)\right],$$
where $N_{1,1;1,0}:=\sum_{(i,g,t):t\geq 2}\mathbb{I}\left\{D_{i,g, t}^*=1, D_{i,g, t-1}^*=1,D_{i,g, t}=1, D_{i,g, t-1}=0\right\}$, and $w\left(\mathbf{D}\right):=\frac{N_{1,1;1,0}}{N_{S}}$ is the weighting factor and between zero and one with probability one. The third term in \eqref{EQ:PROP2b} represents bias arising from misclassified, actually treated units in the comparison cells associated with the observed switching-to-treatment groups. This term disappears if $D_{i,g,t}^*=0$ for all $(i,g,t)$ such that $D_{g,t-1}=0$ and $D_{g,t}=0$. While it is also related to treatment effects for misclassified treated units, its interpretation is less transparent due to the more complex weighting structure, and we therefore do not pursue it further.\footnote{If misclassification extends beyond periods immediately preceding observed switching, an additional bias term analogous to the third term in \eqref{EQ:PROP2b} arises for period $t-1$.}

\subsubsection{Illustration Example}

 For illustration, consider a stylized scenario with a balanced panel satisfying $N_{g,t}=N_{g,t-1}$ for all $(g, t) \in\{1, \ldots, G\} \times\{2, \ldots, T\}$ and treatment effects are homogeneous across groups and over time: $${Y}_{i,g, t}(1)-{Y}_{i,g, t}(0)=\Delta\in \mathbb{R},$$ for all $(i,g, t) \in\{1, \ldots, N_{g,t}\} \times\{1, \ldots, G\} \times\{2, \ldots, T\}$. Suppose misclassification is homogeneous within all observed  just-before-switching cells:
$$\mathrm{E}\left[{D}_{i,g, t-1}^*\mid D_{i,g, t}=1, D_{i,g, t-1}=0, \mathbf{D}_{-(i,g,t,t-1)}\right]=p \in \mathbb{R},$$
for all $(i,g, t) \in\{1, \ldots, N_{g,t}\} \times\{1, \ldots, G\} \times\{2, \ldots, T\}$, where $$\mathbf{D}_{-(i,g,t,t-1)}:=\left({D}_{i,g, t}\right)_{(i,g, t) \in\{1, \ldots, N_{g,t}\} \times\{1, \ldots, G\} \times\{1, \ldots, T\}\setminus \{(i,g,t-1),(i,g,t)\}}$$ denotes the observed treatment status of all units other than $(i,g,t-1)$ and $(i,g,t)$. Also, for simplicity, we rule out misclassification in any comparison cells for the observed switching-to-treatment cells, 
$$\mathrm{E}\left[{D}_{i,g, t}^*\mid D_{i,g, t}=0, D_{i,g, t-1}=0, \mathbf{D}_{-(i,g,t,t-1)}\right]=0,$$
for all $(i,g, t) \in\{1, \ldots, N_{g,t}\} \times\{1, \ldots, G\} \times\{2, \ldots, T\}$. Then, the second term of the decomposition formula (\ref{EQ:PROP2b}) becomes
\begin{align*}
     &\quad\  \mathrm{E}\left[\frac{1}{N_{S}} \sum_{(g,t):t\geq 2, D_{g, t}=1, D_{g, t-1}=0} \frac{N_{g,t}}{N_{g,t-1}}\sum_{i=1}^{N_{g,t}} D_{i,g,t-1}^*\left(Y_{i, g, t-1}(1)-Y_{i, g, t-1}(0)\right)\right]\\
& =    \Delta\mathrm{E}\left[\frac{1}{N_{S}} \sum_{(g,t):t\geq 2, D_{g, t}=1, D_{g, t-1}=0}\sum_{i=1}^{N_{g,t}} \mathrm{E}\left[{D}_{i,g,t-1}^*\mid \mathbf{D}\right]\right]\\
& =    \Delta\mathrm{E}\left[\frac{1}{N_{S}} \sum_{(g,t):t\geq 2, D_{g, t}=1, D_{g, t-1}=0} \sum_{i=1}^{N_{g,t}}\mathrm{E}\left[{D}_{i,g, t-1}^*\mid D_{i,g, t}=1, D_{i,g, t-1}=0, \mathbf{D}_{-(i,g,t,t-1)}\right]\right]\\
     &=\Delta p,
\end{align*}\normalsize
while the last term disappears, so the bias boils down to attenuation bias: 
$$\mathrm{E}\left[\widehat{\mathrm{DID}}\right]=\Delta(1-p).$$

\subsection{ATE of True Switchers}\label{ssec:trueswitch}

As an alternative estimand, we also consider the ATE for \emph{true} (or \emph{actual}) switching units in the period in which treatment is first received. We refer to this parameter as the \emph{ATE of true switchers}. To define it, recall that the unobserved switching indicator $S^{*}_{i,g,t}$ equals one if unit $i$ in group $g$ truly switches into treatment in period $t$, that is, if $D^{*}_{i,g,t}=1$ and $D^{*}_{i,g,t-1}=0$, and equals zero otherwise.

Let $N_{S^*}:=\sum_{(i,g,t): t\geq2} \mathbb{I}\left\{ S^{*}_{i,g,t}=1\right\}$ denote the number of true switching units. The  ATE of true switchers is then given by:
\begin{align*}
	\delta^{S^{*}}&:=\mathrm{E}\left[\frac{1}{N_{S^*}}\sum_{(i, g,t):  t\geq 2, S_{i,g, t}^* =1} \left({Y}_{i, g, t}(1)-{Y}_{i, g, t}(0)\right)\right].
	\end{align*}
As noted in the introduction, $\delta^{S^{*}}$ captures the average treatment effect for units at the time they \emph{actually} initiate treatment, aggregated across all switching periods among eventually treated groups. This parameter is particularly relevant in applications where the initial impact of treatment is of primary policy interest.

We obtain the following result for identification of the ATE of true switchers.

\begin{proposition}[DID for True Switchers]\label{PROPESTIMANDTRUE}
Suppose Assumptions \ref{A1}-\ref{A8} hold. Assume that for every \(t\geq 2\) and $g\in\{1,\ldots,G\}$ such that $D_{g,t}>D_{g,t-1}$, it holds for all $(i,j)\in\{1,\ldots,N_{g,t}\}\times \{1,\ldots,N_{g,t-1}\}$ that:
\[
\mathrm{Pr}\left({S}_{j,g,t-1}^{\ast}=1\mid D_{g,t} = 1, D_{g,t-1} = 0\right)+\mathrm{Pr}\left({S}_{i,g,t}^{\ast}=1\mid D_{g,t} = 1, D_{g,t-1} = 0 \right)=1.
\] 
On the contrary, for all \(t\geq 2\) and $g\in\{1,\ldots,G\}$ such that $D_{g,t}=D_{g,t-1}$, it holds for all $(i,j)\in\{1,\ldots,N_{g,t}\}\times \{1,\ldots,N_{g,t-1}\}$ that:
\[
\mathrm{Pr}\left({S}_{j,g,t-1}^{\ast}=1\mid D_{g,t} = D_{g,t-1} \right)+\mathrm{Pr}\left({S}_{i,g,t}^{\ast}=1\mid D_{g,t} = D_{g,t-1}\right)=0.
\]  Then, the DID estimand $\mathrm{E}\left[\widehat{\mathrm{DID}}\right]$ can be decomposed as:
\begin{align}
 \mathrm{E}\left[\widehat{\mathrm{DID}}\right]   &=\delta^{S^{\ast}} \notag\\
&\quad +\mathrm{E}\left[\left(\frac{1}{N_{S}}-\frac{1}{N_{S^{\ast}}}\right) \sum_{(g,t):t\geq 2, D_{g, t}=1, D_{g, t-1}=0} \sum_{i=1}^{N_{g,t}} D^{\ast}_{i,g, t} S_{i,g,t}^{\ast}\left(Y_{i, g, t}(1)-Y_{i, g, t}(0)\right)	\right]\notag\\
  &\quad + \mathrm{E}\left[\frac{1}{N_{S}} \sum_{(g,t):t\geq 2, D_{g, t}=1, D_{g, t-1}=0} \sum_{i=1}^{N_{g,t}} D^{\ast}_{i,g, t}  \left(1-S_{i,g,t}^{\ast}\right)\left(Y_{i, g, t}(1)-Y_{i, g, t}(0)\right)\right]\notag\\
 &\quad - \mathrm{E}\left[\frac{1}{N_{S}}\sum_{(g,t):t\geq 2, D_{g, t}=1, D_{g, t-1}=0}\sum_{i=1}^{N_{g,t}}\frac{N_{g,t}}{N_{g,t-1}}{D}_{i,g,t-1}^*\left(Y_{i, g, t-1}(1)-Y_{i, g, t-1}(0)\right)\right]\label{EQ:PROP2c}\\
&\quad +\mathrm{E}\left[\left(\frac{1}{N_{S}}-\frac{1}{N_{S^{\ast}}}\right)\sum_{(g,t):t\geq 2, D_{g, t}=0, D_{g, t-1}=0}\sum_{i=1}^{N_{g,t}} {D}_{i,g,t}^*S_{i,g,t}^{\ast}\left(Y_{i, g, t}(1)-Y_{i, g, t}(0)\right)\right]\notag\\
&\quad +\mathrm{E}\left[\frac{1}{N_{S}}\sum_{(g,t):t\geq 2, D_{g, t}=0, D_{g, t-1}=0}\sum_{i=1}^{N_{g,t}} {D}_{i,g,t}^*\left(1-S_{i,g,t}^{\ast}\right)\left(Y_{i, g, t}(1)-Y_{i, g, t}(0)\right)\right]\notag\\
&\quad -\mathrm{E}\left[\frac{1}{N_{S}}\sum_{(g,t):t\geq 2, D_{g, t}=0, D_{g, t-1}=0}\sum_{i=1}^{N_{g,t}}\left(\frac{N_{1,0,t}}{N_{0,0,t}}-1\right) {D}_{i,g,t}^*\left(Y_{i, g, t}(1)-Y_{i, g, t}(0)\right)\right].\notag
\end{align}\normalsize\end{proposition} 
Proposition \ref{PROPESTIMANDTRUE} provides a detailed decomposition of the bias of the DID estimand relative to the ATE of true switchers. The first two terms following $\delta^{S^{*}}$ arise because, in observed switching periods, some individuals may in fact have already switched into treatment in the preceding period and therefore are not true switchers at time $t$. The subsequent terms reflect analogous contamination in earlier periods that are recorded as pre-treatment, as well as in the comparison groups, where some individuals may actually switch into treatment shortly after the comparison period. Together, these terms illustrate how misalignment between observed and true switching times distorts the DID estimand.

Finally, note that the framework above can equivalently be interpreted as capturing treatment anticipation rather than misclassification. In particular, $Y_{i,g,t-1}(1)$ may represent outcomes realized in period $t-1$ if treatment exposure is anticipated prior to the recorded switch in period $t$. Focusing on anticipation in the period immediately preceding the observed switch for simplicity, Propositions \ref{PROPESTIMAND} and \ref{PROPESTIMANDTRUE} continue to apply, illustrating that anticipation effects also generate bias in the DID estimand.
\medskip

\noindent \textbf{Remark 1}: In the anticipation interpretation, the setup formally corresponds to a case of ``full anticipation,'' in the sense that anticipating individuals experience $Y_{i,g,t-1}(1)$ already in period $t-1$. Importantly, however, we place no restriction on the fraction of individuals who anticipate treatment. Moreover, a model with flexible anticipation probabilities and outcomes $Y_{i,g,t-1}(1)$ is observationally equivalent to an alternative model featuring ``fractional anticipation'' combined with different anticipation probabilities. For this reason, we do not view the full-anticipation formulation as restrictive.

\medskip

\noindent \textbf{Remark 2}: Similar to the case of anticipation, our results and methodology are also applicable in setups with spillover effects. For example, when a policy is implemented at the regional or state level, units in neighboring locations may experience treated potential outcomes $Y_{i,g,t}(1)$ even though their observed treatment status remains $D_{i,g,t}=0$. In such cases, bias expressions analogous to those in Propositions \ref{PROPESTIMAND} and \ref{PROPESTIMANDTRUE} can be derived. As an illustrative example, consider the legalization of recreational marijuana consumption for adults in the U.S., where different states adopted legalization in different years, such as New York in 2021. To the extent that some residents of neighboring non-legalizing states, such as Pennsylvania, can benefit from legal purchasing opportunities by crossing state borders, the DID estimator $\widehat{\mathrm{DID}}$ in \eqref{DIDESTIMATOR} will generally be biased relative to the ATE of observed switchers and the ATE of true switchers.
\medskip

\section{Estimation of  ATEs}\label{sec:IdentTrueSwitch}

In this section, we propose modified versions of the $\widehat{\mathrm{DID}}$ estimator that recover both the ATE of observed switchers, $\delta^{S}$, and the ATE of true switchers, $\delta^{S^{*}}$. We begin with estimation of the former.

\subsection{ATE of Observed Switchers}\label{ssec:ATEOS}

Let $N_{1,0,t}^{[1]}:=\sum_{g:  D_{g, t+1}=1, D_{g, t}=0} N_{g,t}$ and $N_{0,0,t}^{[1]}:=\sum_{g:  D_{g, t+1}=0, D_{g, t}=0} N_{g,t}$ and define
\begin{align}
\widehat{\mathrm{DID}}^{S}= \frac{1}{N_{S}} \sum_{t\geq2} \widehat{\mathrm{DID}}^{S}_{t},\label{DIDESTIMATORSTAR}
\end{align}
where \[
\widehat{\mathrm{DID}}^{S}_{t}=N_{1,0,t}\widehat{\operatorname{DID}}_{ t}+N_{1,0,t}\widehat{\operatorname{DID}}_{ t}^{*}+\frac{N_{1,0,t}^{[1]}N_{1,0,t}}{N_{0,0,t}}\widehat{\operatorname{DID}}_{ t}^{**}
\] 
with
\begin{align*}
 \widehat{\operatorname{DID}}_{ t}^{*} &=\frac{1}{N_{1,0,t}}\sum_{g: D_{g, t}=1, D_{g, t-1}=0} N_{g,t}\left(Y_{g, t-1}-Y_{g, t-2}\right) - \frac{1}{N_{0,0,t}}\sum_{g: D_{g, t}=0, D_{g, t-1}=0} N_{g,t}\left(Y_{g, t-1}-Y_{g, t-2} \right)\\
\widehat{\operatorname{DID}}_{ t}^{**} 
&=\frac{1}{N_{1,0,t}^{[1]}}\sum_{g: D_{g, t+1}=1, D_{g, t}=0} N_{g,t}\left(Y_{g, t}-Y_{g, t-1}\right) - \frac{1}{N_{0,0,t}^{[1]}}\sum_{g: D_{g, t+1}=0, D_{g, t}=0} N_{g,t}\left(Y_{g, t}-Y_{g, t-1} \right).
\end{align*}

As with $\widehat{\mathrm{DID}}_{t}$, each of the correction terms $\widehat{\mathrm{DID}}^{*}_{t}$ and $\widehat{\mathrm{DID}}^{**}_{t}$ is set to zero whenever the corresponding treatment or comparison group does not exist. In particular, we set $\widehat{\mathrm{DID}}^{*}_{2}=0$ and $\widehat{\mathrm{DID}}^{**}_{T}=0$, since observations prior to period $1$ and after period $T$ are unavailable.

Heuristically, the estimator \(\widehat{\mathrm{DID}}^{S}\) in (\ref{DIDESTIMATORSTAR}) corrects the original \(\widehat{\mathrm{DID}}\) estimator for the two sources of bias identified in Proposition \ref{PROPESTIMAND}, using appropriate sample counterparts. Specifically, \(\widehat{\operatorname{DID}}_{t}^{*}\) is a ``backward-looking correction'': it compares outcome changes from \(t-2\) to \(t-1\), to adjust for bias due to misclassified treated units in the treatment group at \(t-1\). In contrast, \(\widehat{\operatorname{DID}}_{t}^{**}\) is a ``forward-looking correction'': it compares outcome changes from \(t-1\) to \(t\) between units still observed as untreated at \(t\), but differing in whether they begin treatment at \(t+1\), to adjust for bias due to misclassified treated units in the comparison group at \(t\). Together, these corrections yield an unbiased estimator of the ATE of observed switchers, \(\delta^{S}\).

Since our correction relies on outcome information from both periods preceding and following the observed switch, we impose the following additional conditions:

 \begin{assumption}
	[No Misclassification in Never Treated] \label{A9}  For each $(i,g)\in \{1,\cdots,N_{g,T}\}\times \{1,\cdots, G\}$:  
	 $$\mathrm{Pr}\left(D_{i,g,T}^* = 0 \mid D_{g,T} = 0\right) = 1.
$$
\end{assumption}

 \begin{assumption}
	[No Observed Switch in Initial Period] \label{A10} For each $g\in \{1,\cdots, G\}$:  
	\[\Pr\left(D_{g, 1}= D_{g, 2}\right)=1.\]
\end{assumption}

In staggered adoption designs with the misclassification patterns misclassification patterns described in Assumption \ref{A7}, Assumption \ref{A9} ensures that  the never treated groups are always correctly classified. If a group indeed begins to receive treatment at \(T+1\), it could be misclassified at time \(T\), so that this condition is violated. In such cases, the researcher may address the issue by collecting data for one additional period, omitting the final period from analysis, or excluding the affected group (particularly, if the group is not central to the research question and sufficient comparison groups remain). On the other hand, Assumption \ref{A10} places a restriction on the data to rule out observed switches in the initial period. Combined with Assumption \ref{A7}, it ensures that no misclassification occurs at \(t=1\), which would otherwise be uncorrectable due to the absence of pre-period information. Assumption \ref{A10} rules out observed treatment switches in the initial period. When combined with Assumption \ref{A7}, it guarantees that no misclassification occurs at $t=1$, which would otherwise be uncorrectable due to the absence of pre-treatment information.

We now establish that \(\widehat{\mathrm{DID}}^{S}\) provides an unbiased estimator of \(\delta^{S}\):

\begin{proposition}[DID for Observed Switchers]\label{PROP6}
Suppose Assumptions \ref{A1}-\ref{A10} hold.  Then:
\begin{align*}
 \mathrm{E}\left[\widehat{\mathrm{DID}}^{S} \right]&= \delta^{S}.
\end{align*}\normalsize\end{proposition} 

\medskip

\noindent \textbf{Remark 3}: Throughout the paper, we ignore the presence of additional observed control variables $\mathbf{X}_{i,gt}$. In fact, when $\mathbf{X}_{i,gt}$ is discrete with a finite number of support points, we can readily modify all assumptions to hold also conditional on $\mathbf{X}_{i,gt}$. Together with a suitable overlap condition \citep{Khan2010}, we may then construct the estimator $\widehat{\mathrm{DID}}^{S}$ for each value of the covariate vector to recover an estimate of the ATE of Observed Switchers conditional $\mathbf{X}_{i,11}$. Alternatively, we may also aggregate the estimates from each cell to obtain an estimate of the unconditional ATE of Observed Switchers $\delta^{S}$. When $\mathbf{X}_{i,gt}$ is continuous or discrete with infinite support, estimation is more involved \citep[e.g.,][]{Abadie2005}.
\medskip

\subsection{ATE of True Switchers}\label{EstimationATETRUE}

Since $\delta^{S^{*}}$, the ATE of true switchers, is arguably a more substantively meaningful estimand than $\delta^{S}$, we now propose an estimator for this second target parameter. Recall that $\delta^{S^{*}}$ is defined as
\[
\delta^{S^{\ast}}=\mathrm{E}\left[\frac{1}{N_{S^{\ast}}} \sum_{(i,g,t):t\geq 2,  S^{\ast}_{i,g, t}=1} \left(Y_{i, g, t}(1)-Y_{i, g, t}(0)\right)\right].
\]
 We construct an estimator of $\delta^{S^{\ast}}$ as:
\begin{align}
\widehat{\mathrm{DID}}^{S^{\ast}}&=  \frac{1}{\widehat{N}_{S^{\ast}}}\sum_{t\geq 2}\left(\widehat{\lambda}_{t-1}\widehat{\mathrm{DID}}^{S^{\dagger}}_{t-1}+\left(1-\widehat{\lambda}_{t-1}\right)\widehat{\mathrm{DID}}^{S}_t\right)
,\label{DIDESTIMATORDAG}
\end{align}
where  \[\widehat{\mathrm{DID}}^{S^{\dagger}}_{t-1}=N_{1,0,t-1}^{[1]}\widehat{\operatorname{DID}}_{t-1}^{**}+N_{1,0,t-1}^{[1]}\widehat{\operatorname{DID}}_{t-1}^{***} \quad \text{and} \quad 
\widehat{\lambda}_{t-1}\equiv \frac{\widehat{\operatorname{DID}}_{t-1}^{**} }{\widehat{\operatorname{DID}}_{t-1}^{**} +\widehat{\mathrm{DID}}_{t-1}^{***}},\quad
\]
and
\[
\widehat{N}_{S^{\ast}}\equiv \sum_{t\geq 2}\sum_{D_{g, t}=1, D_{g, t-1}=0}\left(\widehat{\lambda}_{t-1}N_{g,t-1} + (1-\widehat{\lambda}_{t-1})N_{g,t}\right),
\]
with the component:
\begin{align*}
\widehat{\mathrm{DID}}_{t-1}^{***}=&\frac{1}{N_{1,0,t-1}^{[1]}}\sum_{g: D_{g, t}=1, D_{g, t-1}=0} N_{g,t-1}\left(Y_{g, t}-Y_{g, t-1}\right)-\frac{1}{N_{1,1,t-1}^{[1]}}\sum_{g: D_{g, t}=1, D_{g, t-1}=1} N_{g,t-1}\left(Y_{g, t}-Y_{g, t-1}\right),
\end{align*}
while \(\widehat{\mathrm{DID}}^{S}_t\) and \(\widehat{\mathrm{DID}}_{t}^{**}\) are given in (\ref{DIDESTIMATORSTAR}). As before, we adopt the convention that $\widehat{\mathrm{DID}}^{***}_{t}$ is set to zero whenever any of the corresponding groups do not exist in the data. We also set $\widehat{\mathrm{DID}}^{**}_{1}=0$, and define $\widehat{\lambda}_{t-1}=0$ whenever both $\widehat{\mathrm{DID}}^{**}_{t-1}$ and $\widehat{\mathrm{DID}}^{***}_{t-1}$ equal zero.
  Heuristically, the estimator \(\widehat{\mathrm{DID}}^{S^{\ast}}\) in (\ref{DIDESTIMATORDAG}) aggregates treatment effect estimates from both the observed just-before-switching and switching-to-treatment periods, weighted by the estimated probability of a true switch occurring in each. Under the identifying restrictions introduced below, $\widehat{\mathrm{DID}}^{S^{\dagger}}_{t-1}$ estimates the (unscaled) ATE in period $t-1$ for units that truly switch in $t-1$ but are only observed to switch in $t$, whereas $\widehat{\mathrm{DID}}^{S}_{t}$ estimates the (unscaled) ATE in period $t$ for units whose switch is both actual and observed at time $t$. Moreover, the weight $\widehat{\lambda}_{t-1}$ estimates the share of the former group, while $1-\widehat{\lambda}_{t-1}$ estimates the share of the latter.


We impose the following additional assumptions.

\begin{assumption}[Patterns of Switch] \label{A11}    For each $(i,j,g, t) \in \{1, \ldots, N_{g,t}\} \times\{1, \ldots, N_{g,t-1}\} \times \{1, \ldots, G\} \times\{2, \ldots, T\}$:
\begin{enumerate}
	\item \(\mathrm{Pr}\left( S^{\ast}_{i,g,t} = 0 \mid D_{g,t} = 1,\ D_{g,t-1} = 1 \right)= 1\).
	\item \(\mathrm{Pr}\left( S^{\ast}_{j,g,t-1} = 0 \mid D_{g,t} = 0,\ D_{g,t-1} = 0 \right) = 1.\)
\end{enumerate}
 In addition,  for each $(i,g)\in \{1, \ldots, N_{g,T}\} \times\{1,\cdots, G\}$:  
	 $$\mathrm{Pr}\left(S_{i,g,T}^* = 0 \mid D_{g,T} = 0 \right) = 1.
$$
\end{assumption}

\begin{assumption}[Switcher Homogeneity] \label{A13}   For each $(i,j,g,t)\in\{1, \ldots, N_{g,t}\} \times\{1, \ldots, N_{g,t-1}\} \times\{1,\ldots,G\}\times \{2,\ldots,T\}$:
\begin{enumerate}
	\item \(\mathrm{Pr}\left({D}_{j,g,t-1}^{\ast}=1\mid D_{g,t} = 1,\ D_{g,t-1} = 0 \right)=\mathrm{Pr}\left({S}_{j,g,t-1}^{\ast}=1\mid D_{g,t} = 1,\ D_{g,t-1} = 0 \right)= \lambda_{t-1}\).
	\item \(\mathrm{Pr}\left( S^{\ast}_{i,g,t} = 1 \mid D_{g,t} = 1,\ D_{g,t-1} = 0 \right) +\mathrm{Pr}\left( S^{\ast}_{j,g,t-1} = 1 \mid D_{g,t} = 1,\ D_{g,t-1} = 0 \right)=1.\)
\end{enumerate}
\end{assumption}

\begin{assumption}[Treatment Effects Homogeneity] \label{A14} For each $(i,g,t)\in\{1,\cdots,N_{g,t}\}\times\{1,\ldots,G\}\times \{2,\ldots,T\}$, it holds that:
\begin{align*}
\mathrm{E}\left[Y_{i, g, t}(1)-Y_{i, g, t}(0)\mid D_{i,g,t}^{\ast}, S_{i,g,t}^{\ast},\mathbf{D}_g\right]=\mathrm{E}\left[Y_{i, g, t}(1)-Y_{i, g, t}(0)\right].
\end{align*}
\end{assumption}

Assumption \ref{A11} governs the behavior of the latent switching indicator $S^{*}_{i,g,t}$. It restricts true switching to be determined solely by the local information contained in $(D_{g,t-1},D_{g,t})$, rather than by the full observed treatment path of group $g$. In staggered adoption designs, Assumption \ref{A11} complements Assumption \ref{A7} by imposing coherent restrictions in terms of the unobserved actual switching indicator \(S^{\ast}_{i,g,t}\). Specifically, it rules out the possibility that a true switch, i.e., \(S^*_{i,g,t} = 1\), occurs within a run of treated or untreated periods, except at the first treated period or the last untreated period, respectively. The last part of the assumption parallels Assumption \ref{A9} and further rules out such a switch in the last untreated period among the never treated.

Assumption \ref{A13} imposes homogeneity on the switching process. The first part links the latent treatment status $D^{*}_{i,g,t-1}$ to the switching indicator $S^{*}_{i,g,t-1}$ and requires that, among groups observed to switch between periods $t-1$ and $t$, all units share the same probability $\lambda_{t-1}$ of having actually switched in the earlier period.\footnote{Recall that for notational simplicity, we restrict attention to misclassification or anticipation occurring only in the period immediately preceding the observed switch. The framework can be extended to allow for multiple pre-switch periods with distinct switching probabilities, at the cost of more involved notation.} The second part ensures that the total switching probability is fully allocated between the just-before-switching period and the observed switching period.

Finally, Assumption \ref{A14} restricts treatment effect heterogeneity. It requires that expected unit-level treatment effects do not depend on a unit’s actual treatment status or switching status, although they may still vary across individuals, groups, or time. Thus, the assumption limits only how treatment effect heterogeneity may correlate with latent treatment assignment and switching, rather than ruling out heterogeneity altogether.


To recover $\delta^{S^{\ast}}_{0}$ , we also require extensions of Assumptions \ref{A3}-\ref{A5} to potential outcomes under treatment:

\begin{assumption}[Independent Groups] \label{A15} The group-level vectors are mutually independent across \(g\):
$$\left(Y_{g,t}(0),Y_{g,t}(1),D_{g,t},(D_{i,g,t}^{\ast})_{i\in \{1,\cdots,N_{g,t}\}},\left(S_{i,g,t}^{\ast}\right)_{i\in \{1,\cdots,N_{g,t}\}}\right)_{1\leq t\leq T}$$ 
\end{assumption}

\begin{assumption}[Strong Exogeneity] \label{A16}
For each $(d,g, t) \in\{0,1\}\times \{1, \ldots, G\} \times\{2, \ldots, T\}$: $$\mathrm{E}\left[{Y}_{g, t}(d)-{Y}_{g, t-1}(d) \mid \mathbf{{D}}_g\right]=\mathrm{E}\left[{Y}_{g, t}(d)-{Y}_{g, t-1}(d)\right].$$\end{assumption}

\begin{assumption}[Parallel Trends] \label{A17}
For each  $(d,g,k, t) \in\{1, \ldots, G\} \times\{1, \ldots, G\} \times\{2, \ldots, T\}$: 
$$\mathrm{E}\left[{Y}_{g, t}(d)-{Y}_{g, t-1}(d)\right]=\mathrm{E}\left[{Y}_{k, t}(d)-{Y}_{k, t-1}(d)\right].$$
\end{assumption}

Assumption \ref{A15} parallels Assumption 3 from \citet{CH2020} and imposes an independence assumption across groups (but not homogeneity). It replaces Assumption \ref{A3} and extends it (through the independence) to outcomes in the presence of treatment $Y_{i,g,t}(1)$.   Similarly, Assumptions \ref{A16} and \ref{A17} also represent extensions of Assumptions  \ref{A4} and \ref{A5} to potential outcomes under treatment. They are needed to identify the (unobserved) share of true switchers in a given switching-to-treatment cell, which in turn relies on a comparison with groups that have already switched to treatment in earlier periods. In fact, since standard staggered adoption settings without misspecification do not require such assumptions, these conditions may be viewed as the ``price to pay'' for not observing the actual switching status. We note however that more general DiD set-ups that also allow for ``reverse switchers'', i.e. units switching status back from treatment to no treatment, also rely on Assumptions \ref{A15}-\ref{A17} for identification \citep[cf.][]{CH2020,F2023}.

Finally, since the exact period of switching-to-treatment is unknown at the individual level and	 we may only identify the fraction of switchers in a given period in an asymptotic framework, we focus on an estimator that is consistent for  $\delta^{S^{\ast}}$ in this section. To establish consistency, we also require a couple of technical conditions concerning limiting quantities:

\begin{assumption}
	[Existence of Comparison Groups - 2] \label{A18}
For each  $g \in\{1, \ldots, G\}$ and $t\in\{2,\ldots,T\}$ such that 
\[
\Pr\left(D_{g,t} = 1, D_{g,t-1} = 0\right) > 0,
\] 
implies that $\lim_{G\rightarrow\infty} \mathrm{E}\left[\frac{N_{1,0,t}}{G}\right] >
0$ and $\lim_{G\rightarrow\infty} \mathrm{E}\left[\frac{N_{0,0,t}}{G}\right] >
0$. Moreover, it also implies that $\lim_{G\rightarrow\infty} \mathrm{E}\left[\frac{N_{j,0,t}^{[1]}}{G}\right] >
0$ and $\lim_{G\rightarrow\infty} \mathrm{E}\left[\frac{N_{j,0,t-1}^{[1]}}{G}\right] >
0$ as well as $\lim_{G\rightarrow\infty} \mathrm{E}\left[\frac{N_{1,1,t-1}^{[1]}}{G}\right] >
0$. Finally, for at least one $t\in\{2,\ldots,T\}$, it holds that $\Pr\left(D_{g,t} = 1, D_{g,t-1} = 0\right) > 0$.
\end{assumption}

\begin{assumption}[Existence of Moments] \label{A19} $\sup_{(g,t)}N_{g,t} <+\infty$  and $\sup_{(d,i,g,t)}\mathrm{E}\left[Y_{i,g,t}(d)^4\right]<+\infty$. Moreover, let $\lim_{G\rightarrow \infty}\mathrm{E}[\mathbf{V}]$ and $\mathbf{\Sigma}_{\mathbf{V}}$ defined in (\ref{VVECTOR}) and (\ref{ASYVARTS}) of Appendix \ref{AppendixA} exist.
\end{assumption}

Assumption \ref{A18} is an asymptotic analogue of Assumption \ref{A8} and ensures the existence of all comparison groups required to construct $\widehat{\mathrm{DID}}^{S^{\ast}}$.  Assumption \ref{A19} requires that the fourth moments of the potential outcomes \(Y_{i,g,t}(d)\) are uniformly bounded across all indexes. It also assumes that the cell sizes \(N_{g,t}\) are uniformly bounded. Finally, the assumption ensures the existence of the limiting expectation vector \(\lim_{G \to \infty} \mathrm{E}[\mathbf{V}]\) and the asymptotic covariance \(\boldsymbol{\Sigma}_{\mathbf{V}}\), both of which are required for establishing the asymptotic distribution of the estimator.

Considering an asymptotic framework in which the number of groups $G$ tends to infinity while the number of time periods $T$ remains fixed, we obtain the following result.
\begin{proposition}[DID for True Switchers]\label{PROP7}
Suppose that Assumptions \ref{A1}-\ref{A2}, \ref{A6}-\ref{A7} and \ref{A9}-\ref{A19} hold. Moreover, it holds that $\lim_{G\rightarrow \infty}\Pr\left(\left\vert\mathrm{E}\left[\widehat{\operatorname{DID}}_{t-1}^{**}\right]+\mathrm{E}\left[\widehat{\operatorname{DID}}_{t-1}^{***}\right]\right\vert >0\right)>0$ for all $t\in\{2,\ldots,T\}$. Then:
\[
\sqrt{G}\left(\widehat{\mathrm{DID}}^{S^{\ast}}- \delta^{S^{\ast}}\right)\stackrel{d}{\rightarrow}\mathcal{N}(0,\sigma_{\delta^{S^{\ast}}}^{2})
\]
with $\sigma_{\delta^{S^{\ast}}_{0}}^{2}:= \mathbf{J}_{f}	\left(\mathbf{V}^{\infty}\right)\boldsymbol{\Sigma}_{\mathbf{V}}\mathbf{J}_{f}	\left(\mathbf{V}^{\infty}\right)^{\prime}$, where $\mathbf{V}$ is defined in (\ref{VVECTOR}) and $\mathbf{V}^{\infty}=\lim_{G\rightarrow \infty}\mathrm{E}\left[\mathbf{V}\right]$.
\normalsize\end{proposition}

Proposition \ref{PROP7} establishes the limiting distribution of $\widehat{\mathrm{DID}}^{S^{\ast}}$.  Note that the condition:
\begin{equation}\label{EQPOSPROB}
\lim_{G\rightarrow \infty}\Pr\left(\left\vert\mathrm{E}\left[\widehat{\operatorname{DID}}_{t-1}^{**}\right]+\mathrm{E}\left[\widehat{\operatorname{DID}}_{t-1}^{***}\right]\right\vert >0\right)>0
\end{equation}
is a technical condition that ensures a well-defined limit. It requires the existence of a non-zero treatment effect prior to treatment, but can be weakened by introducing a suitable trimming function into the estimator $\widehat{\mathrm{DID}}^{S^{\ast}}$. Specifically, denote by $\mathcal{T}_{DID,G}$ the set of periods such that
\[
\left\vert \mathrm{E}\left[\widehat{\operatorname{DID}}_{t-1}^{**}\right]+\mathrm{E}\left[\widehat{\operatorname{DID}}_{t-1}^{***}\right]\right\vert >c_{G}
\] 
for some $c_{G}\rightarrow 0$ as $G\rightarrow \infty$. Then, the estimator in (\ref{DIDESTIMATORDAG}) can be modified as:
\begin{align}\label{EQESTMOD}
\widehat{\mathrm{DID}}^{S^{\ast}}&=  \frac{1}{\widehat{N}_{S^{\ast}}}\sum_{t\geq 2}\mathbb{I}\left\{t\in \mathcal{T}_{DID,G}\right\}\left(\widehat{\lambda}_{t-1}\widehat{\mathrm{DID}}^{S^{\dagger}}_{t-1}+\left(1-\widehat{\lambda}_{t-1}\right)\widehat{\mathrm{DID}}^{S}_t\right).
\end{align}
Under suitable rate conditions, the modified estimator in (\ref{EQESTMOD}) is a consistent estimator for the ATE of true switchers for all periods in $\mathcal{T}_{DID,G}$. Alternatively, a continuous, differentiable function could also be used for trimming \citep[e.g.,][]{G2016}.  Finally, note that the testing procedure outlined in the next section allows to test for the presence of (non-zero) anticipation or misclassification effects, and may thus also inform the researcher about the suitability of (\ref{EQPOSPROB}).

\section{Specification Tests}\label{sec:testing}

As illustrated in the previous sections, not accounting for misclassification or anticipation in the $\text{DID}$ estimator  from Equation (\ref{DIDESTIMATOR}) leads to a bias not only with respect to $\Delta$ in the case of homogeneous treatment effects, but also with respect to both $\delta_{S}$ and $\delta_{S^{\ast}}$, respectively. As a result, a modification of $\text{DID}$ is required. However, in  the absence of further information, it is not clear how to modify the aforementioned estimator when the researcher does not know how many periods prior to the observed switches are affected. In addition, even if knowledge about affected periods was available, inspection of the proof of Proposition \ref{PROPESTIMAND} suggests that it is generally difficult to disentangle the bias arising from misclassification or anticipation on the one hand and from a failure of  parallel trends (PT) in Assumption \ref{A5} on the other. This is, however, important as the implications are different. In particular, while anticipation or misclassification in pre-switch periods can be addressed through a modified DID estimator using earlier periods for comparison (subject to data availability), a failure of strong exogeneity and thus PT generally invalidates  the DID approach for point identification and estimation altogether.\footnote{Of course, it may still be possible to adopt a partial identification approach as in  \citet{Rambachan2023} using pre-treatment evidence to bound post-treatment outcome differences.} 

To address these challenges, we propose two specification tests based on moment equalities to detect violations of PT in pre-treatment periods on the one hand, and misspecification due to anticipation effects or misclassifcation on the other.  In particular, the first test compares outcome averages from  switching groups with  not-yet switched groups in periods prior to a threshold date assumed to be known and not affected by misclassification.\footnote{For instance, in the context of the study of \citet{Bindler2018}, who analyzed the effects of the (staggered) abolition of capital punishment across different offense categories on jury severity, one may use periods that preceed the enactment of these law changes.} The second test is conceptually identical, but instead focuses on comparisons from periods potentially affected by misclassification. Heuristically, a researcher may therefore use the first test to gather evidence for pre-treatment violations of PT, and then conduct the second test for misclassification  comparing the same groups using periods that could (potentially) be misclassified. Moreover, as both test rely on moment comparisons and can exploit all possible combinations from the time periods considered, they are more likely to detect violations of PT or misclassification in pre-treatment periods than conventional event-study based specification tests which just check for non-zero effects relative to a given reference period (see Remark 6 below).

Formally, recall that the parallel trends assumption naturally  extends to other pre-treatment periods. That is, if for each $(g,k,t)\in\{1,\ldots,G\}\times \{1,\ldots,G\}\times 	\{2,\ldots,T\}$:
$$\mathrm{E}\left[{Y}_{g, t}(0)-{Y}_{g, t-1}(0)\right]=\mathrm{E}\left[{Y}_{k, t}(0)-{Y}_{k, t-1}(0)\right],$$
then it also holds that:
$$\mathrm{E}\left[{Y}_{g, t-1}(0)-{Y}_{g, t-l}(0)\right]=\mathrm{E}\left[{Y}_{k, t-1}(0)-{Y}_{k, t-l}(0)\right]$$
for each $(g,k,t)\in\{1,\ldots,G\}\times \{1,\ldots,G\}\times 	\{(l+1),\ldots,T\}$ with $l\in\{2,\ldots, T-1\}$.  We may use this insight to define the following empirical moment comparisons:
\begin{align}
&\widehat{\mathrm{\tau}}_{0,0}(t,l)\notag\\
=&\frac{1}{N_{1,0,t}}\sum_{g: D_{g, t}=1, D_{g, t-1}=0} N_{g,t}\left(Y_{g, t-2}-Y_{g, t-l}\right)-\frac{1}{N_{0,0,t}}\sum_{g: D_{g, t}=0, D_{g, t-1}=0} N_{g,t}\left(Y_{g, t-2}-Y_{g, t-l}\right).\label{tau00}
\end{align}
and
\begin{align}
&\widehat{\mathrm{\tau}}_{1,0}(t,l)\notag\\
=&\frac{1}{N_{1,0,t}}\sum_{g: D_{g, t}=1, D_{g, t-1}=0} N_{g,t}\left(Y_{g, t-1}-Y_{g, t-l}\right)-\frac{1}{N_{0,0,t}}\sum_{g: D_{g, t}=0, D_{g, t-1}=0}N_{g,t} \left(Y_{g, t-1}-Y_{g, t-l}\right).\label{tau10}
\end{align}
While (\ref{tau00}) is a pre-treatment statistic that compares  outcomes from  switching groups with those of not-yet switched groups for periods that pre-date potential misclassification in $t-1$, namely $t-2$ and $t-l$ with $l>2$, (\ref{tau10}) is a pre-treatment statistic that compares those outcomes using period $t-1$ as reference period. To understand  the heuristic, note that because of the assumption that misspecification only occurs in $t-1$, under Assumptions \ref{A1}-\ref{A4}, \ref{A6}-\ref{A8} we have a potential violation of PT as the only ``bias source'':
\begin{align*}
\mathrm{E}\left[\widehat{\mathrm{\tau}}_{0,0}(t,l)\right]=& \mathrm{E} \left[\frac{1}{N_{1,0,t}}\sum_{g: D_{g, t}=1, D_{g, t-1}=0} N_{g,t}\left(Y_{g, t-2}-Y_{g, t-l}\right)-\frac{1}{N_{0,0,t}}\sum_{g: D_{g, t}=0, D_{g, t-1}=0} N_{g,t}\left(Y_{g, t-2}-Y_{g, t-l}\right)  \right]\\
=&\mathrm{E} \left[\frac{1}{{N}_{1,0,t}}\sum_{g: D_{g, t}=1, D_{g, t-1}=0} N_{g,t} \psi_{0,t-2,l} -\frac{1}{N_{0,0,t}}\sum_{g: D_{g, t}=0, D_{g, t-1}=0} N_{g,t} \psi_{0,t-2,l}  \right],
\end{align*}
with
\[
\psi_{0,t-1,l}=\mathrm{E}\left[Y_{g,t-1}(0)-Y_{g,t-l}(0)\right].
\] 
On the other hand,  the expectation of $\widehat{\mathrm{\tau}}_{1,0}^{(l)}(t)$ (provided it exists) is equal to:
\begin{align*}
\mathrm{E}\left[\widehat{\mathrm{\tau}}_{1,0}(1,l)\right]=&\mathrm{E}\left[\frac{1}{N_{1,0,t}} \sum_{g:  D_{g, t}=1, D_{g, t-1}=0}\sum_{i=1}^{N_{g,t}} \frac{N_{g,t}}{N_{g,t-1}}D_{i,g, t-1}^*\left(Y_{i, g, t-1}(1)-Y_{i, g, t-1}(0)\right)\right]\\
&+ \mathrm{E}\left[\frac{1}{N_{1,0,t}}\sum_{g: D_{g, t}=1, D_{g, t-1}=0} N_{g,t} \psi_{0,t-1,l} -\frac{1}{{N}_{0,0,t}}\sum_{g: D_{g, t}=0, D_{g, t-1}=0} N_{g,t} \psi_{0,t-1,l}  \right].
\end{align*}
This suggests basing the tests for pre-treatment PT and misclassification on moment equalities. Thus, define:
\[
\mathrm{\tau}_{0,0}(t,l)\equiv \lim_{G\rightarrow \infty}\mathrm{E}\left[\widehat{\mathrm{\tau}}_{0,0}(t,l)\right]
\]
and accordingly:
\[
\mathrm{\tau}_{1,0}(t,l)\equiv \lim_{G\rightarrow \infty}\mathrm{E}\left[\widehat{\mathrm{\tau}}_{1,0}(t,l)\right] 
\]
for every $t\in\{(l+1),\ldots,T\}$ and $l\in\{2,\ldots, T-1\}$. Specifically, to test for violations of PT, we want to test that:
\[
H_{0}^{PT}: \mathrm{\tau}_{0,0}(t,l)=0
\]
for all $t\in\{(l+1),\ldots,T\}$ and $l\in\{2,\ldots,T-1\}$, versus:
\[
H_{A}^{PT}: \mathrm{\tau}_{0,0}(t,l)\neq 0
\]
for some $t\in\{(l+1),\ldots,T\}$ and $l\in\{2,\ldots,T-1\}$. Similarly, the second null hypothesis about the presence of misclassification or anticipation is given by:
\[
H_{0}^{MC}: \mathrm{\tau}_{1,0}(t,l)=0
\]
for all $t\in\{(l+1),\ldots,T\}$ and $l\in\{2,\ldots,T-1\}$, versus:
\[
H_{A}^{MC}: \mathrm{\tau}_{1,0}(t,l)\neq 0
\]
for some $t\in\{(l+1),\ldots,T\}$ and $l\in\{2,\ldots,T-1\}$. As mentioned previously, the above tests may be implemented as moment equality tests. We outline the test statistic for $H_{0}^{MC}$ against $H_{A}^{MC}$ in what follows, but note that the statistic for $H_{0}^{PT}$ against $H_{A}^{PT}$ is analogous.  Thus, depending on the preferences of the researcher, one may  use a Carm\'{e}r-von-Mises type statistic that is geared towards detecting many small violations of the null hypothesis:
\begin{equation}\label{TEST10S}
T_{1,0}^{s}=\sum_{(t,l)\in\mathcal{C}_{S}}\left(\sqrt{G}\widehat{\mathrm{\tau}}_{1,0}(t,l)\right)^2,
\end{equation}
or a Kolmogorov-Smirnov type max statistic:
\begin{equation}\label{TEST10M}
T_{1,0}^{m}=\max_{(t,l)\in\mathcal{C}_{S}}\left\vert \sqrt{G}\widehat{\mathrm{\tau}}_{1,0}(t,l) \right\vert,
\end{equation}
where $\mathcal{C}_{S}$ is an index set $\mathcal{C}_{S}=\{(l,t):\; t\in\{(l+1),\ldots,T\};\; l\in\{2,\ldots,T-1\}\}$ with cardinality $\mathcal{C}_{S}$. 

\medskip

\noindent \textbf{Remark 4}: The tests above are designed to cover all pre-treatment periods provided there is only a finite number of them. However, in practice, as researchers may not want to use evidence from far away pre-treatment periods, we can restrict the number of pre-treatment periods used in the corresponding moment equality tests. 
\medskip

\noindent \textbf{Remark 5}: If the researcher believes that also other pre-treatment periods are affected by misspecification, one may set a threshold, say $\kappa$, beyond which the researcher has certainty that no misclassification or anticipation is present, so that the PT only test can be carried out accordingly.
\medskip

\noindent \textbf{Remark 6}: Note that the moment equality testing framework may be used more broadly as a flexible and potentially more powerful alternative to standard pre-trend tests in event-study type analyses. That is, while standard tests only evaluate the significance of a (or some) pre-trend DiD parameter(s) relative to a single, fixed pre-treatment reference period, the  statistics in (\ref{TEST10S}) and (\ref{TEST10M}) may be modified to include also comparisons of other pre-treatment period combinations.  For instance, if there are three pre-treatment periods $(t-3)$, $(t-2)$, and $(t-1)$ available in the data, the moment based test may be used to compare not only outcomes from period $(t-1)$ to periods $(t-2)$ and $(t-3)$, but also outcomes from $(t-2)$ to $(t-3)$.  In practice, this may render the moment based test more powerful to detect specific  violations of misclassification or PT in pre-treatment periods than what would be possible if only a fixed reference period was used (see also Section \ref{sec:empirical}).

\medskip

We make the following assumptions:

\begin{assumption}[Existence of Moments] \label{A2A} Let $\lim_{G\rightarrow \infty}\mathrm{E}[\mathbf{U}_{t,l}]$ and $\mathbf{\Sigma}_{\mathbf{U}}$ defined in (\ref{ASYVAR}) of Appendix \ref{AppendixA} exist. Moreover, it holds that $\sup_{(d,g,t)}\mathrm{E}\left[Y_{g,t}(d)^{8+\delta}\right]<+\infty$ for some $\delta>0$.
\end{assumption}

\begin{assumption}[Stable Groups] \label{A3A} For all
$(g, t)\in\{1,\ldots,G\} \times\{2,\ldots,T\}$, we have that 
\[
\Pr\left(D_{g,t} = 1, D_{g,t-1} = 0\right) > 0
\] 
implies $lim_{G\rightarrow\infty} \mathrm{E}\left[\frac{N_{1,0,t}}{G}\right] >
0$ and $lim_{G\rightarrow\infty} \mathrm{E}\left[\frac{N_{0,0,t}}{G}\right] >
0$. Moreover, $lim_{G\rightarrow\infty} \mathrm{E}\left[\frac{N_{1,0,t}}{G}\right] >
0$ for at least one $t\in\{(l+1),\ldots,T\}$ and $l\in\{2,\ldots,T-1\}$. 
\end{assumption}

Since the limiting distribution of our tests under the null hypothesis is non-standard, we use the nonparametric (block) bootstrap to derive critical values. To do so, let
\[
\mathbf{{Y}}_g:=\left({Y}_{g, t}\right)_{t\in \{1, \ldots, T\}}
\]
denote the time series vector of sample averages from each group $g$. Also, recall that:
\[
\mathbf{{D}}_g:=\left({D}_{g, t}\right)_{t\in \{1, \ldots, T\}}
\]
To this end, denote by $\mathbf{Y}_{g}^{b}$ the bootstrapped observations. In addition, we write $o^{b}_{\Pr}(1)$, in probability when for any $\delta>0$, $\Prb\left(|T^{b}|>\delta\right)=o_{\Pr}(1)$ for any bootstrap statistic $T^{b}$. Also, we write $T^{b}\stackrel{d^{b}}{\rightarrow}D$ in probability, if, conditional on the sample, $T^{b}$ weakly converges to $D$ with probability going to one. Finally, let $\mathrm{E}^{b}\left[\cdot\right]$, $\text{Var}^{b}\left[\cdot\right]$,  and $\text{Cov}^{b}\left[\cdot\right]$ denote the expectation, variance, and covariance,  respectively,  under the bootstrap probability measure $\Prb$, conditional on the original sample (defined on a given probability space $(\Omega, \mathcal{F}, \Pr)$).\footnote{The probability law governing the resample is given by the random variables $I_{g}$, $g=1,\ldots,G$,  where each $I_{g}$ is an i.i.d. discrete uniform random variable on the interval $[1,G]$.}  We draw $B$ bootstrap samples $\{\mathbf{Y}_{g}^{b},\mathbf{D}_{g}^{b}\}_{g=1}^{G}$. The bootstrap statistics are given by:
\[
T_{1,0}^{s, b}=\sum_{(t,l)\in\mathcal{C}_{S}}\left(\sqrt{G}(\widehat{\mathrm{\tau}}_{1,0}^{b}(t,l)-\widehat{\mathrm{\tau}}_{1,0}(t,l))\right)^2,
\]
or
\[
T_{1,0}^{m, b}=\max_{(t,l)\in\mathcal{C}_{S}}\left\vert \sqrt{G}(\widehat{\mathrm{\tau}}_{1,0}^{b}(t,l) -\widehat{\mathrm{\tau}}_{1,0}(t,l))  \right\vert . 
\]
Here, $\widehat{\mathrm{\tau}}_{1,0}^{b}(t,l)$ denotes the bootstrap analog of the test statistic:
\[
\widehat{\mathrm{\tau}}_{1,0}^{b}(t,l)=\frac{1}{N_{1,0,t}^{b}}\sum_{g: D^{b}_{g, t}=1, D^{b}_{g, t-1}=0} N_{g,t}\left(Y_{g, t-1}^{b}-Y_{g, t-l}^{b}\right)-\frac{1}{N^{b}_{0,0,t}}\sum_{g: D^{b}_{g, t}=0, D^{b}_{g, t-1}=0} N_{g,t}\left(Y_{g, t-1}^{b}-Y_{g, t-l}^{b}\right).
\]
with:
\[
N^{b}_{1,0,t}=\sum_{g\in G} \mathbb{I}\{ D^{b}_{g, t}=1, D^{b}_{g, t-1}=0\}N_{g,t}
\]
and
\[
{N}^{b}_{0,0,t}=\sum_{g\in G} \mathbb{I}\{ D^{b}_{g, t}=0, D^{b}_{g, t-1}=0\}N_{g,t}
\]
The critical values are then given by the $(1-\gamma)$, $\gamma\in(0,1)$, quantiles of
the empirical bootstrap distributions of $T_{1,0}^{s, b}$ and $T_{1,0}^{m, b}$, respectively, over $B$ draws, say
$c^{m}_{B,G}(1-\gamma)$ and $c^{s}_{B,G}(1-\gamma)$. In what follows, we focus on the test of $H_{0}^{MC}$, as the arguments for the test involving $T_{0,0}^{s}$ and $T_{0,0}^{m}$ are identical. Thus, define the $\vert \mathcal{C}_{S}\vert \times 1$ vector:
\begin{equation}\label{VECHAT}
\widehat{\boldsymbol{\tau}}_{1,0}=\left(\begin{array}{c} \widehat{\tau}_{1,0}(t,l)\\ \vdots\\ \widehat{\tau}_{1,0}(j,k) \end{array}\right)
\end{equation}
and similarly:
\begin{equation}\label{VEC}
\boldsymbol{\tau}_{1,0}=\left(\begin{array}{c} \tau_{1,0}(t,l)\\ \vdots\\ \tau_{1,0}(j,k) \end{array}\right)
\end{equation}
Note that by Assumptions \ref{A2A} and \ref{A3A}, we have with probability approaching 1 that there exists at least one $(t,l)$ for which $\lim_{G\rightarrow \infty}\mathrm{E}\left[P_{1,0}(t,l)\right]>0$ and so $\min\{N_{1,0,t},{N}_{0,0,t}\}>0$ for sufficiently large $G$. On the other hand, observe that for any element $(t,l)\in\mathcal{C}_{S}$ for which  $\lim_{G\rightarrow \infty}\mathrm{E}\left[P_{1,0}(t,l)\right]=0$, the corresponding element does not contribute to the statistic $T_{1,0}^{s}$ for sufficiently large $G$, and does also not contribute to the asymptotic variance of $\sqrt{G}\left(\widehat{\boldsymbol{\tau}}_{1,0}-\mathbf{0}\right)$. We therefore proceed by analyzing the case where all elements of $\sqrt{G}\left(\widehat{\boldsymbol{\tau}}_{1,0}-\mathbf{0}\right)$ contribute to the asymptotic distribution, i.e. the case where $\lim_{G\rightarrow\infty} \mathrm{E}\left[\frac{N_{1,0,t}}{G}\right] >
0$ and $\lim_{G\rightarrow\infty} \mathrm{E}\left[\frac{N_{0,0,t}}{G}\right] >
0$ hold for all $(t,l)\in\mathcal{C}_{S}$  in Assumption \ref{A3A}.  We have the following statement:

\begin{proposition}\label{PROPTEST} Assume that Assumptions \ref{A1}-\ref{A2}, \ref{A6}-\ref{A7}, and \ref{A15}, \ref{A2A}, \ref{A3A} hold for all $(t,l)\in\mathcal{C}_{S}$, and that $\boldsymbol{\Sigma}_{\mathbf{U}}$ from (\ref{ASYVAR}) is positive definite, and $\mathbf{J}_{f}	\left(\mathbf{U}^{\infty}\right)$ defined in (\ref{Jacobian}) has full column rank.  Then:

(i) Under $H_{0}^{MC}$, it holds that:
\[
\lim_{G,B\rightarrow \infty} \Pr\left(T_{1,0}^{s}>c^{s}_{B,G}(1-\alpha)\right)=\alpha
\]
and
\[
\lim_{G,B\rightarrow \infty} \Pr\left(T_{1,0}^{m}>c^{m}_{B,G}(1-\alpha)\right)=\alpha.
\]

(ii) Under $H_{A}^{MC}$, it holds that:
\[
\lim_{G,B\rightarrow \infty} \Pr\left(T_{1,0}^{s}>c^{s}_{B,G}(1-\alpha)\right)=1
\]
and 
\[
\lim_{G,B\rightarrow \infty} \Pr\left(T_{1,0}^{m}>c^{m}_{B,G}(1-\alpha)\right)=1.
\]

\end{proposition}

Proposition \ref{PROPTEST} establishes the asymptotic size control of the moment based tests under $H_{0}^{MC}$, and their consistency against fixed alternatives.  As remarked previously, identical results also hold for the test of $H_{0}^{PT}$ against $H_{A}^{PT}$. Moreover, note that the asymptotic analysis of these nonparametric moment tests relies on less and substantially weaker assumptions than the analysis of the bias and the estimators, respectively. Specifically, apart from the condition about the timing of misclassification (Assumption \ref{A7}) to make the tests operational, we do not impose any restrictions on the form of misclassification (e.g., homogeneity in Assumption \ref{A13}), strong exogeneity (Assumption \ref{A4}), or PT (Assumption \ref{A5}). The tests therefore allow to assess misspecification under a substantially weakened set of assumptions.  

Finally, even though both tests may of course be carried out separately, we can also combine them to construct a decision rule that provides guidance for practicioners. In particular, there are three cases we would like to distinguish as they bear different implications for the empirical analysis (see above): a violation of PT in pre-treatment periods, possibly with misclassification ($H_{TV}$), PT in pre-treatment periods holding without misclassification ($H_{PTN}$), and  PT  in pre-treatment periods holding with misclassification ($H_{PTM}$). We have that:
\[
H_{TV}=H_{A}^{PT}
\]
\[
H_{PTN}=H_{0}^{PT}\cap H_{0}^{MC}
\]
\[
H_{PTM}=H_{0}^{PT}\cap H_{A}^{MC}
\]
The decision rule is as follows:\medskip

\noindent \textbf{1.} Test $H_{0}^{PT}$ vs. $H_{A}^{PT}$ at level $\alpha$. If $H_{0}^{PT}$ is rejected at level $\alpha$, stop and decide in favor of $H_{TV}$. If $H_{0}^{PT}$ is not rejected, proceed with step 2.\medskip

\noindent \textbf{2.} Test $H_{0}^{MC}$ vs. $H_{A}^{MC}$ at level $\gamma$. If $H_{0}^{MC}$ is not rejected at level $\gamma$, decide in favor of $H_{PTN}$ and use standard estimator. If $H_{0}^{MC}$ is rejected, decide in favor of $H_{PTM}$ and modify estimator.
\medskip

Even though the formal analysis of this decision rule is beyond the scope of the paper, the heuristic is as follows: if there is pre-treatment evidence suggesting a likely violation of PT, DID should generally not be used for point identification and estimation. On the other
hand, if there is pre-treatment evidence supporting PT, we may actually proceed with DID and seek to design the estimator in such a way to avoid misspecification bias. Put differently, we aim to ``maximize'' our confidence in the PT assumption in a first testing step by looking at pre-treatment evidence in favor of PT, before turning to a misspecification test  in a  second step to obtain guidance on how to manipulate the estimator to avoid bias. In fact, building on pre-treatment evidence for PT is in line with recent work by \citet{Rambachan2023} who exploit pre-treatment PT violations as a way to infer about possible post-treatment violations. Moreover, the decision rule is also motivated by findings in \citet{CH2024} who show that conducting pre-treatment specification tests generally leads to, under the null of correct specification,  conservative, but not too liberal post-treatment inference. Applied to our context, we therefore conjecture that post-treatment inference carried out conditional on passing the above pre-treatment specification tests may also maintain asymptotic size control under the null of correct specification.

\section{Simulations}\label{sec:MC}

In order to assess the finite sample performance of the estimators and tests introduced in the previous sections, we conduct a Monte Carlo study. We start by comparing the two proposed estimators $\widehat{\mathrm{DID}}^{S}$ and $\widehat{\mathrm{DID}}^{S^*}$ to an estimator that does not account for misspecification in the treatment indicator. In a second step, we then analyze the rejection rates of the specification tests under various Data Generating Processes (DGP).

Since our asymptotic theory operates at the group level, we generate group level outcomes  instead of modeling individual outcomes and taking group averages. We use 2,000 Monte Carlo  replications and let the total number of groups be $G\in \{100,\ 300,\ 600\}$ across 15 time periods  ($T=15$). The first treatment period for each group is randomly drawn between time periods $t=2$ and $t=15$. A fraction of 5\% of the groups remain untreated throughout the panel. The DGP is chosen to be: 
\[
  Y_{g,t}(D_{g,t}^* ) = 10 - 0.4 \, t + 0.1 \, g + \Delta_t \, \mathbf{1}\{D_{g,t}^*\geq 1\} + e_{g,t},\;  e_{g,t} \sim \mathcal{N} (0,1),
 \]
where $\Delta_{t}$ is the treatment effect, while $e_{g,t}$ denotes the idiosyncratic error term. We simulate both time constant and time-varying treatment effects. Specifically, the constant treatment effect  $\Delta$ is set to 4 and the time varying treatment effects are generated by
\[
 \Delta_t = \Delta \, \left(0.2 + (1.8-0.2) \, \frac{(t-1)}{(T-1)}\right),
\] 
so that the mean treatment effect increases from 0.8 to 7.2 over time.  

While the DGP is based on the true treatment status $D_{g,t}^*$, misclassification is induced via $D_{g,t}$, the recorded group-period treatment status. In line with the simplifying assumption from Sections \ref{sec:IdentTrueSwitch} and \ref{sec:testing}, $D_{g,t}$ equals $D_{g,t}^*$ for all time periods but the period before treatment, in which 50\% of the groups are assigned $D_{g,t}=0$, while $D_{g,t}^*=1$. Groups that enter treatment in the last time period are not misclassified.

Table \ref{tab:estimator_simulation_errors} and Table \ref{tab:estimator_simulation_RMSE} display the simulation results for the estimators presented in Section \ref{sec:IdentTrueSwitch} in comparison to the estimator of \citet{CH2020} that does not account for misspecification in the treatment indicator. As $\widehat{\mathrm{DID}}^{S}$ and $\widehat{\mathrm{DID}}^{S^*}$ recover different target parameters, Table \ref{tab:estimator_simulation_errors} contains the mean bias with respect to $\bar{\Delta}_S*\equiv 1/13 \sum_{t=2}^{14} \Delta_{t}=4$ and $\bar{\Delta}_S\equiv 1/12 \sum_{t=3}^{14} \Delta_{t}=4.2287$, respectively.\footnote{Note that in population the benchmark ATE $\bar{\Delta}_S$ coincides with the ATE of true switchers since misclassification and switching probabilities are constant over time in the DGP.} Under constant treatment effects, both parameters are equivalent, i.e. $\Delta= \bar{\Delta}_S*=\bar{\Delta}_S$. Turning to the results, the mean biases in the third and fourth columns of Table \ref{tab:estimator_simulation_errors} are not close to zero reflecting that $\widehat{\mathrm{DID}}$ cannot match either of the ATEs in the presence of misclassification. This is due to some actually treated groups (individuals) entering the comparison group that is supposed to be non treated. If treatment effects change over time, the inclusion of already treated groups in the reported pre-treatment periods is also an issue. As the fifth column shows, the estimator for observed switchers $\widehat{\mathrm{DID}}^{S}$ captures the average treatment effect between time periods 3 to 14. The mean bias of the estimator for true switchers $\widehat{\mathrm{DID}}^{S^*}$ is close to zero in all cases and decreases as the total number of groups increases. Since this estimator captures anticipatory treatment effects, treatment effects in the second time period are also included.  

Table \ref{tab:estimator_simulation_RMSE} depicts the Root Mean Squared Error (RMSE). For the $\widehat{\mathrm{DID}}$ estimator, the RMSE is larger under constant treatment effects than under time varying treatment effects (see column 3 and 4), which reflects the same pattern seen for the mean bias. For $G=100$ and $G=300$ the RMSE of the estimator for true switchers $\widehat{\mathrm{DID}}^{S^*}$ is larger than the RMSE of the estimator for observed switchers $\widehat{\mathrm{DID}}^{S}$. The estimator for true switchers $\widehat{\mathrm{DID}}^{S^*}$ produces a larger RMSE under time-varying treatment effects than under constant treatment effects. This is to be expected since $\widehat{\mathrm{DID}}^{S^*}$ requires the estimation of unobserved switching probabilities.

\begin{table}[H]
\centering
\begin{tabular}{llllll} 
\toprule
\toprule
\multirow{2}{*}{\textbf{\# Groups}} & \multirow{2}{*}{\textbf{DGP TE}} & \multicolumn{4}{c}{\textbf{Mean Bias}} \\
    & & $\bf{\widehat{\mathrm{DID}}}$ & $\bf{\widehat{\mathrm{DID}}}$ & $\bf{\widehat{\mathrm{DID}}^{S}}$ & $\bf{\widehat{\mathrm{DID}}^{S^*}}$  \\
                        
\midrule
\multirow{2}{*}{G\ =\ 100} & Constant TE  & -2.3324  & -2.3324 & -0.0040 & -0.0296\\ 
                           & Time-vary. TE & -2.0799 & -2.3086 & -0.0015 & -0.0373\\  
                        
\midrule
\multirow{2}{*}{G\ =\ 300} & Constant TE &  -2.3402  & -2.3402 & 0.0010 & -0.0038\\  
                           & Time-vary. TE &-2.0895  & -2.3181 & 0.0007 & -0.0165\\ 
                       
\midrule
\multirow{2}{*}{G\ =\ 600} & Constant TE   & -2.3406 & -2.3406 & 0.0011 & -0.0007\\
                           & Time-vary. TE & -2.0895 & -2.3182 & 0.0024 & -0.0062\\ 
                        
\bottomrule
\bottomrule
\end{tabular}
\caption{\small Results of 2,000 Monte Carlo replications on the Mean Bias for 100, 300 and 600 groups are presented. Columns 3 and 4 refer to the \citet{CH2020} estimator $\widehat{\mathrm{DID}}$ in comparison to $\bar{\Delta}_S*$ and $\bar{\Delta}_S$, respectively. Column 5 refers to the DID estimator of observed switchers $\widehat{\mathrm{DID}}^{S}$ in comparison to $\bar{\Delta}_S$, and column 6 to the DID estimator of true switchers $\widehat{\mathrm{DID}}^{S^*}$ in comparison to $\bar{\Delta}_S*$. Note that $\bar{\Delta}_S$ and $\bar{\Delta}_{S^*}$ are identical under constant treatment effects.}
\label{tab:estimator_simulation_errors}
\end{table}

\begin{table}[H]
\centering
\begin{tabular}{llllll} 
\toprule
\toprule
\multirow{2}{*}{\textbf{\# Groups}} & \multirow{2}{*}{\textbf{DGP TE}} & \multicolumn{4}{c}{\textbf{RMSE}} \\
    & & $\bf{\widehat{\mathrm{DID}}}$ & $\bf{\widehat{\mathrm{DID}}}$ & $\bf{\widehat{\mathrm{DID}}^{S}}$ & $\bf{\widehat{\mathrm{DID}}^{S^*}}$  \\
                        
\midrule
\multirow{2}{*}{G\ =\ 100} & Constant TE & 2.3527 & 2.3527 & 0.1908 & 0.2904 \\
                           & Time-vary. TE & 2.1060 & 2.3321 & 0.2570 & 0.7409 \\ 
                        
\midrule
\multirow{2}{*}{G\ =\ 300} & Constant TE & 2.3468   & 2.3468 & 0.1087 & 0.1067 \\  
                           & Time-vary. TE & 2.0979 & 2.3257 & 0.1503 & 0.1956 \\
                       
\midrule
\multirow{2}{*}{G\ =\ 600} & Constant TE &  2.3441  & 2.3441 & 0.0766 & 0.0743 \\ 
                           & Time-vary. TE & 2.0938 & 2.3221 & 0.1060 & 0.1112 \\  

\bottomrule
\bottomrule
\end{tabular}
\caption{ \small Results of 2,000 Monte Carlo replications on the RMSE for 100, 300 and 600 groups are presented.   Columns 3 and 4 refer to the \citet{CH2020} estimator $\widehat{\mathrm{DID}}$ in comparison to $\bar{\Delta}_S*$ and $\bar{\Delta}_S$, respectively. Column 5 refers to the DID estimator of observed switchers $\widehat{\mathrm{DID}}^{S}$ in comparison to $\bar{\Delta}_S$, and column 6 to the DID estimator of true switchers $\widehat{\mathrm{DID}}^{S^*}$ in comparison to $\bar{\Delta}_S*$. Note that $\bar{\Delta}_S$ and $\bar{\Delta}_{S^*}$ are identical under constant treatment effects.}
\label{tab:estimator_simulation_RMSE}
\end{table}

Next, we turn to the specification tests from Section \ref{sec:testing}. We make a few minor adjustments to the simulation framework presented above. The number of time periods is decreased to $T=10$ since the asymptotics for $T^m_{1,0}$ and $T^s_{1,0}$  rely on a small to moderate number of elements $\hat{\tau}_{1,0}(t,l)$ relative to the sample size $G$. Specifically, $\hat{\tau}_{0,0}(t,l)$ and $\hat{\tau}_{1,0}(t,l)$ are defined for $t\in\{(l+1),\ldots,T\}$ and $l\in\{2,\ldots,T-1\}$. This results in a cardinality of $\vert \mathcal{C}_{S} \vert=28$, i.e. 28 elements enter the statistics $T^m_{1,0}$ and $T^s_{1,0}$. The specification tests are designed to detect a violation of parallel trends in Assumption \ref{A5} and misclassification. Therefore, the data generating process entails 4 scenarios now:\medskip

\noindent \textbf{1.} Parallel Trends in pre-treatment periods and misclassification (PTM)\medskip

\noindent \textbf{2.} Trend violation  in pre-treatment periods  and misclassification (TVM)\medskip

\noindent \textbf{3.} Parallel Trends  in pre-treatment periods  and no misclassification (PTN)\medskip

\noindent \textbf{4.} Trend violation  in pre-treatment periods  and no misclassification (TVN)\medskip

While misclassification is simulated as before, a violation of the PT assumption  in pre-treatment periods  is simulated by adding a normalized time and group dependent trend: $Y_{g,t}^{TV}(D_{g,t}^*)=Y_{g,t}(D_{g,t}^*)+ \frac{t\cdot g}{G}$. For simplicity, the treatment effect is kept constant over time now. The bootstrap is carried out as described in Section \ref{sec:testing}, and we use the 95-th percentile of the bootstrap distribution as critical value. To decrease the computing time we make use of the warp bootstrap method \citep[][]{Giacomini2013} and only draw a single bootstrap subsample within each Monte Carlo replication, resulting in 2,000 bootstrap samples. 

\noindent Table \ref{tab:tests_simulation} summarizes the rejection rates of the 4 tests for the different DGPs and group sizes. Column 3 is based on PT and misclassification (PTM), meaning that the rejection rates for the PT test in pre-treatment periods  should converge to 0.05, while for the MC test we ought to observe a convergence towards one. Examining the results, for $G=100$, we see that the PT and MC tests underreject relative to the nominal level, even though the underrejections become smaller as $G$ grows.  That said, even for $G= 600$,  the PT tests still slightly underreject in the presence of parallel trends (both with and without misclassification). Power results are generally favorable even for sample sizes as small as $G=100$. Note that the MC test rejects in the absence of misclassification and in the presence of parallel trends, since both patterns are not distinguishable when parallel trends are violated.

\begin{table}[ht!]
\centering
\begin{tabular}{llllll} 
\toprule
\toprule
$\textbf{\# Groups}$ & $\textbf{Tests}$ & $\textbf{H}_\text{PTM}$ & $\textbf{H}_\text{TVM}$ & $\textbf{H}_\text{PTN}$ & $\textbf{H}_\text{TVN}$ \\
\midrule
\multirow{4}{*}{G\ =\ 100} & Sum PT & 0.0300  & 0.9235 & 0.0295 & 0.8345\\
                        & Max PT & 0.0155 & 0.9215 & 0.0145 & 0.8015 \\
                        & Sum MC & 0.9225 & 0.8945 & 0.0360 & 0.9145 \\
                        & Max MC & 0.6655& 0.6010 & 0.0210 & 0.8975 \\
\midrule
\multirow{4}{*}{G\ =\ 300} & Sum PT & 0.0355 & 1 & 0.0420 & 1 \\
                        & Max PT & 0.0300 & 1 & 0.0475 & 1\\
                        & Sum MC &  0.9960 & 1 & 0.0320 & 1\\
                        & Max MC &  0.9795 & 1 & 0.0250 & 1\\
\midrule
\multirow{4}{*}{G\ =\ 600} & Sum PT & 0.0410 & 1 & 0.0450 & 1 \\
                        & Max PT & 0.0375 & 1 & 0.0300 & 1 \\
                        & Sum MC & 0.9985 & 1 & 0.0445 & 1 \\
                        & Max MC &  0.9935 & 1 & 0.0400 & 1\\
\bottomrule
\bottomrule
\end{tabular}
\caption{ \small Rejection rates of the four tests for different group numbers and for the four scenarios: PTM (parallel trends and misclassification), TVM (trend violation and misclassification), PTN (parallel trends and no misclassification) and TVN (trend violation and no misclassification). }
\label{tab:tests_simulation}
\end{table}

\section{Empirical Illustration}\label{sec:empirical}

To illustrate the proposed estimators for true and observed switchers as well as the testing procedure, we revisit the application of \citet{BPRSS2024}. The authors analyze how the introduction of a computer based testing system designed to curb cheating in high stakes national school examinations affected exam scores in Indonesia. As reported in their paper, cheating was a widespread problem in the nation-wide examinations at the end of junior and senior secondary school. Therefore, the Indonesian government gradually implemented computer based testing (CBT) instead of paper based exams, which made cheating highly unlikely and led to a drop in exam scores. \citet{BPRSS2024} point to anticipatory behavior by schools aware of their future treatment status, as treatment follows voluntary sign-up for CBT. That is, in order to soften the drop in exam scores, schools might try to counteract cheating practices prior to their actual treatment. The empirical findings of \citet{BPRSS2024} support the presence of such anticipatory behavior one year before the CBT implementation. 

For the empirical analysis, we follow \citet{BPRSS2024} and use a balanced panel data set on exam results between 2010 and 2019. CBT was introduced between 2015-2019, but the 2015 and 2016 cohorts have been excluded, since their characteristics are different from the comparison groups. The number of observations is thus 38,621 schools for 10 years, see the original paper for further details on the data.

\begin{table}[ht!]
    \centering
    \begin{tabular}{l|c|c|c|c|c|c}
    \toprule
    \toprule
    \multirow{2}{*}{\textbf{Cohort}} & \multicolumn{5}{c}{\textbf{Estimators}}&\\
    & Original & C\&S No Ant.& C\&H No Ant. & $\widehat{\mathrm{DID}}^{S}_t$ & $\widehat{\mathrm{DID}}^{S^{\ast}}_t$ & $\widehat{\lambda}_{t-1}$ \\
    \midrule
    Cohort 2017  &  $-5.21^{\ast}$ & $-3.98^{\ast}$ & $-3.01^{\ast}$ & $-4.87^{\ast}$ & / & / \\
    Cohort 2018  &  $-7.45^{\ast}$ & $-5.61^{\ast}$ & $-4.48^{\ast}$ & $-7.16^{\ast}$ & $-7.83^{\ast}$ & $0.158$ \\
    Cohort 2019  &  $-5.69^{\ast}$ & $-3.78^{\ast}$ & $-3.78^{\ast}$ & $-5.69^{\ast}$ & $-5.91^{\ast}$ & $0.297$ \\
    Combined     &  $-6.30^{\ast}$ & $-4.59^{\ast}$& $-3.87^{\ast}$ & $-6.08^{\ast}$ & $-7.00^{\ast}$ & / \\
    \bottomrule
    \bottomrule
    \end{tabular}
    \caption{Comparison of estimators for ATEs in the cohorts 2017, 2018, 2019 and a weighted combination of the cohorts. The star reflects significance at the 5\% level.}
    \label{tab:Application_estimation}
\end{table}

Table \ref{tab:Application_estimation} reports the treatment effect estimates from the original paper of \citet[][Figure 5 and Table A3]{BPRSS2024}, two estimators not accounting for anticipation and our two proposed estimators, as well as the estimated anticipation probabilities. \citet{BPRSS2024} use the  staggered adoption estimator of \citet{CSA2021} where never treated schools serve as the control group and one year of anticipation is allowed for (Original). For comparison, we also present the estimator of \citet{CSA2021} without accounting for anticipation (C\&S No Ant.) and the \citet{CH2020} estimator that also does not allow for anticipation (C\&H No Ant.). We report both the estimates of the cohort specific effects and of the average effects weighted by the number of groups in each cohort. All estimated ATEs are statistically significant at the 5\% level.\footnote{We construct critical values for $H_{0}:\; ATE=0$ against its negation using the bootstrap as described in Section \ref{sec:testing} with $B=999$.} 

The estimates from the ATE estimator of observed switchers ($\widehat{\mathrm{DID}}^{S}_t$) and from the ATE estimator of True Switchers ($\widehat{\mathrm{DID}}^{S^{\ast}}_t$) are substantially larger than those obtained from estimators that do not account for anticipation, underscoring the importance of taking anticipation into account. Specifically, we see for the 2018 cohort a drop in mean exam scores by 7.50 and 7.83 points, respectively (scores range from 0 to 100). In comparison, not accounting for anticipation yields smaller estimated declines of 5.61 (C\&S No Ant.) or 4.48 points (C\&H No Ant.). The estimates based on \citet{CH2020} are even smaller than those from \citet{CSA2021} without accounting for anticipation, since the bias not only results from the treated groups, but also from the comparison groups. Additionally, we estimate the anticipation probability ($\widehat{\lambda}_{t-1}$) to be 15.8\% for the 2018 cohort, and 29.7\% for the 2019 cohort complementing the picture of non-negligible anticipatory behavior among schools.

The estimates from the ATE estimator of observed switchers ($\widehat{\mathrm{DID}}^{S}_t$) are very similar to the ATE estimates in the original study. Conceptually, this is to be expected since the estimator of \citet{CSA2021} should also recover the ATE of observed switchers as the bias from anticipation in the treatment group is offset by dropping the first pre-treatment time period, and the bias from anticipation in the comparison group is avoided by using never treated groups as a comparison. The small numerical differences between our estimator and the estimator of \citet{CSA2021} are likely to stem from the larger comparison group used in our estimator $\widehat{\mathrm{DID}}^{S}_t$. 

The cohort specific estimates for the ATE of true switchers from $\widehat{\mathrm{DID}}^{S^{\ast}}_t$ on the other hand are somewhat larger in absolute terms than in the original study, possibly capturing not only the effect of CBT, but also the effect of anti-cheating measures implemented prior to the start of CBT.\footnote{Note that the estimate for $\widehat{\mathrm{DID}}^{S^{\ast}}_{2017}$ is missing, since an already treated comparison group is lacking.} Moreover, since the size of the estimated treatment effects decline in magnitude with the post-treatment time periods as reported in \citet{BPRSS2024}, it is reasonable that $\widehat{\mathrm{DID}}^{S^{\ast}}_t$ are larger in magnitude as they reflect the (average) treatment effect at the time when schools actually switched behavior. Note that the combined estimate for the ATE of true switchers should be interpreted with caution, since it is a weighted average of two cohorts only which display larger treatment effects than the first cohort. 

Next and for illustration purposes, the testing framework from Section \ref{sec:testing} is applied to test for the presence of parallel pre-trends and anticipatory behavior. We start by carrying out the tests using never treated schools as control groups similar to \citet{BPRSS2024}. Moreover, tests are conducted separately for the three treatment cohorts 2017, 2018 and 2019, where we use 499 bootstrap draws to obtain (5\%) critical values. Table \ref{tab:Application_testing} displays the test results for the null hypotheses of parallel pre-trends, $H_{0}^{PT}$, and of no misclassification (no anticipation), $H_{0}^{MC}$, for each cohort using one comparison time period.\footnote{For instance, the PT test for the 2017 cohort compares the outcomes of schools that are treated in 2017 to schools that are never treated by differencing between their 2015 and 2014 outcomes. The 2016 outcome is discarded since it could be affected by anticipatory behavior. The misclassification test then compares the same schools and their outcomes in 2015 and 2016.} 

In line with \citet{BPRSS2024}, $H_{0}^{PT}$ is rejected at the 5\% level for the 2017 cohort and not rejected for the 2018 cohort. However, while there is no evidence for a violation of parallel pre-trends in \citet{BPRSS2024} for the 2019 cohort, our test rejects the null here $H_{0}^{PT}$ suggesting a possible violation also in that cohort. Similarly, 
the misclassification tests based on both the sum and the max statistic reject $H_{0}^{MC}$ in the pre-treatment time period for all three cohorts, while the results in \citet{BPRSS2024} only indicate anticipatory effects for the 2018 and 2019 cohorts only. All our test results remain unchanged when using 1\% or 10\% significance levels instead. 

Finally, unreported results also show that constructing the test statistics  using not yet treated (rather than just never treated) schools as comparison groups lead to the same rejection decisions with one exception: The PT test for the 2018 cohort rejects at the 5\% level, but does not reject at the 1\% level. This suggests that the PT assumption might be fragile, at least in periods just prior to the anticipation period. Additional evidence for this interpretation comes from PT tests comparing each cohort to never-treated schools and using three comparison periods instead of one\footnote{For instance, the PT test for the 2017 cohort compares the outcomes of schools that are treated in 2017 to schools that are never treated by differencing between their 2015 and 2014 outcomes, 2015 and 2013 outcomes, and 2015 and 2012 outcomes.}. The results yield rejection decisions across all three cohorts, suggesting increased power when multiple moments are incorporated.

\begin{table}[ht!]
    \centering
    \begin{tabular}{l l |c c|c c }
    \toprule
    \toprule
    \multirow{2}{*}{\textbf{Cohort}} & & \multicolumn{2}{c|}{$\mathbf{H_0^{PT}}$}  & \multicolumn{2}{c}{$\mathbf{H_0^{MC}}$}   \\
    & & Sum PT & Max PT & Sum MC & Max MC \\
    \midrule
    \multirow{2}{*}{Cohort 2017}  & Statistic  & 389503.5 & 624.1 & 59139.8 & 243.2 \\
                                  & Critical Val. & 3966.5 & 63.0 & 4079.1 & 63.9 \\
                                  \midrule
    \multirow{2}{*}{Cohort 2018}  & Statistic & 1737.1 & 41.7 & 131859.8 & 363.1 \\
                                  & Critical Val. & 3553.6 & 59.6 & 2680.8 & 51.8 \\
                                  \midrule
    \multirow{2}{*}{Cohort 2019}  & Statistic & 9536.7 & 97.7 & 140575.9 & 374.9 \\
                                  & Critical Val. & 2712.9 & 52.1 & 3134.7 & 56.0 \\
    \bottomrule
    \bottomrule
    \end{tabular}
    \caption{Tests for parallel pre-trends ($H_0^{PT}$) and no misclassification ($H_0^{MC}$) for cohorts 2017, 2018, and 2019 using the sum and max statistics.}
    \label{tab:Application_testing}
\end{table}

\section{Conclusion}\label{sec:conclusion}

This paper investigates the consequences of treatment misclassification and anticipation in staggered adoption Difference-in-Differences designs, showing that these forms of misspecification can induce  bias  in commonly used estimators such as the TWFE estimator and other estimators suitable for treatment effect heterogeneity. We characterize the structure of these biases, and propose bias-corrected estimators that recover well-defined causal parameters under misclassification and anticipation, namely the ATE for observed and actual switching-to-treatment units.

Additionally, we introduce two specification tests based on moment equalities to detect violations of pre-treatment PT and of the timing and extent of misclassification. This provides a diagnostic framework for empirical researchers to assess the validity of DiD designs under staggered adoption. 

To assess the finite sample performance of our estimators and the moment equality tests, we conduct a small scale Monte Carlo study. Our simulation results illustrate the favorable finite-sample performance of both the proposed estimators and the tests, confirming their robustness under misspecification.

We illustrate the methods revisiting the introduction of a CBT system designed to curb cheating behavior in high stakes national school examinations on exam scores in Indonesia. The results indicate that a sizable fraction of schools ($\approx 15-30\%$ depending on the cohort) has likely already anticipated measures against cheating in the year prior to the official CBT implementation. Moreover, our estimates for the ATE of observed and actual switching-to-treatment units suggest slightly more pronounced drops in mean exam scores when schools react to cheating for the first time relative to when they have already been exposed to these measures.  Our specification tests largely support the evidence for anticipation, though they also hint at a violation of parallel pre-trends in periods just prior to those anticipation periods.

\section{Appendix}

\subsection{Appendix: Additional Quantities}\label{AppendixA}

To study the asymptotic behavior of the estimator $\widehat{\mathrm{DID}}_{t}^{S^{\ast}}$ in (\ref{DIDESTIMATORDAG}), we require some additional notation. To this end, paralleling the notation in \citet{CH2020}, for $(d,d',t)\in\{0,1\}\times \{0,1\}\times \{2,\cdots,T\}$, let
\begin{align*}
	P_{d,d',t}&:=\frac{1}{G} \sum_{g \in G} \mathbb{I}\{D_{g,t} = d, D_{g,t-1} = d'\} N_{g,t},\\
	Q_{d,d',t}&:=\frac{1}{G} \sum_{g \in G} \mathbb{I}\{D_{g,t} = d, D_{g,t-1} = d'\} N_{g,t}\left(Y_{g,t} - Y_{g,t-1}\right).
\end{align*}
We also define:
\begin{align*}
	P_{d,d',t-1}^{[1]}&:=\frac{1}{G} \sum_{g \in G} \mathbb{I}\{D_{g,t} = d, D_{g,t-1} = d'\} N_{g,t-1},\\
	Q_{d,d',t}^{*}&:=\frac{1}{G} \sum_{g \in G} \mathbb{I}\{D_{g,t} = d, D_{g,t-1} = d'\} N_{g,t}\left(Y_{g,t-1} - Y_{g,t-2}\right),\\
	Q_{d,d',t}^{**}&:=\frac{1}{G} \sum_{g \in G} \mathbb{I}\{D_{g,t+1} = d, D_{g,t} = d'\} N_{g,t}\left(Y_{g,t} - Y_{g,t-1}\right),\\
	Q_{d,d',t-1}^{***}&:=\frac{1}{G} \sum_{g \in G} \mathbb{I}\{D_{g,t} = d, D_{g,t-1} = d'\} N_{g,t-1}\left(Y_{g,t} - Y_{g,t-1}\right),
\end{align*}
where we set \(Q_{d,d',2}^{*}=0\) and \(Q_{d,d',T}^{**}=0\).
We collect the following quantities into a vector, which will be relevant for later use:
\begin{align}\label{VVECTOR}
	 \mathbf{V}:=\left(\mathbf{V}_2',\cdots,\mathbf{V}_T'\right)',
\end{align}
where
\begin{align*}
	\mathbf{V}_t &:= \left( 
P_{1,0,t}, 
P_{0,0,t}, 
P^{[1]}_{1,0,t-1}, 
P^{[1]}_{0,0,t-1}, 
P^{[1]}_{1,1,t-1}, 
Q_{1,0,t}, 
Q_{0,0,t}, 
Q^{*}_{1,0,t}, 
Q^{*}_{0,0,t}, 
Q^{**}_{1,0,t}, 
Q^{**}_{0,0,t}, 
Q^{***}_{1,0,t-1}, 
Q^{***}_{1,1,t-1}, 
\right)'.
\end{align*}
and denote its asymptotic variance by:
\begin{equation}\label{ASYVARTS}
\mathbf{\Sigma}_{\mathbf{V}}:= \lim_{G\rightarrow \infty}\text{Var}\left[\sqrt{G} \ \mathbf{V}\right],
\end{equation}

For the specification test, we define a similar set of quantities. Specifically, let
\[
Q_{0,0}(t,l)\equiv \frac{1}{G}\sum_{g\in G} \mathbb{I}\{D_{g, t}=0, D_{g, t-1}=0\} N_{g,t}\left(Y_{g, t-1}-Y_{g, t-l}\right),
\]
\[
P_{0,0}(t,l)\equiv \frac{1}{G}\sum_{g\in G} \mathbb{I}\{D_{g, t}=0, D_{g, t-1}=0\} N_{g,t}.
\]
Accordingly, define:
\[
Q_{1,0}(t,l)\equiv \frac{1}{G}\sum_{g\in G} \mathbb{I}\{D_{g, t}=1, D_{g, t-1}=0\} N_{g,t}\left(Y_{g, t-1}-Y_{g, t-l}\right),
\]
as well as
\[
P_{1,0}(t,l)\equiv \frac{1}{G}\sum_{g\in G} \mathbb{I}\{D_{g, t}=1, D_{g, t-1}=0\} N_{g,t}.
\]
Next, we let:
\[
\mathbf{U}(t,l)=\left(\begin{array}{c}P_{0,0}(t,l)\\Q_{0,0}(t,l)\\P_{1,0}(t,l)\\Q_{1,0}(t,l) \end{array}\right),
\]
while $\mathbf{U}$ denotes the vector of stacked $\mathbf{U}(t,l)$ for all combinations $(t,l)$ in $\mathcal{C}_{S}$, which is of dimension $\vert \mathcal{C}_{S}\vert \times 1$. We also define the asymptotic variance:
\begin{equation}\label{ASYVAR}
\mathbf{\Sigma}_{\mathbf{U}}\equiv \lim_{G\rightarrow \infty}\text{Var}\left[\sqrt{G} \ \mathbf{U}\right],
\end{equation}
a $(4\cdot\vert \mathcal{C}_{S}\vert) \times(4\cdot \vert \mathcal{C}_{S}\vert)$ matrix given by:
\[
\boldsymbol{\Sigma}_{\mathbf{U}}\equiv \left(\begin{array}{ccc}\boldsymbol{\Sigma}_{(t,l)}&\ldots & \boldsymbol{\Sigma}_{(t,l),(j,k)}\\ \vdots &\ddots & \vdots \\ \boldsymbol{\Sigma}_{(j,k),(t,l)} &\ldots &\boldsymbol{\Sigma}_{(j,k)} 
\end{array}\right)
\] 
with elements:\footnotesize
\begin{align*}
&\boldsymbol{\Sigma}_{(t,l)}\\
=&\left(\begin{array}{cccc} \text{Var}\left[P_{0,0}(t,l)\right] &\text{Cov}\left[P_{0,0}(t,l),Q_{0,0}(t,l)\right] &0 &0\\
\text{Cov}\left[P_{0,0}(t,l),Q_{0,0}(t,l)\right] &\text{Var}\left[Q_{0,0}(t,l)\right]  &0&0\\
0 &0& \text{Var}\left[P_{1,0}(t,l)\right] & \text{Cov}\left[P_{1,0}(t,l),Q_{1,0}(t,l)\right]\\
0&0&\text{Cov}\left[P_{1,0}(t,l),Q_{1,0}(t,l)\right] &  \text{Var}\left[Q_{1,0}(t,l)\right]
 \end{array}\right)
\end{align*}\normalsize

\subsection{Appendix: Proofs}\label{AppendixB}

\begin{proof}[Proof of Proposition \ref{PROPCHARACTERIZATION}]
First, for Assumption \ref{A3}, we have:
\begin{align*}
&\quad \mathrm{E}\left[Y_{g, t}(0) \mid \mathbf{D}\right]-\mathrm{E}\left[Y_{g, t}(0) \mid \mathbf{D}_g\right]\\
	&=\mathrm{E}\left[\mathrm{E}\left[Y_{g, t}(0) \mid \mathbf{D}^*,\mathbf{D}\right]\mid \mathbf{D}\right]-\mathrm{E}\left[\mathrm{E}\left[Y_{g, t}(0)\mid \mathbf{D}_g^*,\mathbf{D}_g\right]\mid \mathbf{D}_g\right]\\
	&=\mathrm{E}\left[\mathrm{E}\left[Y_{g, t}(0) \mid \mathbf{D}^*\right]\mid \mathbf{D}\right]-\mathrm{E}\left[\mathrm{E}\left[Y_{g, t}(0)\mid \mathbf{D}_g^*\right]\mid \mathbf{D}_g\right]\\
	&=\mathrm{E}\left[\mathrm{E}\left[Y_{g, t}(0) \mid \mathbf{D}_g^*\right]\mid \mathbf{D}\right]-\mathrm{E}\left[\mathrm{E}\left[Y_{g, t}(0)\mid \mathbf{D}_g^*\right]\mid \mathbf{D}_g\right]\\
	&=\mathrm{E}\left[\mathrm{E}\left[Y_{g, t}(0) \mid \mathbf{D}_g^*\right]\mid \mathbf{D}_g\right]-\mathrm{E}\left[\mathrm{E}\left[Y_{g, t}(0)\mid \mathbf{D}_g^*\right]\mid \mathbf{D}_g\right]\\
	&=0,
\end{align*}
where the second equality follows from condition (\ref{EQ:PROP1d}), the third from (\ref{EQ:PROP1a}), and the fourth from (\ref{EQ:PROP1e}).

Next, for Assumption \ref{A4}, we similarly get:
\begin{align*}
&\quad \mathrm{E}\left[{Y}_{g, t}(0)-{Y}_{g, t-1}(0) \mid \mathbf{D}_g\right]-\mathrm{E}\left[{Y}_{g, t}(0)-{Y}_{g, t-1}(0)\right]\\
	&=\mathrm{E}\left[\mathrm{E}\left[{Y}_{g, t}(0)-{Y}_{g, t-1}(0) \mid \mathbf{D}_g^*,\mathbf{D}_g\right]\mid \mathbf{D}_g\right]-\mathrm{E}\left[{Y}_{g, t}(0)-{Y}_{g, t-1}(0)\right]\\
	&=\mathrm{E}\left[\mathrm{E}\left[{Y}_{g, t}(0)-{Y}_{g, t-1}(0) \mid \mathbf{D}_g^*\right]\mid \mathbf{D}_g\right]-\mathrm{E}\left[{Y}_{g, t}(0)-{Y}_{g, t-1}(0)\right]\\
	&=\mathrm{E}\left[{Y}_{g, t}(0)-{Y}_{g, t-1}(0)\right]-\mathrm{E}\left[{Y}_{g, t}(0)-{Y}_{g, t-1}(0)\right]\\
	&=0,
\end{align*}
where the second equality follows from (\ref{EQ:PROP1c}) and the third from (\ref{EQ:PROP1b}).
\end{proof}

\begin{proof}[Proof of Proposition \ref{PROPTWFE}]

By Assumption \ref{A2} and the Frisch--Waugh--Lovell theorem, we have:
\[
  \mathrm{E}\!\left[\widehat{\beta}_{fe}\mid\mathbf D\right] =
  \frac{\sum_{(g,t)} N_{g,t}\varepsilon_{g,t}
        \mathrm{E}\!\left[Y_{g,t}\mid\mathbf D\right]}
       {\sum_{(g,t)} N_{g,t}D_{g,t}\varepsilon_{g,t}},
\]
where $\varepsilon_{g,t}$ is the estimated residual from the unit-level regression of $D_{g,t}$ on
group and time fixed effects:
\[D_{g,t}=\alpha+\gamma_g+\lambda_t+\varepsilon_{g,t}.\]
As remarked in the proof of Theorem 1 in \citet{CH2020}, it holds by construction that:
\begin{equation}\label{PROJECTION1}
\sum_{t=1}^{T} N_{g,t}\varepsilon_{g,t}=0
\end{equation}
for all $g\in\{1,\ldots,G\}$ and
\begin{equation}\label{PROJECTION2}
\sum_{g=1}^{G} N_{g,t}\varepsilon_{g,t}=0
\end{equation}
for all $t\in\{1,\ldots,T\}$. Focusing on the numerator, we apply (\ref{PROJECTION1}) and (\ref{PROJECTION2}) to obatin:
\begin{align*}
\sum_{(g,t)} N_{g,t}\varepsilon_{g,t}\mathrm{E}\!\left[Y_{g,t}\mid\mathbf D\right]
  &= \sum_{(g,t)} N_{g,t}\varepsilon_{g,t}
     \Bigl(\mathrm{E}[Y_{g,t}\mid\mathbf D]
          -\mathrm{E}[Y_{g,1}\mid\mathbf D]
          -\mathrm{E}[Y_{1,t}\mid\mathbf D]
          +\mathrm{E}[Y_{1,1}\mid\mathbf D]\Bigr).
\end{align*}
For each individual term, by \(Y_{i,g,t}=Y_{i,g,t}(0)+D_{i,g,t}^{*}\Delta\)  and Assumption \ref{A3}, we have:
\begin{align*}
\mathrm{E}\left[Y_{g,t}\mid \mathbf{D}\right]
&=\mathrm{E}\left[\frac{1}{N_{g,t}}\sum_{i=1}^{N_{g,t}}Y_{i,g,t}(0)+ D_{i,g,t}^{\ast}\Delta\mid \mathbf{D}\right]=\mathrm{E}\left[Y_{g,t}(0)\mid \mathbf{D}_{g}\right]+ \Delta \left(\frac{1}{N_{g,t}}\sum_{i=1}^{N_{g,t}} \mathrm{Pr}\left(D_{i,g,t}^{\ast}=1\mid \mathbf{D}\right)\right).
\end{align*} 
for all $(g,t)$. Thus, by Assumptions \ref{A4} and \ref{A5}, we have:
\begin{align*}
&\quad \ \sum_{(g,t)}N_{g,t}\varepsilon_{g,t}\left(\mathrm{E}\left[Y_{g,t}\mid \mathbf{D}\right]-\mathrm{E}\left[Y_{g,1}\mid \mathbf{D}\right]-\mathrm{E}\left[Y_{1,t}\mid \mathbf{D}\right]+\mathrm{E}\left[Y_{1,1}\mid \mathbf{D}\right]\right)\\
&=\sum_{(g,t)}N_{g,t}\varepsilon_{g,t}\left(\mathrm{E}\left[Y_{g,t}(0)\right]-\mathrm{E}\left[Y_{g,1}(0)\right]-\mathrm{E}\left[Y_{1,t}(0)\right]+\mathrm{E}\left[Y_{1,1}(0)\right]\right)\\
&\quad +\Delta\sum_{(g,t)}N_{g,t}\varepsilon_{g,t}\left(\frac{1}{N_{g,t}}\sum_{i=1}^{N_{g,t}} \mathrm{Pr}\left(D_{i,g,t}^{\ast}=1\mid \mathbf{D}\right) - \frac{1}{N_{g,1}}\sum_{i=1}^{N_{g,1}} \mathrm{Pr}\left(D_{i,g,1}^{\ast}=1\mid \mathbf{D}\right)\right.\\
&\left.\qquad\qquad\qquad\qquad\quad  -\frac{1}{N_{1,t}}\sum_{i=1}^{N_{1,t}} \mathrm{Pr}\left(D_{i,1,t}^{\ast}=1\mid \mathbf{D}\right) +\frac{1}{N_{1,1}}\sum_{i=1}^{N_{1,1}} \mathrm{Pr}\left(D_{i,1,1}^{\ast}=1\mid \mathbf{D}\right)\right)\\
=&\Delta\sum_{(g,t)}N_{g,t}\varepsilon_{g,t}\left(\frac{1}{N_{g,t}}\sum_{i=1}^{N_{g,t}} \mathrm{Pr}\left(D_{i,g,t}^{\ast}=1\mid \mathbf{D}\right) \right),
\end{align*}
where the last line follows again from (\ref{PROJECTION1}) and (\ref{PROJECTION2}). Thus, we obtain:
\begin{align*}
	\mathrm{E}\left[\widehat{\beta}_{fe}\mid \mathbf{D}\right]=\Delta\frac{\sum_{(g,t)}N_{g,t}\varepsilon_{g,t}\left(\frac{1}{N_{g,t}}\sum_{i=1}^{N_{g,t}} \mathrm{Pr}\left(D_{i,g,t}^{\ast}=1\mid \mathbf{D}\right) \right)}{\sum_{(g,t)}N_{g,t}D_{g,t}\varepsilon_{g,t}},
\end{align*}
and the result of the proposition immediately follows.

\end{proof}

\begin{proof}[Proof of Proposition \ref{PROPESTIMAND}] 
First, note that
\begin{align*}
\mathrm{E}\left[\widehat{\operatorname{DID}}\right]&=\mathrm{E}\left[\sum_{t=2}^T \frac{N_{1,0,t}}{N_{S}} \mathrm{E}\left[\widehat{\operatorname{DID}}_{ t} \mid \mathbf{D}\right]\right].
\end{align*}
 Besides, by Assumptions \ref{A1}-\ref{A5}:
 \begin{align}\label{EQ:PROP3PROOF}
 \begin{aligned}
 \mathrm{E}\left[Y_{g, t}-Y_{g, t-1}\mid \mathbf{{D}}\right]
	&=\frac{1}{N_{g,t}}\sum_{i=1}^{N_{g,t}}\mathrm{E}\left[{D}_{i,g,t}^*\left({Y}_{i,g, t}(1)-{Y}_{i,g, t}(0)\right)+{Y}_{i,g, t}(0)\mid \mathbf{{D}}\right]\\
	&\quad -\frac{1}{N_{g,t-1}}\sum_{i=1}^{N_{g,t-1}}\mathrm{E}\left[{D}_{i,g,t-1}^*\left({Y}_{i,g, t-1}(1)-{Y}_{i,g, t-1}(0)\right)+{Y}_{i,g, t-1}(0)\mid \mathbf{{D}}\right]\\
	&=\frac{1}{N_{g,t}}\sum_{i=1}^{N_{g,t}}\mathrm{E}\left[{D}_{i,g,t}^*\left({Y}_{i,g, t}(1)-{Y}_{i,g, t}(0)\right)\mid \mathbf{{D}}\right]\\
	&\quad -\frac{1}{N_{g,t-1}}\sum_{i=1}^{N_{g,t-1}}\mathrm{E}\left[{D}_{i,g,t-1}^*\left({Y}_{i,g, t-1}(1)-{Y}_{i,g, t-1}(0)\right)\mid \mathbf{{D}}\right]\\
	&\quad +\mathrm{E}\left[{Y}_{1, t}(0)-{Y}_{1, t-1}(0)\right].
 \end{aligned}
\end{align}
Let $t\geq 2$.  Suppose that there exists some switching-to-treatment group $g\in\{1,\cdots,G\}$ such that $D_{g, t-1}=0$, $D_{g, t}=1$. By Assumption \ref{A8}, there exists also some comparison group  $g^{\prime} \in\{1, \ldots, G\}$ with $g^{\prime} \neq g$ such that $D_{g^{\prime}, t-1}=D_{g^{\prime}, t}=0.$ In this case, 
\allowdisplaybreaks
\begin{align*}
&\quad \ N_{1,0,t}\mathrm{E}\left[\widehat{\operatorname{DID}}_{ t} \mid \mathbf{D}\right]\\
&= \sum_{g: D_{g, t}=1, D_{g, t-1}=0} N_{g,t}\mathrm{E}\left[Y_{g, t}-Y_{g, t-1}\mid \mathbf{D}\right] - \frac{N_{1,0,t}}{N_{0,0,t}}\sum_{g: D_{g, t}=0, D_{g, t-1}=0} N_{g,t}\mathrm{E}\left[Y_{g, t}-Y_{g, t-1} \mid \mathbf{D}\right]\\
&= \sum_{g: D_{g, t}=1, D_{g, t-1}=0}  \sum_{i=1}^{N_{g,t}} \mathrm{E}\left[{D}_{i,g,t}^*\left(Y_{i, g, t}(1)-Y_{i, g, t}(0)\right)\mid \mathbf{D}\right]\\
&\quad -\sum_{g: D_{g, t}=1, D_{g, t-1}=0}  \frac{N_{g,t}}{N_{g,t-1}}\sum_{i=1}^{N_{g,t-1}}\mathrm{E}\left[{D}_{i,g,t-1}^*\left({Y}_{i,g, t-1}(1)-{Y}_{i,g, t-1}(0)\right)\mid \mathbf{{D}}\right]\\
&\quad +\sum_{g: D_{g, t}=1, D_{g, t-1}=0} N_{g,t}\mathrm{E}\left[{Y}_{1, t}(0)-{Y}_{1, t-1}(0)\right]\\
&\quad -\frac{N_{1,0,t}}{N_{0,0,t}}\sum_{g: D_{g, t}=0, D_{g, t-1}=0}  \sum_{i=1}^{N_{g,t}} \mathrm{E}\left[{D}_{i,g,t}^*\left(Y_{i, g, t}(1)-Y_{i, g, t}(0)\right)\mid \mathbf{D}\right]\\
&\quad +\frac{N_{1,0,t}}{N_{0,0,t}} \sum_{g: D_{g, t}=0, D_{g, t-1}=0}  \frac{N_{g,t}}{N_{g,t-1}}\sum_{i=1}^{N_{g,t-1}}\mathrm{E}\left[{D}_{i,g,t-1}^*\left({Y}_{i,g, t-1}(1)-{Y}_{i,g, t-1}(0)\right)\mid \mathbf{{D}}\right]\\
&\quad -\frac{N_{1,0,t}}{N_{0,0,t}} \sum_{g: D_{g, t}=0, D_{g, t-1}=0} N_{g,t}\mathrm{E}\left[{Y}_{1, t}(0)-{Y}_{1, t-1}(0)\right]\\
&=\sum_{g: D_{g, t}=1, D_{g, t-1}=0}  \sum_{i=1}^{N_{g,t}} \mathrm{E}\left[Y_{i, g, t}(1)-Y_{i, g, t}(0)\mid \mathbf{D}\right]\\
&\quad -\sum_{g: D_{g, t}=1, D_{g, t-1}=0}  \frac{N_{g,t}}{N_{g,t-1}}\sum_{i=1}^{N_{g,t-1}}\mathrm{E}\left[{D}_{i,g,t-1}^*\left({Y}_{i,g, t-1}(1)-{Y}_{i,g, t-1}(0)\right)\mid \mathbf{{D}}\right]\\
&\quad -\frac{N_{1,0,t}}{N_{0,0,t}}\sum_{g: D_{g, t}=0, D_{g, t-1}=0}  \sum_{i=1}^{N_{g,t}} \mathrm{E}\left[{D}_{i,g,t}^*\left(Y_{i, g, t}(1)-Y_{i, g, t}(0)\right)\mid \mathbf{D}\right],
\end{align*}
where the second equality follows from the formula for $\mathrm{E}\left[Y_{g, t}-Y_{g, t-1}\mid \mathbf{{D}}\right]$, 
and the last from Assumption \ref{A7}. If there does not exist any switching-to-treatment group $g\in\{1,\cdots,G\}$ such that $D_{g, t-1}=0$, $D_{g, t}=1$, the above formula still holds, as $\widehat{\operatorname{DID}}_{ t}=0$, $\mathbb{I}\left\{D_{g, t}=1, D_{g, t-1}=0\right\}=0$ and $N_{1,0,t}=0$ in this case. Thus,
\allowdisplaybreaks
\begin{align*}
 \mathrm{E}\left[\widehat{\mathrm{DID}}\right] &=  \mathrm{E}\left[\mathrm{E}\left[\widehat{\mathrm{DID}} \mid \mathbf{D}\right]\right]\\
 &=\mathrm{E}\left[\sum_{t\geq 2} \frac{N_{1,0,t}}{N_{S}} \mathrm{E}\left[\widehat{\operatorname{DID}}_{ t} \mid \mathbf{D}\right]\right]\\
 &= \mathrm{E}\left[\sum_{t\geq 2}\frac{1}{N_{S}} \sum_{g: D_{g, t}=1, D_{g, t-1}=0}\sum_{i=1}^{N_{g,t}}( Y_{i, g, t}(1)-Y_{i, g, t}(0))\right]\\
 &\quad - \mathrm{E}\left[\sum_{t\geq 2}\frac{1}{N_{S}} \sum_{g: D_{g, t}=1, D_{g, t-1}=0}\sum_{i=1}^{N_{g,t}}\frac{N_{g,t}}{N_{g,t-1}}{D}_{i,g,t-1}^*\left(Y_{i, g, t-1}(1)-Y_{i, g, t-1}(0)\right)\right]\\
 &\quad - \mathrm{E}\left[\sum_{t\geq 2}\frac{N_{1,0,t}}{N_{S}N_{0,0,t}} \sum_{g: D_{g, t}=0, D_{g, t-1}=0}\sum_{i=1}^{N_{g,t}} {D}_{i,g,t}^*\left(Y_{i, g, t}(1)-Y_{i, g, t}(0)\right)\right]\\
 &= \underbrace{\mathrm{E}\left[\frac{1}{N_{S}} \sum_{(g,t):t\geq 2, D_{g, t}=1, D_{g, t-1}=0}\sum_{i=1}^{N_{g,t}} (Y_{i, g, t}(1)-Y_{i, g, t}(0))\right]}_{\delta^S}\\
 &\quad - \mathrm{E}\left[\frac{1}{N_{S}}\sum_{(g,t):t\geq 2, D_{g, t}=1, D_{g, t-1}=0}\sum_{i=1}^{N_{g,t}}\frac{N_{g,t}}{N_{g,t-1}}{D}_{i,g,t-1}^*\left(Y_{i, g, t-1}(1)-Y_{i, g, t-1}(0)\right)\right]\\
 &\quad - \mathrm{E}\left[\frac{1}{N_{S}}\sum_{(g,t):t\geq 2, D_{g, t}=0, D_{g, t-1}=0}\sum_{i=1}^{N_{g,t}}\frac{N_{1,0,t}}{N_{0,0,t}} {D}_{i,g,t}^*\left(Y_{i, g, t}(1)-Y_{i, g, t}(0)\right)\right]
\end{align*}
\end{proof} 

\begin{proof}[Proof of Proposition \ref{PROPESTIMANDTRUE}] 
Recall from the previous proof that under Assumption \ref{A1}-\ref{A8}: 
\begin{align*}
 \mathrm{E}\left[\widehat{\mathrm{DID}}\right] =&  \mathrm{E}\left[\mathrm{E}\left[\widehat{\mathrm{DID}} \mid \mathbf{D}\right]\right]\\
  =&  \mathrm{E}\left[\frac{1}{N_{S}} \sum_{(g,t):t\geq 2, D_{g, t}=1, D_{g, t-1}=0}\sum_{i=1}^{N_{g,t}} (Y_{i, g, t}(1)-Y_{i, g, t}(0))\right]\\
 &\quad - \mathrm{E}\left[\frac{1}{N_{S}}\sum_{(g,t):t\geq 2, D_{g, t}=1, D_{g, t-1}=0}\sum_{i=1}^{N_{g,t}}\frac{N_{g,t}}{N_{g,t-1}}{D}_{i,g,t-1}^*\left(Y_{i, g, t-1}(1)-Y_{i, g, t-1}(0)\right)\right]\\
 &\quad - \mathrm{E}\left[\frac{1}{N_{S}}\sum_{(g,t):t\geq 2, D_{g, t}=0, D_{g, t-1}=0}\sum_{i=1}^{N_{g,t}}\frac{N_{1,0,t}}{N_{0,0,t}} {D}_{i,g,t}^*\left(Y_{i, g, t}(1)-Y_{i, g, t}(0)\right)\right]\\
\equiv &A_{N_{S}} + B_{N_{S}} + C_{N_{S}}
\end{align*}
with the obvious notation. We only decompose $A_{N_{S}}$ and $C_{N_{S}}$, while $B_{N_{S}}$ will contribute directly to the bias term in the proposition. Now, noting that $\mathbb{I}\left\{D_{g, t}=1, D_{g, t-1}=0\right\}= \mathbb{I}\left\{D_{g, t}=1, D_{g, t-1}=0\right\}D_{g, t}^*$ by Assumption \ref{A7}, and $D_{i,g,t}^{\ast}=S_{i,g,t}^{\ast}D_{i,g,t}^{\ast} + (1 - S_{i,g,t}^{\ast})D_{i,g,t}^{\ast}$, we have for $A_{N_{S}}$ that:
\begin{align*}
A_{N_{S}}&=\mathrm{E}\left[\frac{1}{N_{S}} \sum_{(g,t):t\geq 2, D_{g, t}=1, D_{g, t-1}=0}\sum_{i=1}^{N_{g,t}} D_{i,g,t}^{\ast}(Y_{i, g, t}(1)-Y_{i, g, t}(0))\right]\\
&=\mathrm{E}\left[\frac{1}{N_{S^{\ast}}}\sum_{(g,t):t\geq 2, D_{g, t}=1, D_{g, t-1}=0} \sum_{i=1}^{N_{g,t}} D^{\ast}_{i,g, t} S_{i,g,t}^{\ast}\left(Y_{i, g, t}(1)-Y_{i, g, t}(0)\right)	\right]\\
&\quad +\mathrm{E}\left[\left(\frac{1}{N_{S}}-\frac{1}{N_{S^{\ast}}}\right) \sum_{(g,t):t\geq 2, D_{g, t}=1, D_{g, t-1}=0} \sum_{i=1}^{N_{g,t}} D^{\ast}_{i,g, t} S_{i,g,t}^{\ast}\left(Y_{i, g, t}(1)-Y_{i, g, t}(0)\right)	\right]\\
  &\quad + \mathrm{E}\left[\frac{1}{N_{S}} \sum_{(g,t):t\geq 2, D_{g, t}=1, D_{g, t-1}=0} \sum_{i=1}^{N_{g,t}} D^{\ast}_{i,g, t}  \left(1-S_{i,g,t}^{\ast}\right)\left(Y_{i, g, t}(1)-Y_{i, g, t}(0)\right)\right].
\end{align*}
On the other hand, the term $C_{N_{S}}$ can be decomposed as follows:
\allowdisplaybreaks
\begin{align*}
C_{N_{S}}&=\mathrm{E}\left[\frac{1}{N_{S}}\sum_{(g,t):t\geq 2, D_{g, t}=0, D_{g, t-1}=0}\sum_{i=1}^{N_{g,t}} {D}_{i,g,t}^*\left(Y_{i, g, t}(1)-Y_{i, g, t}(0)\right)\right]\\
&\quad -\mathrm{E}\left[\frac{1}{N_{S}}\sum_{(g,t):t\geq 2, D_{g, t}=0, D_{g, t-1}=0}\sum_{i=1}^{N_{g,t}}\left(\frac{N_{1,0,t}}{N_{0,0,t}}-1\right) {D}_{i,g,t}^*\left(Y_{i, g, t}(1)-Y_{i, g, t}(0)\right)\right]\\
&\equiv C_{N_{S}}(1) + C_{N_{S}}(2), 
\end{align*}
where $C_{N_{S}}(1)$ can be decomposed further as:
\begin{align*}
C_{N_{S}}(1)
&=\mathrm{E}\left[\frac{1}{N_{S^{\ast}}}\sum_{(g,t):t\geq 2, D_{g, t}=0, D_{g, t-1}=0}\sum_{i=1}^{N_{g,t}} {D}_{i,g,t}^*S_{i,g,t}^{\ast}\left(Y_{i, g, t}(1)-Y_{i, g, t}(0)\right)\right]\\
&\quad +\mathrm{E}\left[\left(\frac{1}{N_{S}}-\frac{1}{N_{S^{\ast}}}\right)\sum_{(g,t):t\geq 2, D_{g, t}=0, D_{g, t-1}=0}\sum_{i=1}^{N_{g,t}} {D}_{i,g,t}^*S_{i,g,t}^{\ast}\left(Y_{i, g, t}(1)-Y_{i, g, t}(0)\right)\right]\\
&\quad +\mathrm{E}\left[\frac{1}{N_{S}}\sum_{(g,t):t\geq 2, D_{g, t}=0, D_{g, t-1}=0}\sum_{i=1}^{N_{g,t}} {D}_{i,g,t}^*\left(1-S_{i,g,t}^{\ast}\right)\left(Y_{i, g, t}(1)-Y_{i, g, t}(0)\right)\right].
\end{align*}
Note in fact that the first term on the RHS of the last equation is equal to:
\[
\mathrm{E}\left[\frac{1}{N_{S^{\ast}}}\sum_{(g,t):t\geq 2, D_{g, t}=1, D_{g, t-1}=0}\sum_{i=1}^{N_{g,t-1}} {D}_{i,g,t-1}^*S_{i,g,t-1}^{\ast}\left(Y_{i, g, t-1}(1)-Y_{i, g, t-1}(0)\right)\right]
\]
by the assumptions on the switching probabilities in Proposition \ref{PROPESTIMANDTRUE}. Finally, we have that:
\begin{align*}
		\delta^{S^{*}}
	&=\mathrm{E}\left[\frac{1}{N_{S^*}}\sum_{(i, g,t):  t\geq 2, D_{i,g,t}^{\ast}=1, S_{i,g, t}^* =1} \left({Y}_{i, g, t}(1)-{Y}_{i, g, t}(0)\right)\right]\\
	&=\mathrm{E}\left[\frac{1}{N_{S^*}}\sum_{(i, g,t):  t\geq 2, D_{i,g,t}^{\ast}=1, S_{i,g, t}^* =1,D_{g, t}=1, D_{g, t-1}=0} \left({Y}_{i, g, t}(1)-{Y}_{i, g, t}(0)\right)\right]\\
	&\quad +\mathrm{E}\left[\frac{1}{N_{S^*}}\sum_{(i, g,t):  t\geq 2, D_{i,g,t-1}^{\ast}=1, S_{i,g, t-1}^* =1,D_{g, t}=1, D_{g, t-1}=0} \left({Y}_{i, g, t-1}(1)-{Y}_{i, g, t-1}(0)\right)\right]\\
	&=\mathrm{E}\left[\frac{1}{N_{S^*}}\sum_{(g,t):  t\geq 2, D_{g, t}=1, D_{g, t-1}=0} \sum_{i=1}^{N_{g,t}} D_{i,g,t}^{\ast} S_{i,g, t}^*\left({Y}_{i, g, t}(1)-{Y}_{i, g, t}(0)\right)\right]\\
	&\quad +\mathrm{E}\left[\frac{1}{N_{S^*}}\sum_{(g,t):  t\geq 2, D_{g, t}=1, D_{g, t-1}=0} \sum_{i=1}^{N_{g,t-1}}D_{i,g,t-1}^{\ast} S_{i,g, t-1}^*\left({Y}_{i, g, t-1}(1)-{Y}_{i, g, t-1}(0)\right)\right],
	\end{align*}
where the second equality follows from the assumptions on the switching probabilities. Thus, collecting all terms, we obtain:
\allowdisplaybreaks
\begin{align*}
 \mathrm{E}\left[\widehat{\mathrm{DID}}\right] &=  \mathrm{E}\left[\mathrm{E}\left[\widehat{\mathrm{DID}} \mid \mathbf{D}\right]\right]\\
  &=\delta^{S^{\ast}} \\
&\quad +\mathrm{E}\left[\left(\frac{1}{N_{S}}-\frac{1}{N_{S^{\ast}}}\right) \sum_{(g,t):t\geq 2, D_{g, t}=1, D_{g, t-1}=0} \sum_{i=1}^{N_{g,t}} D^{\ast}_{i,g, t} S_{i,g,t}^{\ast}\left(Y_{i, g, t}(1)-Y_{i, g, t}(0)\right)	\right]\\
  &\quad + \mathrm{E}\left[\frac{1}{N_{S}} \sum_{(g,t):t\geq 2, D_{g, t}=1, D_{g, t-1}=0} \sum_{i=1}^{N_{g,t}} D^{\ast}_{i,g, t}  \left(1-S_{i,g,t}^{\ast}\right)\left(Y_{i, g, t}(1)-Y_{i, g, t}(0)\right)\right]\\
 &\quad - \mathrm{E}\left[\frac{1}{N_{S}}\sum_{(g,t):t\geq 2, D_{g, t}=1, D_{g, t-1}=0}\sum_{i=1}^{N_{g,t}}\frac{N_{g,t}}{N_{g,t-1}}{D}_{i,g,t-1}^*\left(Y_{i, g, t-1}(1)-Y_{i, g, t-1}(0)\right)\right]\\
&\quad +\mathrm{E}\left[\left(\frac{1}{N_{S}}-\frac{1}{N_{S^{\ast}}}\right)\sum_{(g,t):t\geq 2, D_{g, t}=0, D_{g, t-1}=0}\sum_{i=1}^{N_{g,t}} {D}_{i,g,t}^*S_{i,g,t}^{\ast}\left(Y_{i, g, t}(1)-Y_{i, g, t}(0)\right)\right]\\
&\quad +\mathrm{E}\left[\frac{1}{N_{S}}\sum_{(g,t):t\geq 2, D_{g, t}=0, D_{g, t-1}=0}\sum_{i=1}^{N_{g,t}} {D}_{i,g,t}^*\left(1-S_{i,g,t}^{\ast}\right)\left(Y_{i, g, t}(1)-Y_{i, g, t}(0)\right)\right]\\
&\quad -\mathrm{E}\left[\frac{1}{N_{S}}\sum_{(g,t):t\geq 2, D_{g, t}=0, D_{g, t-1}=0}\sum_{i=1}^{N_{g,t}}\left(\frac{N_{1,0,t}}{N_{0,0,t}}-1\right) {D}_{i,g,t}^*\left(Y_{i, g, t}(1)-Y_{i, g, t}(0)\right)\right].
\end{align*}
\end{proof}

\begin{proof}[Proof of Proposition \ref{PROP6}] First, we consider  $\widehat{\operatorname{DID}}_{ t}^{**}$.   For  $2\leq t\leq T-1$, suppose that there exists a switching-to-treatment group $g\in\{1,\cdots,G\}$ such that $D_{g, t}=0$, $D_{g, t+1}=1$, then,  by Assumption \ref{A8}, there exists also some comparison group  $g^{\prime} \in\{1, \ldots, G\}$ with $g^{\prime} \neq g$ such that $D_{g^{\prime}, t}=D_{g^{\prime}, t+1}=0$. As a consequence of Assumption \ref{A6}, $D_{g^{\prime}, t-1}=0$, and Assumption \ref{A1} implies that all denominators are nonzero.    Recall that  $N_{1,0,t}^{[1]}:=\sum_{g:  D_{g, t+1}=1, D_{g, t}=0} N_{g,t}$ and $N_{0,0,t}^{[1]}:=\sum_{g:  D_{g, t+1}=0, D_{g, t}=0} N_{g,t}$. By (\ref{EQ:PROP3PROOF}) in the proof of Proposition \ref{PROPESTIMAND}:
\allowdisplaybreaks
\small
\begin{align*}
&\quad \ \frac{N_{1,0,t}^{[1]}N_{1,0,t}}{N_{0,0,t}}\mathrm{E}\left[\widehat{\operatorname{DID}}_{ t}^{**} \mid \mathbf{D}\right]\\
&=\frac{N_{1,0,t}}{N_{0,0,t}}\sum_{g: D_{g, t+1}=1, D_{g, t}=0} N_{g,t}\mathrm{E}\left[Y_{g, t}-Y_{g, t-1} \mid \mathbf{D}\right] - \frac{N_{1,0,t}^{[1]}N_{1,0,t}}{N_{0,0,t}^{[1]}N_{0,0,t}}\sum_{g: D_{g, t+1}=0, D_{g, t}=0} N_{g,t}\mathrm{E}\left[Y_{g, t}-Y_{g, t-1} \mid \mathbf{D}\right]\\
&= \frac{N_{1,0,t}}{N_{0,0,t}}\sum_{g: D_{g, t+1}=1, D_{g, t}=0}  \sum_{i=1}^{N_{g,t}} \mathrm{E}\left[{D}_{i,g,t}^*\left(Y_{i, g, t}(1)-Y_{i, g, t}(0)\right)\mid \mathbf{D}\right]\\
&\quad -\frac{N_{1,0,t}}{N_{0,0,t}} \sum_{g: D_{g, t+1}=1, D_{g, t}=0}  \frac{N_{g,t}}{N_{g,t-1}}\sum_{i=1}^{N_{g,t-1}}\mathrm{E}\left[{D}_{i,g,t-1}^*\left({Y}_{i,g, t-1}(1)-{Y}_{i,g, t-1}(0)\right)\mid \mathbf{{D}}\right]\\
&\quad +\frac{N_{1,0,t}}{N_{0,0,t}} \sum_{g: D_{g, t+1}=1, D_{g, t}=0} N_{g,t}\mathrm{E}\left[{Y}_{1, t}(0)-{Y}_{1, t-1}(0)\right]\\
&\quad -\frac{N_{1,0,t}^{[1]}N_{1,0,t}}{N_{0,0,t}^{[1]}N_{0,0,t}}\sum_{g: D_{g, t+1}=0, D_{g, t}=0}  \sum_{i=1}^{N_{g,t}} \mathrm{E}\left[{D}_{i,g,t}^*\left(Y_{i, g, t}(1)-Y_{i, g, t}(0)\right)\mid \mathbf{D}\right]\\
&\quad +\frac{N_{1,0,t}^{[1]}N_{1,0,t}}{N_{0,0,t}^{[1]}N_{0,0,t}} \sum_{g: D_{g, t+1}=0, D_{g, t}=0}  \frac{N_{g,t}}{N_{g,t-1}}\sum_{i=1}^{N_{g,t-1}}\mathrm{E}\left[{D}_{i,g,t-1}^*\left({Y}_{i,g, t-1}(1)-{Y}_{i,g, t-1}(0)\right)\mid \mathbf{{D}}\right]\\
&\quad -\frac{N_{1,0,t}^{[1]}N_{1,0,t}}{N_{0,0,t}^{[1]}N_{0,0,t}} \sum_{g: D_{g, t+1}=0, D_{g, t}=0} N_{g,t}\mathrm{E}\left[{Y}_{1, t}(0)-{Y}_{1, t-1}(0)\right]\\
&= \frac{N_{1,0,t}}{N_{0,0,t}}\sum_{g: D_{g, t+1}=1, D_{g, t}=0}  \sum_{i=1}^{N_{g,t}} \mathrm{E}\left[{D}_{i,g,t}^*\left(Y_{i, g, t}(1)-Y_{i, g, t}(0)\right)\mid \mathbf{D}\right]
\end{align*}
\normalsize
where the second equality follows from the formula for $\mathrm{E}\left[Y_{g, t}-Y_{g, t-1}\mid \mathbf{{D}}\right]$,  and the last from Assumptions \ref{A6} and  \ref{A7}.   If there exists no switching-to-treatment group $g\in\{1,\cdots,G\}$ such that $D_{g, t}=0$, $D_{g, t+1}=1$, the above formula still holds, as $\widehat{\operatorname{DID}}_{ t}^{**}=0$ and $\mathbb{I}\left\{D_{g, t+1}=1, D_{g, t}=0\right\}=0$. Note that in this case, \(N_{0,0,t}\) must be nonzero, unless there is no group observed to switch at any time at all (in which case we can set all \(\frac{1}{N_{0,0,t}}=0\)). This holds true because otherwise, by Assumption \ref{A6}, it must be \(D_{g,t}=1,D_{g,t-1}=0\)  or  \(D_{g,t}=1,D_{g,t-1}=1\)  for all \(g\in \{1,\cdots,G\}\), and one can show either case violates Assumption \ref{A8}. The last term can be further written as:
\begin{align*}
&\quad \frac{N_{1,0,t}}{N_{0,0,t}}\sum_{g: D_{g, t+1}=1, D_{g, t}=0, D_{g, t-1}=0}  \sum_{i=1}^{N_{g,t}} \mathrm{E}\left[{D}_{i,g,t}^*\left(Y_{i, g, t}(1)-Y_{i, g, t}(0)\right)\mid \mathbf{D}\right]\\
&\quad +\frac{N_{1,0,t}}{N_{0,0,t}}\sum_{g: D_{g, t+1}=0, D_{g, t}=0, D_{g, t-1}=0}  \sum_{i=1}^{N_{g,t}} \mathrm{E}\left[{D}_{i,g,t}^*\left(Y_{i, g, t}(1)-Y_{i, g, t}(0)\right)\mid \mathbf{D}\right]\\
&=	 \frac{N_{1,0,t}}{N_{0,0,t}}\sum_{g: D_{g, t}=0, D_{g, t-1}=0}  \sum_{i=1}^{N_{g,t}} \mathrm{E}\left[{D}_{i,g,t}^*\left(Y_{i, g, t}(1)-Y_{i, g, t}(0)\right)\mid \mathbf{D}\right]
\end{align*}
  When $t=T$,  the formula still continues to hold as $\widehat{\operatorname{DID}}_{T}^{**} =0$  and the last term becomes zero under Assumption \ref{A9}.

For $3\leq t\leq T$, similar steps yield:
\allowdisplaybreaks
\small
\begin{align*}
&\quad \ N_{1,0,t}\mathrm{E}\left[\widehat{\operatorname{DID}}_{ t}^{*} \mid \mathbf{D}\right]\\
&=\sum_{g: D_{g, t}=1, D_{g, t-1}=0} N_{g,t}\mathrm{E}\left[Y_{g, t-1}-Y_{g, t-2} \mid \mathbf{D}\right] \\
&\quad - \sum_{g: D_{g, t}=0, D_{g, t-1}=0} N_{g,t+1}\mathrm{E}\left[Y_{g, t-1}-Y_{g, t-2} \mid \mathbf{D}\right]\\
&= \sum_{g: D_{g, t}=1, D_{g, t-1}=0} \frac{N_{g,t}}{N_{g,t-1}} \sum_{i=1}^{N_{g,t-1}} \mathrm{E}\left[{D}_{i,g,t-1}^*\left(Y_{i, g, t-1}(1)-Y_{i, g, t-1}(0)\right)\mid \mathbf{D}\right]\\
&\quad - \sum_{g: D_{g, t}=1, D_{g, t-1}=0}  \sum_{i=1}^{N_{g,t-2}}\mathrm{E}\left[{D}_{i,g,t-2}^*\left({Y}_{i,g, t-2}(1)-{Y}_{i,g, t-2}(0)\right)\mid \mathbf{{D}}\right]\\
&\quad + \sum_{g: D_{g, t}=1, D_{g, t-1}=0} N_{g,t}\mathrm{E}\left[{Y}_{1, t-1}(0)-{Y}_{1, t-2}(0)\right]\\
&\quad -\sum_{g: D_{g, t}=0, D_{g, t-1}=0}  \frac{N_{g,t}}{N_{g,t-1}} \sum_{i=1}^{N_{g,t-1}} \mathrm{E}\left[{D}_{i,g,t-1}^*\left(Y_{i, g, t-1}(1)-Y_{i, g, t-1}(0)\right)\mid \mathbf{D}\right]\\
&\quad + \sum_{g: D_{g, t}=0, D_{g, t-1}=0}  \sum_{i=1}^{N_{g,t-2}}\mathrm{E}\left[{D}_{i,g,t-2}^*\left({Y}_{i,g, t-2}(1)-{Y}_{i,g, t-2}(0)\right)\mid \mathbf{{D}}\right]\\
&\quad - \sum_{g: D_{g, t}=0, D_{g, t-1}=0} N_{g,t}\mathrm{E}\left[{Y}_{1, t-1}(0)-{Y}_{1, t-2}(0)\right]\\
&= \sum_{g: D_{g, t}=1, D_{g, t-1}=0}  \frac{N_{g,t}}{N_{g,t-1}} \sum_{i=1}^{N_{g,t-1}} \mathrm{E}\left[{D}_{i,g,t-1}^*\left(Y_{i, g, t-1}(1)-Y_{i, g, t-1}(0)\right)\mid \mathbf{D}\right]
\end{align*}
\normalsize
When $t=2$,  the formula still continues to hold as $\widehat{\operatorname{DID}}_{2}^{*} =0$  and the last term becomes zero under Assumption \ref{A10}.

Finally,  combining the above, we obtain:
\begin{align*}
	&\quad \ \mathrm{E}\left[\sum_{t\geq 2} \frac{1}{N_{S}} \left(N_{1,0,t}\widehat{\operatorname{DID}}_{ t}+N_{1,0,t}\widehat{\operatorname{DID}}_{ t}^{*}+\frac{N_{1,0,t}^{[1]}N_{1,0,t}}{N_{0,0,t}}\widehat{\operatorname{DID}}_{ t}^{**}\right) \right]\\
 &= \mathrm{E}\left[\sum_{t\geq 2}\frac{1}{N_{S}} \sum_{g: D_{g, t}=1, D_{g, t-1}=0}  \sum_{i=1}^{N_{g,t}} Y_{i, g, t}(1)-Y_{i, g, t}(0)\right]\\
 &\quad - \mathrm{E}\left[\sum_{t\geq 2}\frac{1}{N_{S}} \sum_{g: D_{g, t}=1, D_{g, t-1}=0}  \sum_{i=1}^{N_{g,t-1}} \frac{N_{g,t}}{N_{g,t-1}}{D}_{i,g,t-1}^*\left(Y_{i, g, t-1}(1)-Y_{i, g, t-1}(0)\right)\right]\\
 &\quad - \mathrm{E}\left[\sum_{t\geq 2}\frac{N_{1,0,t}}{N_{S}N_{0,0,t}} \sum_{g: D_{g, t}=0, D_{g, t-1}=0}  \sum_{i=1}^{N_{g,t}} {D}_{i,g,t}^*\left(Y_{i, g, t}(1)-Y_{i, g, t}(0)\right)\right]\\
 &\quad + \mathrm{E}\left[\sum_{t\geq 2}\frac{1}{N_{S}} \sum_{g: D_{g, t}=1, D_{g, t-1}=0}  \sum_{i=1}^{N_{g,t-1}} \frac{N_{g,t}}{N_{g,t-1}}{D}_{i,g,t-1}^*\left(Y_{i, g, t-1}(1)-Y_{i, g, t-1}(0)\right)\right]\\
 &\quad +\mathrm{E}\left[\sum_{t\geq 2} \frac{N_{1,0,t}}{N_{S}N_{0,0,t}}\sum_{g: D_{g, t}=0, D_{g, t-1}=0}  \sum_{i=1}^{N_{g,t}} {D}_{i,g,t}^*\left(Y_{i, g, t}(1)-Y_{i, g, t}(0)\right)\right]\\ 
 &= \mathrm{E}\left[\frac{1}{N_{S}} \sum_{(i,g,t):t\geq 2, D_{g, t}=1, D_{g, t-1}=0} Y_{i, g, t}(1)-Y_{i, g, t}(0)\right].
\end{align*}
\end{proof}

\begin{proof}[Proof of Proposition \ref{PROP7}]

 Recall the definition of 
\(\mathbf{V}=\left(\mathbf{V}_2',\cdots,\mathbf{V}_T'\right)'\) in (\ref{VVECTOR}) ,
where
\begin{align*}
	\mathbf{V}_t &= \left( 
P_{1,0,t}, 
P_{0,0,t}, 
P^{[1]}_{1,0,t-1}, 
P^{[1]}_{0,0,t-1}, 
P^{[1]}_{1,1,t-1}, 
Q_{1,0,t}, 
Q_{0,0,t}, 
Q^{*}_{1,0,t}, 
Q^{*}_{0,0,t}, 
Q^{**}_{1,0,t}, 
Q^{**}_{0,0,t}, 
Q^{***}_{1,0,t-1}, 
Q^{***}_{1,1,t-1}, 
\right)'.
\end{align*}
Thus, by some algebra, for all combinations of  $P_{d,d^{\prime},t}$, $P_{d,d^{\prime},t-1}^{[1]}$ and \(\widehat{\lambda}_{t-1}\), with $(d,d^{\prime},t)\in\{0,1\}\times \{0,1\}\times \{2,\cdots,T\}$ such that the denominator is strictly positive, we can write
\small\begin{align*}
&\quad \ \widehat{\mathrm{DID}}^{S^{\ast}}
\\&=  \frac{\sum_{t\geq 2}\left(\widehat{\lambda}_{t-1}\left(N_{1,0,t-1}^{[1]}\widehat{\operatorname{DID}}_{t-1}^{**}+N_{1,0,t-1}^{[1]}\widehat{\operatorname{DID}}_{t-1}^{***}\right)+\left(1-\widehat{\lambda}_{t-1}\right)\left(N_{1,0,t}\widehat{\operatorname{DID}}_{ t}+N_{1,0,t}\widehat{\operatorname{DID}}_{ t}^{*}+\frac{N_{1,0,t}^{[1]}N_{1,0,t}}{N_{0,0,t}}\widehat{\operatorname{DID}}_{ t}^{**}\right)\right)}{\sum_{t\geq 2}\sum_{D_{g, t}=1, D_{g, t-1}=0}\left(\widehat{\lambda}_{t-1}N_{g,t-1} + (1-\widehat{\lambda}_{t-1})N_{g,t}\right)}\\
&=  \frac{\sum_{t\geq 2}\widehat{\lambda}_{t-1}\left(P_{1,0,t-1}^{[1]}\left(\frac{Q_{1,0,t-1}^{**}}{P_{1,0,t-1}^{[1]}} - \frac{Q_{0,0,t-1}^{**}}{P_{0,0,t-1}^{[1]}}\right)+P_{1,0,t-1}^{[1]}\left(\frac{Q_{1,0,t-1}^{***}}{P_{1,0,t-1}^{[1]}} -\frac{Q_{1,1,t-1}^{***}}{P_{1,1,t-1}^{[1]}}\right)\right)}{\sum_{t\geq 2}\left(\widehat{\lambda}_{t-1}P^{[1]}_{1,0,t-1} + (1-\widehat{\lambda}_{t-1})P_{1,0,t}\right)}\\
&\quad + \frac{\sum_{t\geq 2}\left(1-\widehat{\lambda}_{t-1}\right)\left(P_{1,0,t}\left(\frac{Q_{1,0,t}}{P_{1,0,t}} - \frac{Q_{0,0,t}}{P_{0,0,t}}\right)+P_{1,0,t}\left(\frac{Q_{1,0,t}^{*}}{P_{1,0,t}} - \frac{Q_{0,0,t}^{*}}{P_{0,0,t}}\right)+\frac{P_{1,0,t}^{[1]}P_{1,0,t}}{P_{0,0,t}}\left(\frac{Q_{1,0,t}^{**}}{P_{1,0,t}^{[1]}} - \frac{Q_{0,0,t}^{**}}{P_{0,0,t}^{[1]}}\right)\right)}{\sum_{t\geq 2}\left(\widehat{\lambda}_{t-1}P^{[1]}_{1,0,t-1} +  (1-\widehat{\lambda}_{t-1})P_{1,0,t}\right)}\\
&=  f_N\left(\mathbf{V}\right) \times f_D\left(\mathbf{V}\right)=:f\left(\mathbf{V}\right),
\end{align*}\normalsize
where  
\footnotesize \begin{align*}
 f_{N}\left(\mathbf{V}\right)
&:=  \sum_{t\geq 2}\left(g_{t}(\boldsymbol{V}) \cdot \left(\left( Q_{1,0,t-1}^{**} - \frac{P_{1,0,t-1}^{[1]}}{P_{0,0,t-1}^{[1]}}Q_{0,0,t-1}^{**}\right)+ \left( Q_{1,0,t-1}^{***} - \frac{P_{1,0,t-1}^{[1]}}{P_{1,1,t-1}^{[1]}}Q_{1,1,t-1}^{***}\right)\right)\right.\\
&\qquad \qquad   \left.+ \left(1-g_{t}(\boldsymbol{V}) \right) \cdot 
\left(\left(Q_{1,0,t} -\frac{ P_{1,0,t}}{P_{0,0,t}}Q_{0,0,t}\right)+\left(Q_{1,0,t}^{*} -\frac{ P_{1,0,t}}{P_{0,0,t}}Q_{0,0,t}^{*}\right)+\left(\frac{P_{1,0,t}}{P_{0,0,t}}Q_{1,0,t}^{**} -\frac{P_{1,0,t}^{[1]}P_{1,0,t}}{P_{0,0,t}^{[1]}P_{0,0,t}}Q_{0,0,t}^{**}\right)\right)\right)
\end{align*}\normalsize
and
\[
f_{D}\left(\mathbf{V}\right):=\left(\sum_{t\geq 2}\left(g_{t}(\boldsymbol{V})P^{[1]}_{1,0,t-1} + (1-g_{t}(\boldsymbol{V}))P_{1,0,t}\right)\right)^{-1}
\]
with
\[
g_{t}(\boldsymbol{V}) := 
\frac{
\displaystyle \frac{Q_{1,0,t-1}^{**}}{P_{1,0,t-1}^{[1]}} - \frac{Q_{0,0,t-1}^{**}}{P_{0,0,t-1}^{[1]}}
}{
\displaystyle 
\left( \frac{Q_{1,0,t-1}^{**}}{P_{1,0,t-1}^{[1]}} - \frac{Q_{0,0,t-1}^{**}}{P_{0,0,t-1}^{[1]}} \right)
+ \left( \frac{Q_{1,0,t-1}^{***}}{P_{1,0,t-1}^{[1]}} - \frac{Q_{1,1,t-1}^{***}}{P_{1,1,t-1}^{[1]}} \right)
}=\widehat{\lambda}_{t-1}.
\]
Note that with the above definitions, we can write:
\[f\left(\mathrm{E}\left[\mathbf{V}\right]\right) =f_N\left(\mathrm{E}\left[\mathbf{V}\right]\right) \times f_D\left(\mathrm{E}\left[\mathbf{V}\right]\right),
\]
where
\small\begin{align*}
&\quad \ f_{N}\left(\mathrm{E}\left[\mathbf{V}\right]\right)\\
&= \sum_{t\geq 2}\left(g_{t}\left(\mathrm{E}\left[\mathbf{V}\right]\right) \cdot \left(\left( \mathrm{E}\left[Q_{1,0,t-1}^{**}\right] - \frac{\mathrm{E}\left[P_{1,0,t-1}^{[1]}\right]}{\mathrm{E}\left[P_{0,0,t-1}^{[1]}\right]} \mathrm{E}\left[Q_{0,0,t-1}^{**}\right] \right) + \left( \mathrm{E}\left[Q_{1,0,t-1}^{***}\right] - \frac{\mathrm{E}\left[P_{1,0,t-1}^{[1]}\right]}{\mathrm{E}\left[P_{1,1,t-1}^{[1]}\right]} \mathrm{E}\left[Q_{1,1,t-1}^{***}\right] \right)\right) \right.\\
&\qquad\qquad + \left(1 - g_{t}\left(\mathrm{E}\left[\mathbf{V}\right]\right)\right) \cdot \left( \left( \mathrm{E}\left[Q_{1,0,t}\right] - \frac{\mathrm{E}\left[P_{1,0,t}\right]}{\mathrm{E}\left[P_{0,0,t}\right]} \mathrm{E}\left[Q_{0,0,t}\right] \right) + \left( \mathrm{E}\left[Q_{1,0,t}^{*}\right] - \frac{\mathrm{E}\left[P_{1,0,t}\right]}{\mathrm{E}\left[P_{0,0,t}\right]} \mathrm{E}\left[Q_{0,0,t}^{*}\right] \right) \right. \\
&\qquad\qquad\qquad\qquad\qquad\qquad\quad  \left. \left. + \left( \frac{\mathrm{E}\left[P_{1,0,t}\right]}{\mathrm{E}\left[P_{0,0,t}\right]} \mathrm{E}\left[Q_{1,0,t}^{**}\right] - \frac{\mathrm{E}\left[P_{1,0,t}^{[1]}\right] \cdot \mathrm{E}\left[P_{1,0,t}\right]}{\mathrm{E}\left[P_{0,0,t}^{[1]}\right] \cdot \mathrm{E}\left[P_{0,0,t}\right]} \mathrm{E}\left[Q_{0,0,t}^{**}\right] \right) \right) \right)
\end{align*}\normalsize
and\[
f_{D}\left(\mathrm{E}\left[\mathbf{V}\right]\right)=\left(\sum_{t\geq 2}\left(g_{t}\left(\mathrm{E}\left[\mathbf{V}\right]\right)\mathrm{E}\left[P^{[1]}_{1,0,t-1}\right] + (1-g_{t}\left(\mathrm{E}\left[\mathbf{V}\right]\right))\mathrm{E}\left[P_{1,0,t}\right]\right)\right)^{-1}
\]
with\[
g_{t}\left(\mathrm{E}\left[\mathbf{V}\right]\right) = 
\frac{
\displaystyle \frac{\mathrm{E}\left[ Q_{1,0,t-1}^{**} \right]}{\mathrm{E}\left[ P_{1,0,t-1}^{[1]} \right]}
-
\frac{\mathrm{E}\left[ Q_{0,0,t-1}^{**} \right]}{\mathrm{E}\left[ P_{0,0,t-1}^{[1]} \right]}
}{
\displaystyle 
\left( 
\frac{\mathrm{E}\left[ Q_{1,0,t-1}^{**} \right]}{\mathrm{E}\left[ P_{1,0,t-1}^{[1]} \right]}
-
\frac{\mathrm{E}\left[ Q_{0,0,t-1}^{**} \right]}{\mathrm{E}\left[ P_{0,0,t-1}^{[1]} \right]}
\right)
+
\left(
\frac{\mathrm{E}\left[ Q_{1,0,t-1}^{***} \right]}{\mathrm{E}\left[ P_{1,0,t-1}^{[1]} \right]}
-
\frac{\mathrm{E}\left[ Q_{1,1,t-1}^{***} \right]}{\mathrm{E}\left[ P_{1,1,t-1}^{[1]} \right]}
\right)
}
\]
In the following, we will start by arguing that 
\[
\sqrt{G}\left(f\left(\mathbf{V}\right) - f\left(\mathrm{E}\left[\mathbf{V}\right]\right) \right)\stackrel{d}{\rightarrow}\mathcal{N}\left(0,\sigma^{2}_{\delta^{S^{\ast}}}\right).
\]
In a second step, we will show that
\[
\sqrt{G}\left( f\left(\mathrm{E}\left[\mathbf{V}\right]\right) -\delta^{S^{\ast}}\right)=o(1),
\] 
which together imply the desired result by Slutsky's theorem.
 
 \vspace{10pt}
 
\textbf{Step 1}: We begin by analyzing the components of $\mathbf{V}$ in (\ref{VVECTOR}), focusing on $Q_{1,0,t}=\frac{1}{G} \sum_{g =1}^G \mathbb{I}\{D_{g,t} = 1, D_{g,t-1} = 0\} N_{g,t} (Y_{g,t} - Y_{g,t-1})$. The arguments for the remaining terms follow analogously. For a given period $t\geq 2$, we denote \(Z_g :=
\mathbb I\{D_{g,t}=1,D_{g,t-1}=0\}
N_{g,t}\left(Y_{g,t}-Y_{g,t-1}\right)\), \(\mu_g:=\mathrm{E}\left[Z_g \right]\) and \(\sigma_g^2:=\mathrm{Var}\left[Z_g \right]\).  We observe:
\begin{align*}
\sup_{g}\mathrm{E}\left[Z_g^4\right]&\leq  \sup_{g}N_{g,t}^{4} \times \sup_{g}\mathrm{E}\left[(Y_{g,t}-Y_{g,t-1})^4\right]\\
&\leq  \sup_{g}N_{g,t}^{4} \times \left(8\sup_{g}\mathrm{E}\left[Y_{g,t}^4\right]+8\sup_{g}\mathrm{E}\left[Y_{g,t-1}^4\right]\right)\\
&\leq  \sup_{g}N_{g,t}^{4} \times 16\sup_{i,g}\mathrm{E}\left[Y_{i,g,t}^4\right]\\
&\leq C_1
\end{align*}
for some constant \(C_1>0\), where the second line follows from the elementary inequality \(|a+b|^p \leq 2^{p-1}\left(|a|^p+|b|^p\right)\) for all \(a,b\in \mathbb{R}\) and \(p>1\), the third  from the Jensen's inequality to the convex function \(x \mapsto x^4\), and the last from Assumption \ref{A19}.  Applying again the above elementary inequality and Jensen's inequality to  \(x \mapsto x^4\), this implies a uniform bound over \(g\):
\begin{align*}
 \mathrm{E}\left[\left|Z_g-\mu_g\right|^4\right] \leq 8 \mathrm{E}\left[Z_g^4\right]+8 \mathrm{E}\left[Z_g\right]^4\leq 16 \mathrm{E}\left[Z_g^4\right]\leq 16C_1.
\end{align*}
Finally, since \(\operatorname{Var}\left(\sqrt{G}Q_{1,0, t}\right)=\frac{1}{G}\sum_{g=1}^G\operatorname{Var}\left(Z_g\right)\) under Assumptions \ref{A15}, Assumption \ref{A19} implies that
\begin{align*}
\frac{S_G^2}{G}:=\frac{1}{G}\sum_{g=1}^G \sigma_g^2=\operatorname{Var}\left(\sqrt{G}Q_{1,0, t}\right)\xrightarrow{G\rightarrow\infty } C_2
\end{align*}
for some constant \(C_2>0\), so that \[  \frac{C_{2}}{2}\leq \frac{S_G^2}{G}\] for sufficiently large \(G\) by the definition of sequence limit.
 Hence, we verify Lyapunov's condition with \(\delta=2\):
\[
\frac{1}{S_G^{2+\delta}}
\sum_{g=1}^{G}\mathrm{E}\!\left[\lvert Z_g-\mu_g\rvert^{2+\delta}\right]\leq \frac{16 C_1 G}{\left(\frac{C_2}{2} G\right)^2}=\frac{64 C_1}{C_2^2} \frac{1}{G}\xrightarrow{G\rightarrow\infty }0.
\]
Lyapunov's CLT thus yields the scalar convergence
\begin{align*}
\sqrt{G}\left(Q_{1,0, t}-\mathrm{E}\left[Q_{1,0, t}\right]\right) \stackrel{d}{\rightarrow}N\left(0,C_2\right).
\end{align*}
Applying the same argument to every coordinate of \(\mathbf V\) and to any linear combination \(\mathbf{c}' \mathbf{V}\) (with \(\mathbf{c}'\)  fixed and of suitable dimension), and invoking the Cram\'{e}r--Wold device, we obtain
\[
\sqrt{G}\left(\mathbf{V}-\mathrm{E}\left[\mathbf{V}\right]\right)\stackrel{d}{\rightarrow} N(\mathbf{0},\boldsymbol{\Sigma}_{\mathbf{V}}),
\]   
where $\boldsymbol{\Sigma}_{\mathbf{V}}= \lim_{G\rightarrow \infty}\text{Var}\left[\sqrt{G} \ \mathbf{V}\right]$ as in (\ref{ASYVARTS}).

\vspace{5pt}

Now, by Assumption \ref{A19}, the limit
\(
\mathbf V^{\infty}:=\lim_{G\to\infty}\mathrm{E}[\mathbf V]
\)
exists. Besides, paralleling the existence argument in \citet{CH2020}, one can show that $f$ is differentiable in a neighborhood of $\mathbf{V}^{\infty}$. Thus, the uniform Delta method \citep[see, e.g.,][Theorem 3.8]{VanderVaart1998} yields:
\begin{align*}
\sqrt{G}\left( f\left(\mathbf{V}\right)-f\left(\mathrm{E}\left[\mathbf{V}\right]\right) \right)
&=\mathbf{J}_{f}	\left(\mathbf{V}^{\infty}\right) \times \sqrt{G}\left(\mathbf{V}-\mathrm{E}\left[\mathbf{V}\right]\right) + o_{\Pr}(1)\stackrel{d}{\rightarrow}\mathcal{N}\left(0,\sigma^{2}_{\delta^{S^{\ast}}}\right),
\end{align*}
where \(\sigma^{2}_{\delta^{S^{\ast}}}=\mathbf{J}_{f}\bigl(\mathbf{V}^{\infty}\bigr)\boldsymbol{\Sigma}_{\mathbf{V}}
\mathbf{J}_{f}\bigl(\mathbf{V}^{\infty}\bigr)^{\prime}
\).

\vspace{10pt}

\textbf{Step 2.1}:  Next, we continue to show that
\[
\sqrt{G}\left( f\left(\mathrm{E}\left[\mathbf{V}\right]\right) -\delta^{S^{\ast}}\right)=o(1).
\]
To this end, we define
\begin{align*}
	T^S:=\frac{1}{G}\sum_{(i, g,t):  t\geq 2, S_{i,g, t}^* =1} \left({Y}_{i, g, t}(1)-{Y}_{i, g, t}(0)\right)\quad \text{ and } \quad  
	P^S:=\frac{N_{S^*}}{G},
\end{align*}
so that \[ \delta^{S^{*}}=\mathrm{E}\left[\frac{T^S}{P^S}\right].\] In the following, we will  show 
  \(f_N\left(\mathrm{E}\left[\mathbf{V}\right]\right) = \mathrm{E}\left[T^S\right]\) and \(f_D\left(\mathrm{E}\left[\mathbf{V}\right]\right)=\left(\mathrm{E}\left[P^S\right]\right)^{-1}\), so that 
\begin{align*}
	f\left(\mathrm{E}\left[\mathbf{V}\right]\right) =\frac{\mathrm{E}\left[T^S\right]}{\mathrm{E}\left[P^S\right]},
\end{align*}
where
\small\begin{align*}
&\quad \ f_{N}\left(\mathrm{E}\left[\mathbf{V}\right]\right)\\
&= \sum_{t\geq 2}\left(g_{t}\left(\mathrm{E}\left[\mathbf{V}\right]\right) \cdot \left(\left( \mathrm{E}\left[Q_{1,0,t-1}^{**}\right] - \frac{\mathrm{E}\left[P_{1,0,t-1}^{[1]}\right]}{\mathrm{E}\left[P_{0,0,t-1}^{[1]}\right]} \mathrm{E}\left[Q_{0,0,t-1}^{**}\right] \right) + \left( \mathrm{E}\left[Q_{1,0,t-1}^{***}\right] - \frac{\mathrm{E}\left[P_{1,0,t-1}^{[1]}\right]}{\mathrm{E}\left[P_{1,1,t-1}^{[1]}\right]} \mathrm{E}\left[Q_{1,1,t-1}^{***}\right] \right)\right) \right.\\
&\qquad\qquad + \left(1 - g_{t}\left(\mathrm{E}\left[\mathbf{V}\right]\right)\right) \cdot \left( \left( \mathrm{E}\left[Q_{1,0,t}\right] - \frac{\mathrm{E}\left[P_{1,0,t}\right]}{\mathrm{E}\left[P_{0,0,t}\right]} \mathrm{E}\left[Q_{0,0,t}\right] \right) + \left( \mathrm{E}\left[Q_{1,0,t}^{*}\right] - \frac{\mathrm{E}\left[P_{1,0,t}\right]}{\mathrm{E}\left[P_{0,0,t}\right]} \mathrm{E}\left[Q_{0,0,t}^{*}\right] \right) \right. \\
&\qquad\qquad\qquad\qquad\qquad\qquad\quad  \left. \left. + \left( \frac{\mathrm{E}\left[P_{1,0,t}\right]}{\mathrm{E}\left[P_{0,0,t}\right]} \mathrm{E}\left[Q_{1,0,t}^{**}\right] - \frac{\mathrm{E}\left[P_{1,0,t}^{[1]}\right] \cdot \mathrm{E}\left[P_{1,0,t}\right]}{\mathrm{E}\left[P_{0,0,t}^{[1]}\right] \cdot \mathrm{E}\left[P_{0,0,t}\right]} \mathrm{E}\left[Q_{0,0,t}^{**}\right] \right) \right) \right)
\end{align*}\normalsize
and\[
f_{D}\left(\mathrm{E}\left[\mathbf{V}\right]\right)=\left(\sum_{t\geq 2}\left(g_{t}\left(\mathrm{E}\left[\mathbf{V}\right]\right)\mathrm{E}\left[P^{[1]}_{1,0,t-1}\right] + (1-g_{t}\left(\mathrm{E}\left[\mathbf{V}\right]\right))\mathrm{E}\left[P_{1,0,t}\right]\right)\right)^{-1}
\]
with\[
g_{t}\left(\mathrm{E}\left[\mathbf{V}\right]\right) = 
\frac{
\displaystyle \frac{\mathrm{E}\left[ Q_{1,0,t-1}^{**} \right]}{\mathrm{E}\left[ P_{1,0,t-1}^{[1]} \right]}
-
\frac{\mathrm{E}\left[ Q_{0,0,t-1}^{**} \right]}{\mathrm{E}\left[ P_{0,0,t-1}^{[1]} \right]}
}{
\displaystyle 
\left( 
\frac{\mathrm{E}\left[ Q_{1,0,t-1}^{**} \right]}{\mathrm{E}\left[ P_{1,0,t-1}^{[1]} \right]}
-
\frac{\mathrm{E}\left[ Q_{0,0,t-1}^{**} \right]}{\mathrm{E}\left[ P_{0,0,t-1}^{[1]} \right]}
\right)
+
\left(
\frac{\mathrm{E}\left[ Q_{1,0,t-1}^{***} \right]}{\mathrm{E}\left[ P_{1,0,t-1}^{[1]} \right]}
-
\frac{\mathrm{E}\left[ Q_{1,1,t-1}^{***} \right]}{\mathrm{E}\left[ P_{1,1,t-1}^{[1]} \right]}
\right)
}
\]
Then, since \(\lim_{G\rightarrow \infty} \mathrm{E}\left[P^S\right]>0\) by Assumption \ref{A18}, the problem reduces to showing
\[
\sqrt{G}\left( \mathrm{E}\left[T^S\right]-\mathrm{E}\left[P^S\right]\mathrm{E}\left[\frac{T^S}{P^S}\right]\right)=o(1).
\]
Examining the components in \(f\left(\mathrm{E}\left[\mathbf{V}\right]\right) \), we first observe,  by (\ref{EQ:PROP3PROOF}) in the proof of Proposition \ref{PROPESTIMAND} (which continue to hold under Assumptions \ref{A1}-\ref{A2} and \ref{A15}-\ref{A17}), that
\begin{align*}
	\mathrm{E}\left[Q_{1,0,t-1}^{**}\right] &= \mathrm{E}\left[ \frac{1}{G} \sum_{g=1}^G \mathbb{I}\{D_{g,t}=1, D_{g,t-1}=0\}  N_{g,t-1}  (Y_{g,t-1} - Y_{g,t-2}) \right]\\
&=  \frac{1}{G} \mathrm{E}\left[\sum_{g=1}^G \mathbb{I}\{D_{g,t}=1, D_{g,t-1}=0\}  \sum_{i=1}^{N_{g,t-1}}\mathrm{E}\left[{D}_{i,g,t-1}^*\left({Y}_{i,g, t-1}(1)-{Y}_{i,g, t-1}(0)\right)\mid \mathbf{{D}}_g\right]\right]\\
&\quad +\frac{1}{G}  \mathrm{E}\left[\sum_{g=1}^G\mathbb{I}\{D_{g,t}=1, D_{g,t-1}=0\} \frac{N_{g,t-1}}{N_{g,t-2}}\sum_{i=1}^{N_{g,t-2}}\mathrm{E}\left[{D}_{i,g,t-2}^*\left({Y}_{i,g, t-2}(1)-{Y}_{i,g, t-2}(0)\right)\mid \mathbf{{D}}_g\right]\right]\\
&\quad +\frac{1}{G} \mathrm{E}\left[\sum_{g=1}^G \mathbb{I}\{D_{g,t}=1, D_{g,t-1}=0\}  N_{g,t-1} \right]\times \mathrm{E}\left[{Y}_{1, t-1}(0)-{Y}_{1, t-2}(0)\right]\\
&=  \frac{1}{G}  \mathrm{E}\left[\sum_{g=1}^G\mathbb{I}\{D_{g,t}=1, D_{g,t-1}=0\}  \sum_{i=1}^{N_{g,t-1}}\mathrm{E}\left[{D}_{i,g,t-1}^*\left({Y}_{i,g, t-1}(1)-{Y}_{i,g, t-1}(0)\right)\mid \mathbf{{D}}_g\right]\right]\\
&\quad +\frac{1}{G}  \mathrm{E}\left[\sum_{g=1}^G\mathbb{I}\{D_{g,t}=1, D_{g,t-1}=0\}  N_{g,t-1} \right]\times \mathrm{E}\left[{Y}_{1, t-1}(0)-{Y}_{1, t-2}(0)\right],
\end{align*}
where the third equality follows from Assumption \ref{A7}. By similar arguments: 
\begin{align*}
	\mathrm{E}\left[Q_{0,0,t-1}^{**}\right] 
	&= \frac{1}{G} \mathrm{E}\left[ \sum_{g=1}^G\mathbb{I}\{D_{g,t}=0, D_{g,t-1}=0\}  N_{g,t-1} \right]\times \mathrm{E}\left[{Y}_{1, t-1}(0)-{Y}_{1, t-2}(0)\right],\\
	\mathrm{E}\left[P_{1,0,t-1}^{[1]}\right] 
	&= \frac{1}{G} \mathrm{E}\left[ \sum_{g=1}^G\mathbb{I}\{D_{g,t}=1, D_{g,t-1}=0\}  N_{g,t-1} \right],\\
	\mathrm{E}\left[P_{0,0,t-1}^{[1]}\right] 
	&= \frac{1}{G}  \mathrm{E}\left[\sum_{g=1}^G\mathbb{I}\{D_{g,t}=0, D_{g,t-1}=0\}  N_{g,t-1} \right].
\end{align*}
Combining the above, we obtain
\begin{align*}
	&\quad \ \mathrm{E}\left[Q_{1,0,t-1}^{**}\right] - \frac{\mathrm{E}\left[P_{1,0,t-1}^{[1]}\right]}{\mathrm{E}\left[P_{0,0,t-1}^{[1]}\right]}\mathrm{E}\left[Q_{0,0,t-1}^{**}\right]\\
&=  \frac{1}{G}  \mathrm{E}\left[\sum_{g=1}^G\mathbb{I}\{D_{g,t}=1, D_{g,t-1}=0\}  \sum_{i=1}^{N_{g,t-1}}\mathrm{E}\left[{D}_{i,g,t-1}^*\left({Y}_{i,g, t-1}(1)-{Y}_{i,g, t-1}(0)\right)\mid \mathbf{{D}}_g\right]\right]\\
&=  \frac{1}{G}  \mathrm{E}\left[\sum_{g=1}^G\mathbb{I}\{D_{g,t}=1, D_{g,t-1}=0\}  \sum_{i=1}^{N_{g,t-1}}\mathrm{E}\left[Y_{i, g, t-1}(1)-Y_{i, g, t-1}(0)\mid {D}_{i,g,t-1}^*=1,\mathbf{D}_g\right]\mathrm{Pr}\left[{D}_{i,g,t-1}^{\ast}=1\mid \mathbf{{D}}_g\right]\right]\\
	&=\lambda_{t-1}\times \frac{1}{G} \sum_{g=1}^G \mathrm{E}\left[\mathbb{I}\{D_{g,t}=1, D_{g,t-1}=0\}  \sum_{i=1}^{N_{g,t-1}}\left({Y}_{i,g, t-1}(1)-{Y}_{i,g, t-1}(0)\right)\right]
\end{align*}
where the last step follows from Assumptions \ref{A13} and \ref{A14}, and the fact that under Assumptions \ref{A6}: 
\begin{align*}
&\quad \ \mathbb{I}\{D_{g,t}=1, D_{g,t-1}=0\}\times \mathrm{Pr}\left[{D}_{i,g,t-1}^{\ast}=1\mid \mathbf{{D}}_g\right]\\
&=\mathbb{I}\{D_{g,t}=1, D_{g,t-1}=0\}\times \mathrm{Pr}\left[{D}_{i,g,t-1}^{\ast}=1\mid {D}_{g,t}=1,D_{g,t-1}=0\right]
\end{align*} 
almost surely. Under the same set of assumptions, one can also show, in parallel to (\ref{EQ:PROP3PROOF}), that:
\begin{align*}
\begin{aligned}
\mathrm{E}\left[Y_{g, t}-Y_{g, t-1} \mid \mathbf{D}\right] 
&= -\frac{1}{N_{g, t}} \sum_{i=1}^{N_{g, t}} \mathrm{E}\left[\left(1-D_{i, g, t}^*\right)\left(Y_{i, g, t}(1)-Y_{i, g, t}(0)\right) \mid \mathbf{D}\right] \\
& \quad +\frac{1}{N_{g, t-1}} \sum_{i=1}^{N_{g, t-1}} \mathrm{E}\left[\left(1-D_{i, g, t-1}^*\right)\left(Y_{i, g, t-1}(1)-Y_{i, g, t-1}(0)\right) \mid \mathbf{D}\right] \\
& \quad +\mathrm{E}\left[Y_{1, t}(1)-Y_{1, t-1}(1)\right],
\end{aligned}
\end{align*}
Using this equation, we obtain:
\begin{align*}
	\mathrm{E}\left[Q_{1,0,t-1}^{***}\right] &= \mathrm{E}\left[ \frac{1}{G} \sum_{g=1}^G \mathbb{I}\{D_{g,t}=1, D_{g,t-1}=0\}  N_{g,t-1}  (Y_{g,t} - Y_{g,t-1}) \right]\\
&= -\frac{1}{G} \mathrm{E}\left[\sum_{g=1}^G \mathbb{I}\{D_{g,t}=1, D_{g,t-1}=0\} \frac{N_{g, t-1}}{N_{g, t}}\sum_{i=1}^{N_{g,t}} \mathrm{E}\left[\left(1-D_{i, g, t}^*\right)\left(Y_{i, g, t}(1)-Y_{i, g, t}(0)\right) \mid \mathbf{D}_g\right] \right] \\
&\quad +\frac{1}{G} \mathrm{E}\left[\sum_{g=1}^G \mathbb{I}\{D_{g,t}=1, D_{g,t-1}=0\} \sum_{i=1}^{N_{g,t-1}} \mathrm{E}\left[\left(1-D_{i, g, t-1}^*\right)\left(Y_{i, g, t-1}(1)-Y_{i, g, t-1}(0)\right) \mid \mathbf{D}_g\right] \right] \\
&\quad + \frac{1}{G} \mathrm{E}\left[\sum_{g=1}^G \mathbb{I}\{D_{g,t}=1, D_{g,t-1}=0\} \cdot N_{g,t} \right] \cdot \mathrm{E}\left[Y_{1,t}(1) - Y_{1,t-1}(1)\right]\\
&=  \frac{1}{G}  \mathrm{E}\left[\sum_{g=1}^G\mathbb{I}\{D_{g,t}=1, D_{g,t-1}=0\}  \sum_{i=1}^{N_{g,t-1}}\mathrm{E}\left[\left(1-{D}_{i,g,t-1}^*\right)\left({Y}_{i,g, t-1}(1)-{Y}_{i,g, t-1}(0)\right)\mid \mathbf{{D}}_g\right]\right]\\
&\quad +\frac{1}{G}  \mathrm{E}\left[\sum_{g=1}^G\mathbb{I}\{D_{g,t}=1, D_{g,t-1}=0\}  N_{g,t-1} \right]\times \mathrm{E}\left[{Y}_{1, t-1}(0)-{Y}_{1, t-2}(0)\right],
\end{align*}
where the the third equality follows from Assumption \ref{A7}. Repeating the above arguments yield:
\begin{align*}
	&\quad \ \mathrm{E}\left[Q_{1,0,t-1}^{***}\right] - \frac{\mathrm{E}\left[P_{1,0,t-1}^{[1]}\right]}{\mathrm{E}\left[P_{1,1,t-1}^{[1]}\right]}\mathrm{E}\left[Q_{1,1,t-1}^{***}\right]	\\
&=  \frac{1}{G}  \mathrm{E}\left[\sum_{g=1}^G\mathbb{I}\{D_{g,t}=1, D_{g,t-1}=0\}  \sum_{i=1}^{N_{g,t-1}}\mathrm{E}\left[\left(1-{D}_{i,g,t-1}^*\right)\left({Y}_{i,g, t-1}(1)-{Y}_{i,g, t-1}(0)\right)\mid \mathbf{{D}}_g\right]\right]\\
&=  \frac{1}{G}  \mathrm{E}\left[\sum_{g=1}^G\mathbb{I}\{D_{g,t}=1, D_{g,t-1}=0\}  \sum_{i=1}^{N_{g,t-1}}\mathrm{E}\left[Y_{i, g, t-1}(1)-Y_{i, g, t-1}(0)\mid {D}_{i,g,t-1}^*=0,\mathbf{D}_g\right]\mathrm{Pr}\left[{D}_{i,g,t-1}^{\ast}=0\mid \mathbf{{D}}_g\right]\right]\\
	&=\left(1-\lambda_{t-1}\right)\times \frac{1}{G} \sum_{g=1}^G \mathrm{E}\left[\mathbb{I}\{D_{g,t}=1, D_{g,t-1}=0\}  \sum_{i=1}^{N_{g,t-1}}\left({Y}_{i,g, t-1}(1)-{Y}_{i,g, t-1}(0)\right)\right]
\end{align*}\normalsize
where the last line follows from Assumption \ref{A13}. Therefore, we obtain
\[
g_{t}\left(\mathrm{E}\left[\mathbf{V}\right]\right) = 
\frac{
\displaystyle \frac{\mathrm{E}\left[ Q_{1,0,t-1}^{**} \right]}{\mathrm{E}\left[ P_{1,0,t-1}^{[1]} \right]}
-
\frac{\mathrm{E}\left[ Q_{0,0,t-1}^{**} \right]}{\mathrm{E}\left[ P_{0,0,t-1}^{[1]} \right]}
}{
\displaystyle 
\left( 
\frac{\mathrm{E}\left[ Q_{1,0,t-1}^{**} \right]}{\mathrm{E}\left[ P_{1,0,t-1}^{[1]} \right]}
-
\frac{\mathrm{E}\left[ Q_{0,0,t-1}^{**} \right]}{\mathrm{E}\left[ P_{0,0,t-1}^{[1]} \right]}
\right)
+
\left(
\frac{\mathrm{E}\left[ Q_{1,0,t-1}^{***} \right]}{\mathrm{E}\left[ P_{1,0,t-1}^{[1]} \right]}
-
\frac{\mathrm{E}\left[ Q_{1,1,t-1}^{***} \right]}{\mathrm{E}\left[ P_{1,1,t-1}^{[1]} \right]}
\right)
}=\lambda_{t-1}.
\]
Finally, since:
\footnotesize\begin{align*}
&\widehat{\mathrm{DID}}_{t}^{S}\\
=& \left( \mathrm{E}\left[Q_{1,0,t}\right] - \frac{\mathrm{E}\left[P_{1,0,t}\right]}{\mathrm{E}\left[P_{0,0,t}\right]} \mathrm{E}\left[Q_{0,0,t}\right] \right) + \left( \mathrm{E}\left[Q_{1,0,t}^{*}\right] - \frac{\mathrm{E}\left[P_{1,0,t}\right]}{\mathrm{E}\left[P_{0,0,t}\right]} \mathrm{E}\left[Q_{0,0,t}^{*}\right] \right)    \\
&+ \left( \frac{\mathrm{E}\left[P_{1,0,t}\right]}{\mathrm{E}\left[P_{0,0,t}\right]} \mathrm{E}\left[Q_{1,0,t}^{**}\right] - \frac{\mathrm{E}\left[P_{1,0,t}^{[1]}\right] \cdot \mathrm{E}\left[P_{1,0,t}\right]}{\mathrm{E}\left[P_{0,0,t}^{[1]}\right] \cdot \mathrm{E}\left[P_{0,0,t}\right]} \mathrm{E}\left[Q_{0,0,t}^{**}\right] \right) 
\\
	=&\frac{1}{G} \sum_{g=1}^G\mathrm{E}\left[ \mathbb{I}\{D_{g, t}=1, D_{g, t-1}=0\}   \sum_{i=1}^{N_{g,t}}\left({Y}_{i,g, t}(1)-{Y}_{i,g, t}(0)\right)\right].
\end{align*}\normalsize
we also obtain:
\small\begin{align*}
&\quad \ f_{N}\left(\mathrm{E}\left[\mathbf{V}\right]\right)\\
&=\mathrm{E}\left[ \frac{1}{G} \sum_{t\geq 2}\sum_{g: D_{g, t}=1, D_{g, t-1}=0} \left(\lambda_{t-1}\times  \sum_{i=1}^{N_{g,t-1}}\left({Y}_{i,g, t-1}(1)-{Y}_{i,g, t-1}(0)\right)\right.\right.\left. + \left(1-\lambda_{t-1}\right) \times \left.   \sum_{i=1}^{N_{g,t}}\left({Y}_{i,g, t}(1)-{Y}_{i,g, t}(0)\right)\right)\right],
\end{align*}\normalsize
and
\begin{align*}
	f_{D}\left(\mathrm{E}\left[\mathbf{V}\right]\right)
	&=\left(\sum_{t \geq 2}\left(\lambda_{t-1} \mathrm{E}\left[P^{[1]}_{1,0,t-1}\right]+\left(1-\lambda_{t-1}\right) \mathrm{E}\left[P_{1,0,t}\right]\right)\right)^{-1}.
\end{align*}

On the other hand, note that by Bayes' rule, one can show that in staggered adoption designs, Assumptions \ref{A11} implies that 
 \[\mathrm{P}\!\left(D_{g,t}=1,D_{g,t-1}=0 \mid S^{\ast}_{i,g,t}=1\right)
+\mathrm{P}\!\left(D_{g,t+1}=1,D_{g,t}=0 \mid S^{\ast}_{i,g,t}=1\right)=1, \quad t\in \{2,\cdots,T-1\},
\]
and similarly: 
 \[\mathrm{P}\!\left(D_{g,T}=1,D_{g,T-1}=0 \mid S^{\ast}_{i,g,T}=1\right)=1.
\]
Thus,
\small\begin{align*}
	&\quad \ \mathrm{E}\left[P^S \right]
	=\mathrm{E}\left[\frac{N_{S^*}}{G}\right]	=\mathrm{E}\left[\frac{1}{G} \sum_{(i,g,t):t\ge2} \mathbb{I}\{S^{*}_{i,g,t}=1\}\right]\\
	&= \mathrm{E}\left[\frac{1}{G}\sum_{t=2}^{T-1} \sum_{g=1}^{G} \left( \mathbb{I}\{D_{g,t+1}=1,D_{g,t}=0\}+\mathbb{I}\{D_{g,t}=1,D_{g,t-1}=0\}\right) \sum_{i=1}^{N_{g,t}}\mathbb{I}\{S^{*}_{i,g,t}=1\}\right]\\
	&\quad +\mathrm{E}\left[\frac{1}{G} \sum_{g=1}^{G} \mathbb{I}\{D_{g,T}=1,D_{g,T-1}=0\}\sum_{i=1}^{N_{g,T}}\mathbb{I}\{S^{*}_{i,g,T}=1\}\right]\\
	&= \mathrm{E}\left[\frac{1}{G}\sum_{t=3}^T \sum_{g=1}^{G} \mathbb{I}\{D_{g,t}=1,D_{g,t-1}=0\} \sum_{i=1}^{N_{g,t-1}} \mathbb{I}\{S_{i,g,t-1}^*=1\}+\sum_{t=2}^T \sum_{g=1}^{G} \mathbb{I}\{D_{g,t}=1,D_{g,t-1}=0\} \sum_{i=1}^{N_{g,t}} \mathbb{I}\{S_{i,g,t}^*=1\}\right]\\
	&= \frac{1}{G} \sum_{t \geq 2} \sum_{g=1}^{G}\mathrm{E}\left[\mathbb{I}\left\{ D_{g,t} = 1,\ D_{g,t-1} = 0 \right\}\mathrm{E}\left[  \sum_{i=1}^{N_{g,t-1}} \mathbb{I}\left\{ S_{i,g,t-1}^{\ast} = 1\right\} +\sum_{i=1}^{N_{g,t}} \mathbb{I}\left\{ S_{i,g,t}^{\ast} = 1 \right\}\mid \mathbf{D}_g\right]\right]\\
	&=\frac{1}{G}\sum_{t \geq 2} \sum_{g=1}^{G} \mathrm{E}\left[\mathbb{I}\left\{ D_{g,t} = 1,\ D_{g,t-1} = 0 \right\} 
\left( \lambda_{t-1} N_{g,t-1} 
+ \left(1-\lambda_{t-1}\right)N_{g,t}\right)\right]\\
	&=\sum_{t \geq 2}\left(\lambda_{t-1} \mathrm{E}\left[P^{[1]}_{1,0,t-1}\right]+\left(1-\lambda_{t-1}\right) \mathrm{E}\left[P_{1,0,t}\right]\right)=\left(f_{D}\left(\mathrm{E}\left[\mathbf{V}\right]\right)\right)^{-1},
\end{align*}\normalsize
where the second equality follows from reindexing, and the third from Assumptions \ref{A10} and \ref{A11}, the fourth from \(\mathrm{Pr}\left( S^{\ast}_{g,1} = 0 \mid D_{g,2} = 0,\ D_{g,1} = 0 \right) = 1\), and the sixth  from Assumption \ref{A13}. Similarly, we obtain:
\footnotesize\begin{align*}
		&\quad \ \mathrm{E}\left[T^S\right]=\mathrm{E}\left[\frac{1}{G}\sum_{(i,g,t):t\ge2} \mathbb{I}\{S_{i,g,t}^*=1\} \left({Y}_{i, g, t}(1)-{Y}_{i, g, t}(0)\right)\right]\\
	&=\frac{1}{G} \sum_{t \geq 2} \sum_{g=1}^{G} \mathrm{E}\left[\mathbb{I}\left\{ D_{g,t} = 1,\ D_{g,t-1} = 0 \right\} \left( \mathrm{E}\left[\sum_{i=1}^{N_{g,t-1}} \mathbb{I}\left\{ S_{i,g,t-1}^{\ast} = 1\right\}\left({Y}_{i,g, t-1}(1)-{Y}_{i,g, t-1}(0)\right) \mid \mathbf{D}_g\right]\right.\right.\\
	&\qquad\qquad\qquad\qquad\qquad\qquad\qquad\qquad\qquad\  \left.\left.+ \mathrm{E}\left[\sum_{i=1}^{N_{g,t}} \mathbb{I}\left\{ S_{i,g,t}^{\ast} = 1 \right\}\left({Y}_{i,g, t}(1)-{Y}_{i,g, t}(0)\right)\mid \mathbf{D}_g\right]\right)\right]\\
	&=\frac{1}{G} \sum_{t \geq 2} \sum_{g=1}^{G} \mathrm{E}\left[\mathbb{I}\left\{ D_{g,t} = 1,\ D_{g,t-1} = 0 \right\} \left( \mathrm{E}\left[\sum_{i=1}^{N_{g,t-1}} \left({Y}_{i,g, t-1}(1)-{Y}_{i,g, t-1}(0)\right) \mid S_{i,g,t-1}^{\ast} = 1,\mathbf{D}_g\right]\times \Pr\left(S_{i,g,t-1}^{\ast} = 1\mid \mathbf{D}_g\right)\right.\right.\\
	&\qquad\qquad\qquad\qquad\qquad\qquad\qquad\qquad\qquad\  \left.\left.+ \mathrm{E}\left[\sum_{i=1}^{N_{g,t}} \left({Y}_{i,g, t}(1)-{Y}_{i,g, t}(0)\right)\mid S_{i,g,t}^{\ast} = 1,\mathbf{D}_g\right]\times \Pr\left(S_{i,g,t}^{\ast} = 1\mid \mathbf{D}_g\right)\right)\right]\\
&= \mathrm{E}\left[ \frac{1}{G}\sum_{t\geq 2}\sum_{g: D_{g, t}=1, D_{g, t-1}=0} \left(\lambda_{t-1}\times  \sum_{i=1}^{N_{g,t-1}}\left({Y}_{i,g, t-1}(1)-{Y}_{i,g, t-1}(0)\right)\right.\right.\left. + \left(1-\lambda_{t-1}\right) \times \left.   \sum_{i=1}^{N_{g,t}}\left({Y}_{i,g, t}(1)-{Y}_{i,g, t}(0)\right)\right)\right],
	\end{align*}\normalsize
where the last step follows from Assumptions \ref{A13} and \ref{A14}.

Thus, we conclude that
\begin{align*}
	f\left(\mathrm{E}\left[\mathbf{V}\right]\right) =\frac{\mathrm{E}\left[T^S\right]}{\mathrm{E}\left[P^S\right]},
\end{align*}
and  the problem reduces to showing
\(
\sqrt{G}\left( \mathrm{E}\left[T^S\right]-\mathrm{E}\left[P^S\right]\mathrm{E}\left[\frac{T^S}{P^S}\right]\right)=o(1).
\)
\vspace{10pt}

\textbf{Step 2.2}: Next, we decompose:
\begin{align}\label{EQ:PROP6STEP2}
\mathrm{E}\left[T^S\right] - \mathrm{E}\left[P^S\right] \cdot \mathrm{E}\left[\frac{T^S}{P^S}\right] 
&= \underbrace{\mathrm{E}\left[ \left( P^S - \mathrm{E}\left[P^S\right] \right) \cdot \frac{\widetilde{T}^S}{P^S} \cdot I_G \right]}_{=:A} 
+ \underbrace{\mathrm{E}\left[ \left( P^S - \mathrm{E}\left[P^S\right] \right) \cdot \frac{\widetilde{T}^S}{P^S} \cdot (1 - I_G) \right]}_{=:B},
\end{align}
where \[\widetilde{T}^S 
= \mathrm{E}\left[T^S \mid \left(S_{i,g,t}^{\ast}\right)_{(i,g,t)\in \{1,\cdots,N_{g,t}\}\times \{1,\cdots,G\}\times \{1,\cdots,T\}} \right] \quad \text{and} \quad 
I_G := \mathbb{I}\left\{ \left|P^S - \mathrm{E}\left[P^S\right]\right| < \varepsilon_G \right\}
\] for some \(\varepsilon_G > 0\) that will be specified below. 

\vspace{10pt}

We begin with the second term \(B\) on the right-hand side of (\ref{EQ:PROP6STEP2}). Applying the Cauchy--Schwarz inequality twice yields:
\begin{align*}
 \sqrt{G} \left|B\right|&= \sqrt{G} \left\vert \mathrm{E}\left[ \left( P^S - \mathrm{E}\left[P^S\right] \right) \cdot \frac{\widetilde{T}^S}{P^S} \cdot (1 - I_G) \right] \right\vert \\
&\leq \sqrt{G} \cdot \left( \mathrm{E}\left[ \left( P^S - \mathrm{E}\left[P^S\right] \right)^2 \right] \right)^{1/2} \cdot \left( \mathrm{E}\left[ \left( \frac{\widetilde{T}^S}{P^S} \right)^4 \right]^{1/2} \cdot \mathrm{E}\left[ (1 - I_G)^4 \right]^{1/2} \right)^{1/2} \\
&= \underbrace{\sqrt{G} \cdot \left( \mathrm{Var}\left[P^S\right] \right)^{1/2}}_{=:B_1} \cdot \underbrace{\left( \mathrm{E}\left[ \left( \frac{\widetilde{T}^S}{P^S} \right)^4 \right] \cdot \mathrm{Pr}(I_G = 0) \right)^{1/4}}_{=:B_2}.
\end{align*}
It is easy to check that the term \(B_1\) converges to a finite limit by Assumption \ref{A19}. It remains to show that the remaining term tends to zero. We observe that 
\begin{align*}
\widetilde{T}^S 
&= \mathrm{E}\left[T^S \mid \left(S_{i,g,t}^{\ast}\right) \right] 
= \frac{1}{G} \sum_{(i,g,t):t \geq 2} \mathbb{I}\{S_{i,g,t}^* = 1\} \cdot \mathrm{E}\left[ Y_{i,g,t}(1) - Y_{i,g,t}(0) \mid \left(S_{i,g,t}^{\ast}\right) \right], \\
P^S 
&= \frac{1}{G} \sum_{(i,g,t):t \geq 2} \mathbb{I}\{S_{i,g,t}^* = 1\}.
\end{align*}
Hence, for any fixed \(G\), applying Jensen's inequality  to the convex function \(x \mapsto x^4\) and the elementary inequality \(|a+b|^p \leq 2^{p-1}\left(|a|^p+|b|^p\right)\) for all \(a,b\in \mathbb{R}\) and \(p>1\) yields: 
\begin{align*}
\left( \frac{\widetilde{T}^S}{P^S} \right)^4 
&= \left( 
\frac{ \frac{1}{G} \sum_{(i,g,t):t \geq 2} \mathbb{I}\{S_{i,g,t}^* = 1\} \cdot \left| \mathrm{E}\left[ Y_{i,g,t}(1) - Y_{i,g,t}(0) \mid \left(S_{i,g,t}^{\ast}\right) \right] \right| }{ \frac{1}{G} \sum_{(i,g,t):t \geq 2} \mathbb{I}\{S_{i,g,t}^* = 1\} } 
\right)^4  \\
&\leq 
\left( 
\max_{(i,g,t):t \geq 2,1\leq g\leq G} \left| \mathrm{E}\left[ Y_{i,g,t}(1) - Y_{i,g,t}(0) \mid \left(S_{i,g,t}^{\ast}\right) \right] \right|
\right)^4 \\
&\leq 
\max_{(i,g,t):t \geq 2,1\leq g\leq G} \left( \mathrm{E}\left[ \left| Y_{i,g,t}(1) - Y_{i,g,t}(0) \right|^4 \mid \left(S_{i,g,t}^{\ast}\right) \right] \right) \\
&\leq 
16 \sup_{(i,g,t):t \geq 2,g\in \mathbb{N}\setminus \{0\}} \mathrm{E}\left[ Y_{i,g,t}(d)^4\mid \left(S_{i,g,t}^{\ast}\right) \right]\\
&\leq C_3,
\end{align*}
for some constant \(C_3>0\)  by Assumption \ref{A19}.
Moreover, since the above bound holds uniformly in \(G\), it holds that
\[
\sup_{G\in \mathbb{N}\setminus \{0\}} \mathrm{E}\left[\left( \frac{\widetilde{T}^S}{P^S} \right)^4\right] \leq  C_3.
\]
On the other hand, we observe that the random variables \(\left\{\sum_{(i,t):t \geq 2} \mathbb{I}\{S_{i,g,t}^* = 1\}\right\}\) are independent across \(g\) under Assumption \ref{A15}. Moreover, by Assumption \ref{A19}, for each \(g \in \{1,\dots,G\}\), there exists constant \(C_4 > 0\) such that 
\begin{align*}
0 \leq \sum_{(i,t):t \geq 2} \mathbb{I}\{S_{i,g,t}^* = 1\} \leq \sum_{t \geq 2} N_{g,t} \leq C_4.
\end{align*}
Hence, applying Hoeffding's inequality for the sum of bounded independent random variables yields:
\begin{align*}
\mathrm{Pr}(I_G = 0) 
&= \mathrm{Pr}\left( \left| P^S - \mathrm{E}\left[P^S\right] \right| \geq \varepsilon_G \right) \\
&= \mathrm{Pr}\left( \left| \sum_{g=1}^G \sum_{(i,t): t \geq 2} \mathbb{I}\{S_{i,g,t}^* = 1\} - \mathrm{E}\left[\sum_{g=1}^G \sum_{(i,t):t \geq 2} \mathbb{I}\{S_{i,g,t}^* = 1\} \right] \right| \geq G\varepsilon_G \right) \\
&\leq 2 \exp\left( -\frac{2G^2 \varepsilon_G^2}{\sum_{g=1}^G C_4^2} \right)
= 2 \exp\left( -\frac{2G \varepsilon_G^2}{C_4^2} \right).
\end{align*}
Now, choose \(\varepsilon_G := \sqrt{\frac{\log G}{G}}\).   Then,
\[
\mathrm{Pr}(I_G = 0) \leq 2 G^{-2/ C_4^2},
\]
which implies
\[
\lim_{G\rightarrow \infty}B_2=\lim_{G\rightarrow \infty}\left( \mathrm{E}\left[ \left( \frac{\widetilde{T}^S}{P^S} \right)^4 \right] \cdot \mathrm{Pr}(I_G = 0) \right)^{1/4} \leq \lim_{G\rightarrow \infty}\left( C_3 \cdot 2 G^{-2/ C_4^2} \right)^{1/4} = 0.
\]
Hence, 
\[ 	\sqrt{G}\cdot B=o(1).\]

\vspace{10pt}

Finally, consider the first term \(A\) on the right-hand side of (\ref{EQ:PROP6STEP2}). We decompose:
\begin{align*}
A&=  \mathrm{E}\left[ \left(P^S - \mathrm{E}\left[P^S\right]\right) \cdot \frac{\widetilde{T}^S}{P^S} \cdot I_G \right] \\
&= \underbrace{\mathrm{E}\left[ \left(\widetilde{T}^S - \mathrm{E}\left[\widetilde{T}^S\right]\right) \cdot \left(P^S - \mathrm{E}\left[P^S\right]\right) \cdot \frac{I_G}{P^S} \right]}_{=:A_1} 
+ \underbrace{\mathrm{E}\left[\widetilde{T}^S\right] \cdot \mathrm{E}\left[ \left(P^S - \mathrm{E}\left[P^S\right]\right) \cdot \frac{I_G}{P^S} \right]}_{=:A_2}.
\end{align*}
For the term \(A_1\), we apply Cauchy-Schwarz:
\begin{align*}
\left|A_1\right|&= \left| \mathrm{E}\left[ \left(\widetilde{T}^S - \mathrm{E}\left[\widetilde{T}^S\right]\right) \cdot \left(P^S - \mathrm{E}\left[P^S\right]\right) \cdot \frac{I_G}{P^S} \right] \right| \\
&\leq \left( \mathrm{Var}\left[\widetilde{T}^S\right] \right)^{1/2} \cdot 
\left( \mathrm{E}\left[ \left( \left(P^S - \mathrm{E}\left[P^S\right]\right) \cdot \frac{I_G}{P^S} \right)^2 \right] \right)^{1/2}.
\end{align*}
Observe that on the event \(\left\{I_G=1\right\}\), we have \(\left|P^S-\mathrm{E}\left[P^S\right]\right|<\varepsilon_G\) with   \(\varepsilon_G = \sqrt{\frac{\log G}{G}}=o(1)\). This implies:
\[P^S >\mathrm{E}\left[P^S\right]-\varepsilon_G,\]
where $\mathrm{E}\left[P^S\right]-\varepsilon_G>0$ for sufficiently large $G$. Consequently,
\footnotesize\[
\left(\mathrm{E}\left[ \left( \frac{P^S - \mathrm{E}\left[P^S\right]}{P^S} \cdot I_G \right)^2 \right]\right)^{1/2} \leq \left(\mathrm{E}\left[ \left( \frac{P^S - \mathrm{E}\left[P^S\right]}{P^S} \right)^2 \right]\right)^{1/2} \leq \left(\mathrm{E}\left[ \left(\frac{P^S - \mathrm{E}\left[P^S\right]}{\mathrm{E}\left[P^S\right] - \varepsilon_G}\right)^2 \right]\right)^{1/2}
 \leq \frac{\left(\mathrm{Var}\left[P^S\right]\right)^{1/2}}{\mathrm{E}\left[P^S\right] - \varepsilon_G}.
\]\normalsize
By Assumption \ref{A19}, the quantities \(\mathrm{E}\left[P^S\right]\) ,  \(\sqrt{G} \cdot \left( \mathrm{Var}\left[P^S\right] \right)^{1/2}\)
and \begin{align*}
	\sqrt{G} \cdot \left( \mathrm{Var}\left[\widetilde{T}^S\right] \right)^{1/2}
	&=\sqrt{G} \cdot \left(\mathrm{Var}\left[\mathrm{E}\left[T^S \mid \left(S_{i,g,t}^{\ast}\right) \right]\right]\right)^{1/2}\\
	&= \sqrt{G}\cdot  \left(\mathrm{Var}\left[T^S\right]-\mathrm{E}\left[\mathrm{Var}\left[T^S \mid \left(S_{i,g,t}^{\ast}\right) \right]\right]\right)^{1/2}\\
	&\leq \sqrt{G} \cdot  \left(\mathrm{Var}\left[T^S\right]\right)^{1/2}
\end{align*} converge to finite limits. Thus, for sufficiently large \(G\), 
\begin{align*}
	\sqrt{G}\cdot |A_1|\leq \sqrt{G} \cdot  \left( \mathrm{Var}\left[T^S\right] \right)^{1/2} \cdot \frac{\sqrt{G} \cdot  \left( \mathrm{Var}\left[P^S\right] \right)^{1/2} }{	\sqrt{G}\cdot  \left(\mathrm{E}\left[P^S\right] - \varepsilon_G\right)}\xrightarrow{G\rightarrow \infty}0.
\end{align*}

It remains to examine the term \(A_2\). We rewrite:
\begin{align*}
	A_2=\mathrm{E}\left[\widetilde{T}^S\right] \mathrm{E}\left[\left(P^S-\mathrm{E}\left[P^S\right]\right) \frac{I_G}{P^S}\right]=-\mathrm{E}\left[\widetilde{T}^S\right] \mathrm{E}\left[\left(P^S-\mathrm{E}\left[P^S\right]\right) \frac{1-I_G}{P^S}\right] .
\end{align*}
Recalling what we have derived above:
\[
\mathrm{Pr}(I_G = 0) \leq 2 G^{-2/ C_4^2}
\quad \text{and}\quad 
\left(\mathrm{E}\left[ \left( \frac{P^S - \mathrm{E}\left[P^S\right]}{P^S} \right)^2 \right]\right)^{1/2} \leq \frac{\left(\mathrm{Var}\left[P^S\right]\right)^{1/2}}{\mathrm{E}\left[P^S\right] - \varepsilon_G}
\]
for sufficiently large \(G\). Thus,  applying  Cauchy--Schwarz yields:
\begin{align*}
\left|A_2\right|&=  \left|\mathrm{E}\left[\widetilde{T}^S\right] \cdot \mathrm{E}\left[ \left(P^S - \mathrm{E}\left[P^S\right]\right) \cdot \frac{1-I_G}{P^S} \right]\right|\\
&\leq  \left|\mathrm{E}\left[\widetilde{T}^S\right]\right| \cdot\left( \mathrm{E}\left[ \left( \frac{P^S - \mathrm{E}\left[P^S\right]}{P^S} \right)^2 \right] \right)^{1/2}\cdot \left( \mathrm{E}\left[\left(1-I_G\right)^2 \right] \right)^{1/2}
 \\
&\leq  \left|\mathrm{E}\left[\widetilde{T}^S\right] \right|\cdot \frac{\left(\mathrm{Var}\left[P^S\right]\right)^{1/2}}{\mathrm{E}\left[P^S\right] - \varepsilon_G}\cdot 2 G^{-1/ C_4^2}.
\end{align*}
Moreover, \(\mathrm{E}\left[\widetilde{T}^S\right]=\mathrm{E}\left[T^S\right]\) by definition, and the quantities \(\mathrm{E}\left[P^S\right]\) ,  \(\sqrt{G} \cdot \left( \mathrm{Var}\left[P^S\right] \right)^{1/2}\) and
 \(\mathrm{E}\left[T^S\right]\) converges to a finite limit by Assumption \ref{A19}.  For sufficiently large \(G\), 
\begin{align*}
	\sqrt{G}\cdot |A_2|\leq    \left|\mathrm{E}\left[T^S\right]\right|\cdot \frac{\sqrt{G}\cdot \left(\mathrm{Var}\left[P^S\right]\right)^{1/2}}{\mathrm{E}\left[P^S\right] - \varepsilon_G}\cdot 2 G^{-1/ C_4^2}\xrightarrow{G\rightarrow \infty}0.
\end{align*}
This completes the proof.
\end{proof}
\begin{proof}[Proof of Proposition \ref{PROPTEST}]

(i) We only examine the asymptotic behavior of $T_{1,0}^{s}$ under $H_{0}^{MC}$. The arguments for $T_{1,0}^{m}$ are analogous. To examine the asymptotic behavior of $\sqrt{G}\left(\widehat{\boldsymbol{\tau}}_{1,0}-\mathbf{0}\right)$ under $H_{0}^{MC}$ defined in (\ref{VECHAT}), recall the definition of $\boldsymbol{U}$ and note that for some generic constants $C>0$ it holds that:
\begin{align*}
&\sup_{g,t}N_{g,t}^{4}\mathrm{E}\left[1\{D_{g,t}=1,D_{g,t-1}=0\}(Y_{g,t}-Y_{g,t-l})^4\right]\\
\leq & \sup_{g,t}N_{g,t}^{4} \times \sup_{g,t}\mathrm{E}\left[(Y_{g,t}-Y_{g,t-l})^4\right]\\
\leq & C\sup_{g,t}N_{g,t}^{4} \times \sup_{g,t}\mathrm{E}\left[Y_{g,t}^4\right]<+\infty,
\end{align*}
where the last inequality follows from Assumption \ref{A19}.  Thus, by another application of Assumption \ref{A19} and the Lypuanov's Central Limit Theorem:
\[
\sqrt{G}\left(\mathbf{U}-\mathrm{E}\left[\mathbf{U}\right]\right)\stackrel{d}{\rightarrow} N(\mathbf{0},\boldsymbol{\Sigma}_{\mathbf{U}}),
\]   
where $\boldsymbol{\Sigma}_{\mathbf{U}}$ is defined in (\ref{ASYVAR}). Next, define:
\[
\boldsymbol{f}\left(\begin{array}{c} p_{0,0}(t,l),q_{0,0}(t,l),p_{1,0}(t,l),q_{1,0}(t,l)\\ \ldots \\ p_{0,0}(j,k),q_{0,0}(j,k),p_{1,0}(j,k),q_{1,0}(j,k)\end{array}\right)\equiv \left(\begin{array}{c}\frac{q_{1,0}(t,l)- \frac{p_{1,0}(t,l)}{p_{0,0}(t,l)}q_{0,0}(t,l)}{p_{1,0}(t,l)}\\ \ldots\\\frac{q_{1,0}(j,k)- \frac{p_{1,0}(j,k)}{p_{0,0}(j,k)}q_{0,0}(j,k)}{p_{1,0}(j,k)}\end{array} \right)
\]
for all elements of $\mathcal{C}_{S}$ such that the denominator is strictly positive. Moreover, let
\[
\boldsymbol{f}\left(\mathrm{E}\left[\mathbf{U}\right]\right)=\boldsymbol{f}_{N}\cdot \boldsymbol{f}_{D},
\]
where
\[
\boldsymbol{f}_{N}=\left(\begin{array}{c} f_{N}(t,l) \\\ldots \\ f_{N}(j,k)\end{array}\right)
\]
with elements:
\begin{align*}
f_{N}(t,l)=&\frac{1}{G}\sum_{g\in G}N_{g,t}\left(\mathrm{E}\left[1\{D_{g,t}=1,D_{g,t-1}=0\}(Y_{g,t}-Y_{g,t-l})\right]\right.\\
&\left.-\frac{\mathrm{E}\left[P_{1,0}(t,l)\right]}{\mathrm{E}\left[P_{0,0}(t,l)\right]}\mathrm{E}\left[1\{D_{g,t}=1,D_{g,t-1}=0\}(Y_{g,t}-Y_{g,t-l})\right]\right)
\end{align*}
and
\[
\boldsymbol{f}_{D}=\left(\begin{array}{c} f_{D}(t,l) \\\ldots \\ f_{D}(j,k)\end{array}\right)=\left(\begin{array}{c}\frac{1}{\mathrm{E}\left[P_{1,0}(t,l)\right]}\\\ldots \\\frac{1}{\mathrm{E}\left[P_{1,0}(j,k)\right]}\end{array}\right)^{\prime}.
\]
First, note that under $H_{0}^{MC}$:
\begin{align*}
&\sqrt{G}\left(\widehat{\boldsymbol{\tau}}_{1,0}-\mathbf{0}\right)\\
=&\sqrt{G}\left(\widehat{\boldsymbol{\tau}}_{1,0}-\boldsymbol{\tau}_{1,0}\right)\\
=&\sqrt{G}\left(\widehat{\boldsymbol{\tau}}_{1,0}-\boldsymbol{f}_{N}\cdot \boldsymbol{f}_{D}\right) + \sqrt{G}\left(\boldsymbol{f}_{N}\cdot \boldsymbol{f}_{D}-\boldsymbol{\tau}_{1,0}\right),
\end{align*}
where $\boldsymbol{\tau}_{1,0}$ was defined in (\ref{VEC}). By Assumption \ref{A2A}, we have that $\mathrm{E}\left[\mathbf{U}(t,l)\right]$ converges pointwise for every $t$ and $l$ to $\mathbf{U}^{\infty}(t,l)$. Thus, also $\mathrm{E}\left[\mathbf{U}\right]$ converges to $\mathbf{U}^{\infty}$. Moreover, as $\boldsymbol{f}$ is differentiable in a neighborhood of $\mathbf{U}^{\infty}$, it follows again by the uniform Delta method \citep[e.g., Theorem 3.8][]{VanderVaart1998} that:
\[
\sqrt{G} \left(\widehat{\boldsymbol{\tau}}_{1,0} -\boldsymbol{f}_{N}\cdot \boldsymbol{f}_{D}\right)=\mathbf{J}_{f}	\left(\mathbf{U}^{\infty}\right) \times \sqrt{G}\left(\mathbf{U}-\mathrm{E}\left[\mathbf{U}\right]\right) + o_{\Pr}(1),
\]
where:
\begin{equation}\label{Jacobian}
\mathbf{J}_{f}	\left(\mathbf{U}^{\infty}\right)=\left(\begin{array}{ccccc}\mathbf{J}_{f}	\left(\mathbf{U}^{\infty}(t,l)\right)^{\prime} &\mathbf{0}_{4} &\ldots &\ldots &\mathbf{0}_{4}\\ \mathbf{0}_{4} & \mathbf{J}_{f}	\left(\mathbf{U}^{\infty}(t,m)\right)^{\prime} &\ddots &\ddots &\vdots\\
\vdots &\ddots &\ddots &\ddots &\vdots\\
\vdots &\ddots  &\ddots &  \mathbf{J}_{f}	\left(\mathbf{U}^{\infty}(m,k)\right)^{\prime}& \mathbf{0}_{4}\\
\mathbf{0}_{4}&\ldots &\ldots &\mathbf{0}_{4}&  \mathbf{J}_{f}	\left(\mathbf{U}^{\infty}(j,k)\right)^{\prime}\end{array}\right)
\end{equation}
is a $\vert \mathcal{C}_{S}\vert \times (4\cdot \vert \mathcal{C}_{S}\vert)$ matrix each element $\mathbf{J}_{f}	\left(\mathbf{U}^{\infty}(t,l)\right)$ is given by:
\[
\mathbf{J}_{f}	\left(\mathbf{U}^{\infty}(t,l)\right)=\lim_{G\rightarrow \infty}\left(\begin{array}{c}\frac{\mathrm{E}\left[Q_{0,0}(t,l)\right]}{\mathrm{E}\left[P_{0,0}(t,l)^2\right]}\\-\frac{1}{\mathrm{E}\left[P_{0,0}(t,l)\right]} \\-\frac{\mathrm{E}\left[Q_{1,0}(t,l)\right]}{\mathrm{E}\left[P_{1,0}(t,l)^2\right]}\\ \frac{1}{\mathrm{E}\left[P_{1,0}(t,l)\right]}\end{array}\right)
\]
and $\mathbf{0}_{4}$ denotes a $1\times 4$ vector of zeros. We thus have that:
\[
\sqrt{G} \left(\widehat{\boldsymbol{\tau}}_{1,0} -\boldsymbol{f}\left(\mathrm{E}\left[\mathbf{U}\right]\right)\right)\stackrel{d}{\rightarrow}N\left(\mathbf{0},\mathbf{J}_{f}	\left(\mathbf{U}^{\infty}\right)\boldsymbol{\Sigma}_{\mathbf{U}} \mathbf{J}_{f}	\left(\mathbf{U}^{\infty}\right)^{\prime}\right),
\]
where $\mathbf{\Sigma}_{\mathbf{U}}$ is positive definite and $\mathbf{J}_{f}	\left(\mathbf{U}^{\infty}\right)$ has full rank, so also $\mathbf{J}_{f}	\left(\mathbf{U}^{\infty}\right)\boldsymbol{\Sigma}_{\mathbf{U}} \mathbf{J}_{f}	\left(\mathbf{U}^{\infty}\right)^{\prime}$ is positive definite. In particular, using ``$\infty$'' also in the expectation operator to denote limiting expressions, e.g. $\mathrm{E}_{\infty}\left[Q_{0,0}(t,l)\right]:=\lim_{G\rightarrow \infty}\mathrm{E}\left[Q_{0,0}(t,l)\right]$, for a specific element $(t,l)$ from $\mathcal{C}_{S}$ we have that:
\[
\sqrt{G} \left(\widehat{\tau}_{1,0}(t,l) - \tau_{1,0}(t,l)\right)\stackrel{d}{\rightarrow}N\left(0,\sigma^{2}(t,l)\right),
\]
where:
\begin{equation}\label{EQASYVARELEMENT}
\sigma^{2}(t,l):=\Omega_{1,0}^{(1)}(t,l) + \Omega_{1,0}^{(2)}(t,l)+ 2\Omega_{1,0}^{(12)}(t,l)+\Omega_{0,0}^{(1)}(t,l) + \Omega_{0,0}^{(2)}(t,l)+ 2\Omega_{0,0}^{(12)}(t,l),
\end{equation}
with:
\[
\Omega_{1,0}^{(1)}(t,l)=\left(-\frac{\mathrm{E}_{\infty}\left[Q_{1,0}(t,l)\right]}{\mathrm{E}_{\infty}\left[P_{1,0}(t,l)^2\right]}\right)^2\text{Var}_{\infty}\left[P_{1,0}(t,l)\right],
\]
\[
\Omega_{1,0}^{(2)}(t,l)=\left(\frac{1}{\mathrm{E}_{\infty}\left[P_{1,0}(t,l)^2\right]}\right)^2\text{Var}_{\infty}\left[Q_{1,0}(t,l)\right],
\]
and
\[
\Omega_{1,0}^{(12)}(t,l)=\left(-\frac{\mathrm{E}_{\infty}\left[Q_{1,0}(t,l)\right]}{\mathrm{E}_{\infty}\left[P_{1,0}(t,l)^2\right]\mathrm{E}_{\infty}\left[P_{1,0}(t,l)\right]}\right)\text{Cov}_{\infty}\left[Q_{1,0}(t,l),P_{1,0}(t,l)\right]
\]
(the terms $\Omega_{0,0}^{(1)}(t,l)$, $\Omega_{0,0}^{(2)}(t,l)$, $\Omega_{0,0}^{(12)}(t,l)$ are defined analogously).
Replicating the arguments from the proof of Proposition \ref{PROP7}, we also obtain that: 
\[
\sqrt{G}\left(\boldsymbol{f}_{N}\cdot \boldsymbol{f}_{D}-\boldsymbol{\tau}_{1,0}\right)=\sqrt{G}\left(\boldsymbol{f}_{N}\cdot \boldsymbol{f}_{D}-\mathbf{0}\right)=o_{\Pr}(1).
\]
Next, we show the first order asymptotic validity of the bootstrap moment equalities under the null hypothesis. Specifically, we want to show that, for any $\nu>0$:
\[
\Pr\left(\sup_{\mathbf{x}\in\mathbb{R}^{\mathcal{C}_{S}} }\left\vert \Prb \left(\sqrt{G}(\widehat{\boldsymbol{\tau}}_{1,0}^{b}-\widehat{\boldsymbol{\tau}}_{1,0})\leq \mathbf{x}\right)-\Pr\left(\sqrt{G}(\widehat{\boldsymbol{\tau}}_{1,0}-\boldsymbol{\tau}_{1,0})\leq \mathbf{x}\right)\right\vert >\nu \right)\rightarrow 0,
\]
where the inequality is understood elementwise. We will do so by showing that for each element $(t,l)$ of $\mathcal{C}_{S}$ pointwise:
\begin{equation}\label{BValidityELEMENT}
\Pr\left(\sup_{x\in\mathbb{R} }\left\vert \Prb \left(\sqrt{G}(\widehat{\tau}_{1,0}^{b}(t,l)-\widehat{\tau}_{1,0}(t,l))\leq x\right)-\Pr\left(\sqrt{G}(\widehat{\tau}_{1,0}(t,l)-\tau_{1,0}(t,l))\leq x\right)\right\vert >\nu \right)\rightarrow 0,
\end{equation}
Since the cardinality of $\mathcal{C}_{S}$ is finite, the uniform result will follow by the Cram\'{e}r-Wold device once we establish that that the bootstrap covariances adequately mimic the actual covariances of the moment equalities. 

We start with the convergence in (\ref{BValidityELEMENT}). To show this result, we establish that for each element $(t,l)$ in $\mathcal{C}_{S}$:
\[
\sqrt{G}(\widehat{\tau}_{1,0}^{b}(t,l) -\widehat{\tau}_{1,0}(t,l))\stackrel{d^{b}}{\rightarrow}N\left(0,\sigma^{2}(t,l)\right),
\] 
in probability. Then, since:
\[
\sqrt{G} \left(\widehat{\tau}_{1,0}(t,l) - \tau_{1,0}(t,l)\right)\stackrel{d}{\rightarrow}N\left(0,\sigma^{2}(t,l)\right),
\]
the result will follow from the Triangle inequality and Polya's Theorem given the continuity of the normal distribution. As in the first part, we start by arguing that:
\[
\sqrt{G}\left(\mathbf{U}^{b}(t,l)-\mathrm{E}^{b}\left[\mathbf{U}^{b}(t,l)\right]\right)\stackrel{d^{b}}{\rightarrow} N(\mathbf{0},\boldsymbol{\Sigma}(t,l)),
\]   
  in probabilty, which we achieve by applying Lyapunov's CLT for triangular arrays as the bootstrap observations are row-wise independent. To establish the latter result, first note that as the bootstrap mean equals the sample mean:
\[
\mathrm{E}^{b}\left[Q_{j,0}^{b}(t,l)\right]=Q_{1,0}(t,l)
\]
and
\[
\mathrm{E}^{b}\left[P_{1,0}^{b}(t,l)\right]=P_{1,0}(t,l),
\]
so that:
\[
\mathrm{E}^{b}\left[Q_{j,0}^{b}(t,l)\right]-Q_{j,0}(t,l)=0
\]
as well as 
\[
\mathrm{E}^{b}\left[P_{j,0}^{b}(t,l)\right]-P_{j,0}(t,l)=0
\]
for $j\in\{0,1\}$. Secondly, focusing on $Q_{1,0}^{b}(t,l)$ for illustrative purposes, because the (block) bootstrap observations are i.i.d. across groups, conditional on the sample, we have that:
\begin{align*}
&\text{Var}^{b}\left[ Q_{1,0}^{b}(t,l)\right]\notag\\
=&\frac{1}{G}\sum_{g\in G}\left\{N_{g,t} \mathbb{I}\{ D_{g, t}=1, D_{g, t-1}=0\}\left\{ \left(Y_{g, t-1}-Y_{g, t-l}\right)\right.\right.\\
&\left.\left.- \left(\frac{1}{G}\sum_{g\in G}\mathbb{I}\{D_{g,t}=1,D_{g,t-1}=0\}N_{g,t}\left(Y_{g, t-1}-Y_{g, t-l}\right)\right)\right\} \right\}^2\\
\stackrel{\Pr}{\rightarrow}& \mathrm{Var}_{\infty}\left[Q_{1,0} (t,l)\right],
\end{align*}
and the same argument applies to the other expressions as well. Focusing again on $Q_{1,0}^{b}(t,l)$ for simplicity and denoting $Z_{g}^{b}:=\mathbb{I}\{D^{b}_{g, t}=1, D^{b}_{g, t-1}=0\} N_{g,t}\left(Y_{g, t-1}^{b}-Y_{g, t-l}^{b}\right)$, it remains to analyze the Lyapunov condition:
\[
\frac{1}{\sum_{g\in G} \mathrm{E}^{b}\left[\left\vert Z_{g}^{b}\right\vert^{2} \right]^{1+\delta/2}}\sum_{g\in G} \mathrm{E}^{b}\left[\left\vert Z_{g}^{b}\right\vert^{2+\delta} \right]\stackrel{\mathrm{Pr}}{\rightarrow}0
\]
for some $\delta>0$. That is:
\[
\frac{1}{\left(\sum_{g\in G} \left\vert Z_{g}\right\vert^{2}\right)^{1+\delta/2}}\sum_{g\in G} \left(\left\vert Z_{g}\right\vert^{2+\delta} \right)\stackrel{\mathrm{Pr}}{\rightarrow}0.
\]
Thus, first define $\overline{Z}_{g}=\vert Z_{g} \vert^{2+\delta} - \mathrm{E}\left[\vert Z_{g}^{b} \vert^{2+\delta}\right]$, and note that for some $\epsilon>0$:
\begin{align*}
&\mathrm{Pr}\left[\frac{1}{G}\sum_{g}\left\vert \overline{Z}_{g} \right\vert < \epsilon \right]\\
\leq & \frac{\mathrm{E}\left[\left(\sum_{g}\overline{Z}_{g} \right)^4 \right]}{G^4 \epsilon^4}\\
\leq &\frac{C\left(\sum_{g}\mathrm{E}\left[\overline{Z}_{g}^4 \right] + \left(\sum_{g}\mathrm{E}\left[\overline{Z}_{g}^2 \right]\right)^2\right)}{G^4 \epsilon^4},
\end{align*}
where the last line follows from Rosenthal's inequality for some $C>1$. The RHS of the last line is bounded, however, by the moment conditions in Assumption \ref{A2A} since:
\begin{align*}
&\mathrm{E}\left[\overline{Z}_{g}^4 \right]\\
=&\mathrm{E}\left[\left(\vert Z_{g} \vert^{2+\delta} - \mathrm{E}\left[\vert Z_{g}^{b} \vert^{2+\delta}\right]\right)^4 \right]\\
=&O\left(\mathrm{E}\left[\vert Z_{g} \vert^{8+4\delta}\right]\right)\\
\leq &O\left(\mathrm{E}\left[\sup_{g,t}\vert Y_{g,t} \vert^{8+4\delta}\right]\right)\\
=&O(1).
\end{align*}
This in turn implies that
\[
\sum_{g\in G} \mathrm{E}^{b}\left[\left\vert Z_{g}^{b}\right\vert^{2+\delta} \right]=O_{\mathrm{Pr}}\left(G\right).
\]
Similarly, for the denominator we have that:
\[
\sum_{g\in G} \mathrm{E}^{b}\left[\left\vert Z_{g}^{b}\right\vert^{2} \right]^{1+\delta/2}=O_{\mathrm{Pr}}\left(G^{1-\delta/2}\right),
\]
which implies that Lyapunov's condition is satisfied. It follows that:
\[
\sqrt{G}\left(Q_{1,0}^{b}(t,l)-\mathrm{E}^{b}\left[Q_{1,0}^{b}(t,l)\right]\right)\stackrel{d^{b}}{\rightarrow} N(\mathbf{0},\mathrm{Var}_{\infty}\left[Q_{1,0}(t,l)\right]),
\]   
in probability, and since the same argument can also be made for the other terms that:
\[
\sqrt{G}\left(\mathbf{U}^{b}(t,l)-\mathrm{E}^{b}\left[\mathbf{U}^{b}(t,l)\right]\right)\stackrel{d^{b}}{\rightarrow} N(\mathbf{0},\boldsymbol{\Sigma}(t,l)),
\]   
in probability. Moreover, since $f(t,l)=f_{N}(t,l)\cdot f_{D}(t,l)$ is differentiable, by the uniform Delta method and the same arguments as in the first part of the proof:
\[
\sqrt{G} \left(\widehat{\tau}_{1,0}(t,l) - \tau_{1,0}(t,l)\right)\stackrel{d}{\rightarrow}N\left(0,\sigma^{2}(t,l)\right),
\]
in probability, with $\sigma^{2}(t,l)$ defined in (\ref{EQASYVARELEMENT}).

It remains to analyze the covariance of different elements of the vector $\sqrt{G}\left(\widehat{\boldsymbol{\tau}}^{b} - \widehat{\boldsymbol{\tau}}\right)$. For simplicity, we only examine the terms involving $(1,0)$:
\begin{align*}
&\text{Cov}^{b}\left[\sqrt{G}\left( \frac{Q_{1,0}^{b}(t,l)}{P^{b}_{1,0}(t,l)}-\frac{Q_{1,0}(t,l)}{P_{1,0}(t,l)}\right),\sqrt{G}\left( \frac{Q_{1,0}^{b}(j,k)}{P^{b}_{1,0}(j,k)}-\frac{Q_{1,0}(j,k)}{P_{1,0}(j,k)}\right)\right]\\
=&\mathrm{E}^{b}\left[G\left( \frac{Q_{1,0}^{b}(t,l)}{P^{b}_{1,0}(t,l)}-\frac{Q_{1,0}(t,l)}{P_{1,0}(t,l)}\right)\left( \frac{Q_{1,0}^{b}(j,k)}{P^{b}_{1,0}(j,k)}-\frac{Q_{1,0}(j,k)}{P_{1,0}(j,k)}\right)\right]\\
&-\mathrm{E}^{b}\left[\sqrt{G}\left( \frac{Q_{1,0}^{b}(t,l)}{P^{b}_{1,0}(t,l)}-\frac{Q_{1,0}(t,l)}{P_{1,0}(t,l)}\right)\right]\times\mathrm{E}^{b}\left[\sqrt{G}\left( \frac{Q_{1,0}^{b}(j,k)}{P^{b}_{1,0}(j,k)}-\frac{Q_{1,0}(j,k)}{P_{1,0}(j,k)}\right)\right]\\
=&\mathrm{E}^{b}\left[G\left( \frac{Q_{1,0}^{b}(t,l)}{P^{b}_{1,0}(t,l)}-\frac{Q_{1,0}(t,l)}{P_{1,0}(t,l)}\right)\left( \frac{Q_{1,0}^{b}(j,k)}{P^{b}_{1,0}(j,k)}-\frac{Q_{1,0}(j,k)}{P_{1,0}(j,k)}\right)\right],
\end{align*}
where the equality follows since, conditional on the sample, the bootstrap mean of the recentered terms is zero. Straightforward algebra also shows that the above term can be rewritten as:
\begin{align*}
&\mathrm{E}^{b}\left[G\left( \frac{Q_{1,0}^{b}(t,l)-Q_{1,0}(t,l)}{P_{1,0}(t,l)}\right)\left( \frac{Q_{1,0}^{b}(j,k)-Q_{1,0}(j,k)}{P_{1,0}(j,k)}\right)\right]\\
+&\mathrm{E}^{b}\left[G\left( \frac{1}{P^{b}_{1,0}(t,l)}-\frac{1}{P_{1,0}(t,l)}\right)\left( \frac{1}{P^{b}_{1,0}(j,k)}-\frac{1}{P_{1,0}(j,k)}\right)Q_{1,0}^{b}(t,l)Q_{1,0}^{b}(j,k)\right]\\
=&I_{1,0}^{b}+II_{1,0}^{b}
\end{align*}
We will only examine $I_{1,0}^{b}$ in what follows, the term $II_{1,0}^{b}$ follows by similar arguments. Moreover, let
\[
m(t,l)\equiv \frac{1}{G}\sum_{g\in G}\mathbb{I}\{ D_{g, t}=1, D_{g, t-1}=0\} N_{g,t} \left(Y_{g, t-1}-Y_{g, t-l}\right),
\]
noting that $m(t,l)=\mathrm{E}^{b}\left[\mathbb{I}\{ D^{b}_{g, t}=1, D^{b}_{g, t-1}=0\} N_{g,t} \left(Y_{g, t-1}^{b}-Y_{g, t-l}^{b}\right)\right]$. Thus, observe that:
\begin{align*}
&\mathrm{E}^{b}\left[G\left( \frac{Q_{1,0}^{b}(t,l)-Q_{1,0}(t,l)}{P_{1,0}(t,l)}\right)\left( \frac{Q_{1,0}^{b}(j,k)-Q_{1,0}(j,k)}{P_{1,0}(j,k)}\right)\right]\\
=& \mathrm{E}^{b}\left[\frac{1}{P_{1,0}(t,l)\sqrt{G}}\sum_{g\in G}\left(\mathbb{I}\{ D^{b}_{g, t}=1, D^{b}_{g, t-1}=0\} N_{g,t} \left(Y_{g, t-1}^{b}-Y_{g, t-l}^{b}\right) - m(t,l)\right)\right.\\
&\left.\times\frac{1}{P_{1,0}(j,k)\sqrt{G}}\sum_{g\in G}\left(\mathbb{I}\{ D^{b}_{g, j}=1, D^{b}_{g, j-1}=0\} N_{g,j} \left(Y_{g, j-1}^{b}-Y_{g, j-k}^{b}\right)-m(j,k)\right)\right]\\
=&\mathrm{E}^{b}\left[\frac{1}{P_{1,0}(t,l)P_{1,0}(j,k) G}\sum_{g\in G}\left(\mathbb{I}\{ D^{b}_{g, t}=1, D^{b}_{g, t-1}=0\} N_{g,t} \left(Y_{g, t-1}^{b}-Y_{g, t-l}^{b}\right) - m(t,l)\right)\right.\\
&\left.\times\left(\mathbb{I}\{ D^{b}_{g, j}=1, D^{b}_{g, j-1}=0\} N_{g,j} \left(Y_{g, j-1}^{b}-Y_{g, j-k}^{b}\right)-m(j,k)\right)\right]\\
&+\mathrm{E}^{b}\left[\frac{1}{P_{1,0}(t,l)P_{1,0}(j,k) G}\sum_{g\in G}\sum_{h\neq g}\left(\mathbb{I}\{ D^{b}_{g, t}=1, D^{b}_{g, t-1}=0\} N_{g,t} \left(Y_{g, t-1}^{b}-Y_{g, t-l}^{b}\right) - m(t,l)\right)\right.\\
&\left.\times\left(\mathbb{I}\{ D^{b}_{h, j}=1, D^{b}_{h, j-1}=0\} N_{h,j} \left(Y_{h, j-1}^{b}-Y_{h, j-k}^{b}\right)-m(j,k)\right)\right]\\
=&\frac{1}{P_{1,0}(t,l)P_{1,0}(j,k) G}\sum_{g\in G}\mathbb{I}\{ D_{g, t}=1, D_{g, t-1}=0\}\mathbb{I}\{ D_{g, j}=1, D_{g, j-1}=0\} N_{g,t} N_{g,j}\\
&\times \left(Y_{g, t-1}-Y_{g, t-l}\right) \left(Y_{g, j-1}-Y_{g, j-k}\right)\\
&-\left(\frac{1}{P_{1,0}(t,l) G}\sum_{g\in G}\mathbb{I}\{ D_{g, t}=1, D_{g, t-1}=0\}N_{g,t} \left(Y_{g, t-1}-Y_{g, t-l}\right)\right)\\
&\times\left(\frac{1}{P_{1,0}(j,k) G}\sum_{g\in G}\mathbb{I}\{ D_{g, j}=1, D_{g, j-1}=0\} N_{g,j} \left(Y_{g, j-1}-Y_{g, j-k}\right)\right),
\end{align*}
where the last equality follows since the blocks are, conditional on the sample, independent. We thus have that the bootstrap covariance is equal to the sample covariance, which converges in probability to $\lim_{G\rightarrow \infty}\mathrm{Cov}\left[\sqrt{G}(\widehat{\tau}_{1,0}(t,l)-\tau_{1,0}(t,l)),\sqrt{G}(\widehat{\tau}_{1,0}(j,k)-\tau_{1,0}(j,k))\right]$.  Uniform convergence  now  follows  by  an  application  of  the  Cram\'{e}r-Wold  device,  see  for  instance \citet{White2000}. Proposition \ref{PROPTEST} (i) then follows from continuous mapping.

(ii) Under the alternative $H_{A}^{MC}$:
\begin{align*}
&\sqrt{G}\left(\widehat{\boldsymbol{\tau}}_{1,0}-\mathbf{0}\right)\\
=&\sqrt{G}\left(\widehat{\boldsymbol{\tau}}_{1,0}-\boldsymbol{\tau}_{1,0}\right) +\sqrt{G}\left(\boldsymbol{\tau}_{1,0}-\mathbf{0}\right)
\end{align*}
For the first term on the RHS, it follows by the arguments of part (i) that:
\[
\sqrt{G}\left(\widehat{\boldsymbol{\tau}}_{1,0}-\boldsymbol{\tau}_{1,0}\right)\stackrel{d}{\rightarrow}N\left(\mathbf{0},\mathbf{J}_{f}	\left(\mathbf{U}^{\infty}\right)\boldsymbol{\Sigma}_{\mathbf{U}} \mathbf{J}_{f}	\left(\mathbf{U}^{\infty}\right)^{\prime}\right).
\]
On the other hand, for the second term, we have that:
\[
\left\vert \sqrt{G}\left(\boldsymbol{\tau}_{1,0}-\mathbf{0}\right)\right\vert \rightarrow + \infty.
\]
Since the bootstrap mimics the distribution under $H_{0}^{MC}$, the result follows.
\end{proof}

\bibliographystyle{chicago}
\bibliography{AGL}

\end{document}